\newif\ifprocs
\procsfalse
\documentclass[11pt]{article}
\usepackage{wrapfig}
\usepackage{color}
\usepackage{fullpage}
\usepackage{amssymb}
\usepackage{amsmath}
\usepackage{theorem}
\usepackage{dsfont}
\usepackage{xspace}
\usepackage[ruled, vlined, linesnumbered]{algorithm2e}
\usepackage{color}

\usepackage{titlesec}

\usepackage{ifpdf}
\ifpdf    % we are running pdflatex
\usepackage{hyperref}
\else    % we are running latex
\usepackage[hypertex]{hyperref}
\fi

\titleformat{\paragraph}
{\normalfont\normalsize\bfseries}{\theparagraph}{1em}{}
\titlespacing*{\paragraph}
{0pt}{3.25ex plus 1ex minus .2ex}{1.5ex plus .2ex}

\newcommand{\veps}{\varepsilon}
\newcommand{\eps}{\varepsilon}
\newcommand{\E}{\mathbf{E}}
\renewcommand{\Pr}{\mathbf{Pr}}
\newcommand{\abs}[1]{\left| #1 \right|}
\newcommand{\norm}[1]{\left\lVert #1 \right\rVert}
\newcommand{\var}[1]{\mathbf{Var}\left[#1\right]}
\newcommand{\vol}{\text{vol}}

\newcommand{\sk}[1]{\text{sk}(#1)}

\newcommand{\windeg}[1]{\delta_{#1}^{\tt{in}}(\vec{G})}
\newcommand{\woutdeg}[1]{\delta_{#1}^{\tt{out}}(\vec{G})}
\newcommand{\uwindeg}[1]{d_{#1}^{\tt{in}}(\vec{G})}
\newcommand{\uwoutdeg}[1]{d_{#1}^{\tt{out}}(\vec{G})}

\newcommand{\windegP}[1]{\delta_{#1}^{\tt{in}}(\vec{P})}
\newcommand{\woutdegP}[1]{\delta_{#1}^{\tt{out}}(\vec{P})}

\newcommand{\uwoutdegP}[1]{d_{#1}^{\tt{out}}(\vec{P})}
\renewcommand{\L}{\mathcal{L}}
\renewcommand{\S}{\mathcal{S}}

\newcommand{\IND}{Indexing}

\newcommand{\ignore}[1]{}

\newcommand{\zo}{\{0,1\}}

\newcommand{\R}{{\mathbb{R}}}

\newcommand{\tO}{\tilde{O}}
\DeclareMathOperator{\wt}{wt}
\DeclareMathOperator{\ske}{{\bf sk}}
\DeclareMathOperator{\est}{{\bf est}}
\DeclareMathOperator{\IC}{IC}
\DeclareMathOperator{\polylog}{polylog}
\newcommand{\calP}{{\cal P}}
\newcommand{\Lhigh}{L_{\mathrm{high}}}
\newcommand{\indic}{\mathds{1}}

\newcommand{\EX}{\mathbf{E}}

\providecommand{\minn}[1]{\min\{{#1}\}}

\providecommand{\ceil}[1]{\lceil #1 \rceil}
\providecommand{\floor}[1]{\lfloor #1 \rfloor}
\providecommand{\card}[1]{\lvert#1\rvert}
\providecommand{\aset}[1]{\{ #1 \}}
\newcommand{\whp}{with high probability\xspace}
\newcommand{\half}{\tfrac12}
\def\compactify{\itemsep=0pt \topsep=0pt \partopsep=0pt \parsep=0pt}

\newcommand{\aanote}[1]{}
\newcommand{\rnote}[1]{}
\newtheorem{theorem}{Theorem}[section]
\newtheorem{lemma}[theorem]{Lemma}

\newtheorem{claim}[theorem]{Claim}
\newtheorem{observation}[theorem]{Observation}
\newtheorem{corollary}[theorem]{Corollary}

\newtheorem{definition}[theorem]{Definition}

\newcommand{\qinhides}[1]{}

\newenvironment{proof}{\trivlist\item[]\emph{Proof:}}%
{\unskip\nobreak\hskip 1em plus 1fil\nobreak$\Box$
\parfillskip=0pt%
\endtrivlist}

\begin{document}

\title{On Sketching Quadratic Forms}

\author{Alexandr Andoni\thanks{Columbia University. Work done in part while the author was at Microsoft Research Silicon Valley. Email: andoni@cs.columbia.edu} 
\and  Jiecao Chen\thanks{Indiana University Bloomington. Work supported in part by NSF CCF-1525024, and IU's Office of the Vice Provost for Research through the Faculty Research Support Program. Email: jiecchen@indiana.edu} 
\and Robert Krauthgamer\thanks{Weizmann Institute of Science, Israel. Work supported in part by a US-Israel BSF grant \#2010418, an Israel
Science Foundation grant \#897/13, and by the Citi Foundation. Part of the work was done at Microsoft Research Silicon Valley. Email: robert.krauthgamer@weizmann.ac.il} 
\and Bo Qin\thanks{Hong Kong University of Science and Technology. Work of this author partially supported by Hong Kong RGC GRF grant 16208415. Email: bqin@cse.ust.hk} 
\and David P. Woodruff\thanks{IBM Almaden Research. Work supported in part by the XDATA program of the Defense Advanced Research Projects Agency (DARPA), administered through Air Force Research Laboratory contract FA8750-12-C-0323.
Email: dpwoodru@us.ibm.com} 
\and Qin Zhang\thanks{Indiana University Bloomington. Work supported in part by NSF CCF-1525024, and IU's Office of the Vice Provost for Research through the Faculty Research Support Program. Email: qzhangcs@indiana.edu}
}

%\begin{titlepage}
%\date{}

\maketitle

 \begin{abstract}
 We undertake a systematic study of sketching a quadratic form: given
an $n \times n$ matrix $A$, create a succinct sketch $\sk{A}$ which
can produce (without further access to $A$) a multiplicative
$(1+\eps)$-approximation to $x^T A x$ for any desired query $x \in
\mathbb{R}^n$. 
While a general matrix does not admit non-trivial sketches, 
positive semi-definite (PSD) matrices admit sketches of size 
$\Theta(\eps^{-2} n)$, via the Johnson-Lindenstrauss lemma, 
achieving the ``for each'' guarantee, namely, for each query $x$, with a constant
probability the sketch succeeds. (For the stronger ``for all''
guarantee, where the sketch succeeds for all $x$'s simultaneously,
again there are no non-trivial sketches.)

We design significantly better sketches for the important subclass of graph Laplacian matrices, which we also extend to symmetric diagonally dominant matrices. A sequence of work culminating in that of Batson, Spielman, and Srivastava (SIAM Review, 2014), shows that by choosing and reweighting $O(\eps^{-2} n)$  edges in a graph, one achieves the ``for all" guarantee. Our main results advance this front.

\begin{enumerate}
\item
For the ``for all'' guarantee, we prove that Batson et al.'s bound
is optimal even when we restrict to ``cut queries''
$x\in \{0,1\}^n$. Specifically, an arbitrary sketch that can
$(1+\eps)$-estimate the weight of {\em all} cuts $(S,\bar S)$ in an
$n$-vertex graph must be of size $\Omega(\eps^{-2} n)$ bits.
Furthermore, if the sketch is a cut-sparsifier (i.e., itself a
weighted graph and the estimate is the weight of the corresponding cut
in this graph), then the sketch must have $\Omega(\eps^{-2} n)$
edges. 
%Our lower bound even applies to unweighted graphs
%(Laplacians). 

In contrast, previous lower bounds showed the bound only for {\em
  spectral-sparsifiers}.
%No lower bound was previously known on the size of cut sparsifiers.
%despite the long sequence of work on graph sparsification.

\item 
For the ``for each'' guarantee, we design a sketch of size 
$\tilde O(\eps^{-1} n)$ bits for ``cut queries'' $x\in \{0,1\}^n$. 
We apply this
sketch to design an algorithm for the distributed minimum cut problem. We
prove a nearly-matching lower bound of $\Omega(\eps^{-1} n)$ bits.
For general queries $x \in \mathbb{R}^n$, we construct sketches of
size $\tilde{O}(\eps^{-1.6} n)$ bits.  
\end{enumerate}

Our results provide the first separation between the sketch size needed
for the ``for all'' and ``for each'' guarantees for Laplacian
matrices.

%%% Local Variables: 
%%% mode: latex
%%% TeX-master: "paper-short"
%%% End: 

 \end{abstract}

% \end{titlepage}

\section{Introduction} \label{sec:intro}

Sketching emerges as a fundamental building block 
used in numerous algorithmic contexts to reduce memory, runtime, or communication requirements.
Here we focus on sketching \emph{quadratic forms}, defined as follows:
Given a matrix $A\in \R^{n\times n}$, compute a sketch of it, $\sk{A}$,
which suffices to estimate the quadratic form $x^T A x$ 
for every query vector $x\in \R^n$. 
Typically, we aim at $(1+\eps)$-approximation, 
i.e., the estimate is in the range $(1\pm \eps) x^T A x$, and
sketches that are randomized. The randomization guarantee comes in two flavors. The first one requires that the sketch $\sk{A}$ succeeds 
(produces a $(1+\eps)$-approximation) on all queries $x$ simultaneously. 
The second one requires that for every fixed query $x$, 
the sketch succeeds with high probability. 
The former is termed the ``for all'' guarantee and the latter the ``for each'' guarantee, following the prevalent terminology in compressive sensing. 
The main goal is then to design a sketch $\sk{A}$ of small size.
%Ideally, the space (storage) requirement of $\sk{A}$ should be close to $O(n)$, 
%which could be significantly more economical than the original matrix $A$. 

Sketching quadratic forms is a basic task with many applications. In
fact, the
definition from above abstracts several specific concepts studied before.  One important example is the sparsification of a graph
$G$, where we take the matrix $A$ to be the Laplacian
of $G$ and restrict the sketch to be of a specific form, namely, a
Laplacian of a sparse subgraph $G'$.  Then a cut-sparsifier
corresponds to the setting of query vectors $x \in \{0,1\}^n$, in which
case $x^TAx$ describes the weight of the corresponding cut in $G$.
Also, a spectral-sparsifier corresponds to query vectors $x\in \mathbb{R}^n$
in which case $x^TAx$ is a Laplacian Rayleigh quotient.
Cut queries to a graph have been studied in the context of privacy in
databases \cite{gru12,jt12,bbds13,u13,u14} where, for example,
vertices represent users and edges represent email correspondence
between users, and email correspondences between groups of users are
of prime interest.  These papers study also directional covariance
queries on a matrix, which correspond to evaluating the quadratic form
of a positive semidefinite~(PSD) matrix, as well as evaluating the
quadratic form of a low-rank matrix, which could correspond to, e.g.,
a user-movie rating matrix. Finally, sketching quadratic forms has
appeared and has been studied in other contexts 
\cite{AHK05, agm12, agm12b, KLMMS14, m14}.

Quadratic form computations also arise in numerical linear algebra.
Consider the least squares regression problem of minimizing $\norm{By-c}_2^2$ 
for an input matrix $B$ and vector $c$.
Writing the input as an adjoined matrix $M=[B, c]$ and denoting $x = (y, -1)$,
the objective is just $\norm{By-c}_2^2 = \norm{Mx}_2^2 = x^T M^T M x$, 
and thus regression queries 
can be modeled by a quadratic form over the PSD matrix $A=M^T M$.
Indeed, for a concrete example where a small-space sketch $\sk{A}$ 
leads to memory savings (in the data-stream model)
in regression problems, see \cite{cw09}.

To simplify the exposition,
%streamline the \emph{exposition} of our results, 
let us assume that the matrix $A$ is of size $n\times n$
and its entries are integers bounded by a polynomial in $n$,
and fix the success probability to be $90\%$.
When we consider a graph $G$, we let $n$ denote its number of vertices,
with edge-weights that are positive integers bounded by a polynomial in $n$.
We use $\tilde{O}(f)$ to denote $f\cdot (\log f)^{O(1)}$,
which suppresses the distinction between counting bits and machine words.
%Precise details appear in the actual statements of the relevant theorems.

The general quadratic forms, i.e., when the square matrix $A$ is
arbitrary, require a sketch of size $\tilde{\Omega}(n^2)$ bits, even
in the ``for each'' model 
(see Appendix \ref{app:general-each}).

Hence we restrict our attention to the class of PSD matrices $A$, 
and its subclasses like graph Laplacians, which occur in many applications.
We provide tight or near-tight bounds for these classes, 
as detailed in Table~\ref{tab:results}.
Overall, our results show that the specific class of matrices 
as well as the model (``for each'' vs.\ ``for all'' guarantee)
can have a dramatic effect on the sketch size, namely, 
quadratic vs.\ linear dependence on $n$ or on $\eps$.

\subsection{Our Contributions}
\label{sec:results} 

We start by characterizing the sketching complexity for general PSD
matrices $A$, in both the ``for all'' and ``for each'' models.  First,
we show that, for the ``for all'' model, sketching an arbitrary
PSD matrix $A$ requires $\Omega(n^2)$ bits
(Theorem~\ref{thm:all-exact}); i.e., storing the entire matrix is
essentially optimal.
In contrast, for the ``for each'' model, we show 
that the Johnson-Lindenstrauss lemma immediately yields 
a sketch of size $O(n \eps^{-2} \log n)$ bits and this is tight up to
the logarithmic factor (see Section~\ref{sec:psd-each}).
% $\Omega(n \eps^{-2})$.
% (Theorem~\ref{thm:all-approx})
We conclude that the bounds for the two models are quite different:
quadratic vs.\ linear in $n$. 

Surprisingly, one can obtain significantly smaller sketches when $A$
is the Laplacian of a graph $G$, a subclass of PSD matrices that
occurs in many applications.  Specifically, we refer to a celebrated
result of Batson, Spielman, and Srivastava \cite{BSS14}, which is the
culmination of a rich line of research on graph sparsification
\cite{BK96,ST04a,ST11,SS11,FHHP11,KP12}.  They show that every graph
Laplacian $A$ 
%that is the Laplacian of a
%graph $G$ 
admits a sketch in the ``for all'' model 
%(formed by
%retaining and reweighting a subset of $O(n \eps^{-2})$ edges of $G$),
whose size is $O(n \eps^{-2} \log n)$ bits. This stands in contrast to
the $\tilde{\Omega}(n^2)$ lower bound for general PSD matrices.  Their
sketch has a particular structure: it is itself a
graph, consisting of a reweighted subset of edges in $G$ and works in
the ``for all'' model. 
Batson et al.~\cite{BSS14} also prove a lower bound for the case
of {\em spectral sparsification} (for cut sparsifiers, the bound
remained open).

The natural question is whether there are qualitatively better
sketches we can construct by relaxing the guarantees or considering
more specific cases. Indeed, we investigate this research direction by
pursuing the following concrete questions:
%we can obtain new sketches
%for 1) more general sketches, 2) 
\begin{enumerate} \compactify
\renewcommand{\theenumi}{Q\arabic{enumi}}
\item \label{it:q1}
Can we improve the ``for all'' upper bound $O(n\eps^{-2})$ 
by using an arbitrary data structure?
\item \label{it:q2}
Can we improve the bound by restricting attention to {\em cut} queries? 
Specifically, can the optimal size of {\it cut-sparsifiers} 
be smaller than that of spectral-sparsifier?
\item \label{it:q3} 
Can we improve the ``for each'' bound beyond the $\tilde O(n\eps^{-2})$
bound that follows from general PSD matrices result
(and also from the ``for all'' model via \cite{BSS14})?
%% \item \label{it:q4} 
%% What about slightly larger subclasses of PSD matrices, such as symmetric diagonally dominant
%% (SDD) matrices?
\end{enumerate}
We make progress on all of the above questions, 
often providing (near) tight results.% that resolve the questions.

%% First, regarding \ref{it:q4}, 
%% we show in Appendix \ref{app:sdd-each} by simple reductions,
%% that if one has a sketching procedure for Laplacian matrices
%% (in either the ``for each'' or ``for all'' model), 
%% then by slightly changing the procedure, which results in only a constant
%% factor blowup in sketch size, one obtains a procedure for SDD matrices. 
%% Thus, in the rest of the paper we focus on the case where $A$ is a Laplacian. 

In all of these questions,
%In the remaining questions, 
the main quantitative focus is the dependence 
on the accuracy parameter $\eps$.
%It is only natural to ask what is the best dependence on $\eps$. 
We note that improving the dependence on $\eps$ is important for a
variety of reasons.  From a theoretical angle, a quadratic dependence
is common for estimates with two-sided error, and hence sub-quadratic
dependence elucidates new interesting phenomena.  From a practical
angle, we can set $\eps$ to be the smallest value for which the sketch
still fits in memory (i.e., we can get {\em better} estimates with the
same memory).  In general, quadratic dependence might be
prohibitive for large-scale matrices: if, say, $\eps$ is $1\%$ then
$1/\eps^2 = 10000$. 

We answer \ref{it:q1} negatively by showing that every sketch that
satisfies the ``for all'' guarantee requires $\Omega(n \eps^{-2})$
bits of space, even if the sketch is an arbitrary data structures (see
Section~\ref{sec:sdd-all}).  This matches the upper bound of
\cite{BSS14} (up to a logarithmic factor, which stems from the
difference between counting words and bits).

Our answer to \ref{it:q1} essentially answers \ref{it:q2} as well: our
lower bound actually holds even if we only consider cut queries $x\in
\{0,1\}^n$.  Indeed, an immediate consequence of the $\Omega(n
\eps^{-2})$ bits lower bound is that a cut-sparsifier $G'$ must have
$\Omega(n \eps^{-2}/\log n)$ edges.
%(just consider the contrapositive).
We strengthen this further and obtain a tight lower bound of 
$\Omega(n \eps^{-2})$ edges (even in the case when the cut-sparsifier
$G'$ is a not necessarily a subgraph of $G$). 
Such an edge lower bound was not known before.
% even under the additional assumptions that $G'$ is a subgraph of $G$ and its edge weights are uniform.
The previous lower bound for a cut-sparsifier $G'$, 
due to Alon~\cite{Alon97-expansion}, uses two additional requirements ---
that the sparsifier $G'$ has \emph{regular degrees} and
\emph{uniform edge weights} ---
to reach the same conclusion that $G'$ has $\Omega(n/\eps^2)$ edges. Put 
differently, Alon's lower bound is \emph{quantitatively} optimal
--- it concludes the tight lower bound of $\Omega(n/\eps^2)$ edges --- 
but it is unsatisfactory
\emph{qualitatively}, as it does not cover a cut-sparsifier $G'$ that
has edge weights or has non-regular degrees, which may potentially
lead to a smaller sparsifier. Similarly, the results of \cite{Nil91,
  BSS14} apply to spectral-sparsification, which is a harder problem
than cut-sparsification.  Our result subsumes all of these bounds, and
for cut sparsifiers it is in fact the first lower bound under no assumption. 
Our lower bound holds even for input graphs $G$ that are unweighted.

On the upside, we answer \ref{it:q3} positively by showing how
to achieve the ``for each'' guarantee using $n \eps^{-1} \polylog(n)$
bits of space (see Section~\ref{sec:cut-each}).  This bound can be
substantially smaller than in the ``for all'' model when $\eps$ is
small: e.g., when $\eps = 1/\sqrt{n}$ we obtain size
$n^{3/2}\polylog(n)$ instead of the $O(n^2)$ needed in the ``for all''
model.  We also show that $\Omega(n \eps^{-1})$ bits of space is
necessary for the ``for each'' guarantee (Theorem~\ref{thm:nEpsLB}).

We then give an application for the ``for each'' sketch to
showcase that it is useful algorithmically despite having a guarantee
that is is weaker than that of a ``for all'' cut-sparsifier.  In
particular, we show how to $(1+\eps)$-approximate the global minimum
cut of a graph whose edges are distributed across multiple servers (see
Section~\ref{sec:min-cut}).

Finally, we consider a ``for each'' sketch of a Laplacian matrix 
under arbitrary query vectors $x \in \mathbb{R}^n$, 
which we refer to as {\em spectral queries} on the graph $G$.
Such spectral queries give more flexibility than cut queries. 
For example, if the graph corresponds to a physical system, 
e.g., the edges correspond to electrical resistors, 
then spectral queries can evaluate the total heat dissipation of the system 
for a given set of potentials on the vertices. 
Also, a spectral query $x$ that is a permutation of $\{1, 2, \ldots, n\}$ gives the average squared distortion of a line embedding of $G$.  
We design in Section~\ref{sec:spectral-each} a sketch for spectral queries
that uses $n \eps^{-1.6} \polylog(n)$ bits of space. These upper
bounds also apply to the symmetric diagonally-dominant (SDD) matrices.
%In Section~\ref{sec:sdd-each} we further show that the upper bounds for Laplacian matrices also hold for general SDD matrices. 

Our results and previous bounds are summarized in
Table~\ref{tab:results}.

%Section numbers are corresponding to the full version of the paper.

\begin{table*}[t]
  \centering
  \begin{tabular}{|l|ll|ll|}
	\hline
	& \multicolumn{2}{c|}{``for all'' model}  & \multicolumn{2}{c|}{``for each'' model}\\
	Matrix family & upper bound & lower bound & upper bound  & lower bound \\
	\hline
	\hline
	General
	 & $\tilde{O}(n^2)$ &  $\Omega(n^2)$  &   $\tilde{O}(n^2)$ &
        $ {\Omega}(n^2)$ App.~\ref{app:general-each}
        \\
	\hline
	PSD
	 & $\tilde{O}(n^2)$ &  $ \Omega(n^2)$
         Sec.~\ref{sec:psd-all} 
         &  $ \tilde{O}(n \epsilon^{-2})$
         Sec.~\ref{sec:psd-each} 
        &  $ \Omega(n \eps^{-2})$ 
        Sec.~\ref{sec:psd-each} 
        \\
	\hline
	%% SDD  &  $\tilde{O}(n \eps^{-2})$ App.~\ref{app:sdd-each} &  $\Omega(n \eps^{-2})$ & $\tilde{O}(n \eps^{-1.6})$ App.~\ref{app:sdd-each} & $\Omega(n \eps^{-1})$  \\
	%% \hline
	Laplacian, SDD &  $\tilde{O}(n \eps^{-2})$ \cite{BSS14} &
        $\Omega(n \eps^{-2})$ \cite{BSS14}  & $ \tilde{O}(n
        \eps^{-1.6})$ Sec.~\ref{sec:spectral-each} & $\Omega(n \eps^{-1})$ Sec.~\ref{sec:cut-each} \\
	edge-count: & $O(n \eps^{-2})$ \cite{BSS14} &  $\Omega(n \eps^{-2})$ \cite{BSS14} & &\\
	\hline
	Laplacian+cut queries & $\tilde{O}(n \eps^{-2})$ \cite{BSS14}
        &  $\Omega(n \eps^{-2})$  Sec.~\ref{sec:sdd-all}  &
        $\tilde{O}(n \eps^{-1})$ Sec.~\ref{sec:cut-each} &  $\Omega(n \eps^{-1})$ Sec.~\ref{sec:cut-each}\\
	edge-count: & $O(n \eps^{-2})$ \cite{BSS14} &  $\Omega(n \eps^{-2})$ Sec.~\ref{sec:sdd-all}  & &\\
	\hline
  \end{tabular}
  \caption{Bounds for sketching quadratic forms,
expressed in bits, except when counting edges. 
%% Our results are in bold.
%PSD stands for positive semidefinite matrix, and SDD stands for symmetric diagonally dominant matrix. Cut means the query vector $x \in \{0,1\}^n$, and ``spectral'' means the query vector $x \in \mathbb{R}^n$.
}
  \label{tab:results}
\end{table*}

\subsection{Highlights of Our Techniques}
\label{sec:technique} 
In this section we give technical overviews for our three main
results: (1) the lower bound for cut queries on Laplacian matrices
(answering Q1 and Q2);
(2) the upper bound for cut queries on Laplacian matrices; and (3) the
upper bound for spectral queries on Laplacian matrices (answering Q3).  We always use
$G$ to denote the corresponding graph of the considered Laplacian
matrix.

%A query $x^T L x$ for $x \in \{0,1\}^n$ is equivalent to the weight of the corresponding cut  in $G$.

\subsubsection{Lower Bound for Sketching Laplacian Matrices with Cut Queries, ``For All'' Model}

We first prove our $\Omega(n \eps^{-2})$-bit lower bound using communication complexity for arbitrary data structures. We then show how to obtain an $\Omega(n\eps^{-2})$
edge lower bound for cut sparsifiers by encoding a sparsifier in a careful way
so that if it had $o(n/\eps^2)$ edges, it would violate an $\Omega(n \eps^{-2})$ bit
lower bound in the communication problem. 

 For the $\Omega(n \eps^{-2})$ bit lower bound, the natural
thing to do would be to give Alice a graph $G$, and Bob a cut
$S$. Alice produces a sketch of $G$ and sends it to Bob, who must
approximate the capacity of $S$. The communication cost of this
problem lower bounds the sketch size. However, as we just saw, Alice
has an upper bound with only $\tO(n \eps^{-1})$ bits of communication. We
thus need for Bob to solve a much harder problem which uses the fact
that Alice's sketch preserves all cuts.

We let $G$ be a disjoint union of $\eps^2 n/2$ graphs $G_i$, where
each $G_i$ is a bipartite graph with $\frac{1}{\eps^2}$ vertices in
each part. Each vertex in the left part is independently connected to
a random subset of half the vertices in the right part. Bob's problem
is now, given a vertex $v$ in the left part of one of the $G_i$, as
well as a subset $T$ of half of the vertices in the right part of that $G_i$,
decide if $|N(v) \cap T| > \frac{1}{4\eps^2} + \frac{c}{\eps}$ ($N(v)$ is the set of neighboring vertices of $v$), or if
$|N(v) \cap T| < \frac{1}{4\eps^2} - \frac{c}{\eps}$, for a small
constant $c > 0$. Most vertices $v$ will satisfy one of these
conditions, by anti-concentration of the binomial distribution. Note
that this problem is not a cut query problem, and so {\em a priori} it is
not clear how Bob can use Alice's sketch to solve it.

To solve the problem, Bob will do an exhaustive enumeration on cut
queries, and here is where we use that Alice's sketch preserves all
cuts. Namely, for each subset $S$ of half of the vertices in the left
part of $G_i$, Bob queries the cut $S \cup T$.  As Bob ranges over all
(exponentially many) such cuts, what will happen is that for most
vertices $u$ in the left part for which $|N(u) \cap T| >
\frac{1}{4\eps^2} + \frac{c}{\eps}$, the capacity of $S \cup T$ is a
``little bit'' larger if $u$ is excluded from $S$. This little bit is
not enough to be detected, since $|N(u) \cap T| = \Theta \left
(\frac{1}{\eps^2} \right )$ while the capacity of $S \cup T$ is
$\Theta \left ( \frac{1}{\eps^4} \right )$. However, as Bob range over
all such $S$, he will eventually get lucky in that $S$ contains all
vertices $u$ for which $|N(u) \cap T| > \frac{1}{4\eps^2} +
\frac{c}{\eps}$, and now since there are about $\frac{1}{2\eps^2}$
such vertices, the little $\frac{c}{\eps}$ bit gets ``amplified'' by a
factor of $\frac{1}{2\eps^2}$, which is just enough to be detected by
a $(1+\eps)$-approximation to the capacity of $S \cup T$. If Bob finds
the $S$ which maximizes the (approximate) cut value $S \cup T$, he can
check if his $v$ is in $S$, and this gives him a correct answer with
large constant probability. 

We believe our main contribution
is in designing a communication problem which requires Alice's sketch 
to preserve all cuts instead of only a single cut. 
There are also several details in the communication
lower bound for the problem itself, including a direct-sum theorem
for a constrained version of the Gap-Hamming-Distance problem,
which could be independently useful. 

For the $\Omega(n \eps^{-2})$ edge lower bound for cut sparsifiers, 
the straightforward encoding would
encode each edge using $O(\log n)$ bits, and cause us to lose a $\log n$ factor
in the lower bound. Instead, we show how to randomly
round each edge weight in the sparsifier 
to an adjacent {\it integer}, and observe that the integer 
weights sum up to a small value in our communication problem. 
This ultimately allows to transmit, in a communication-efficient manner,
all the edge weights together with the edge identities.

%Our lower bound for arbitrary data structures is given in Theorem \ref{thm:sketchLB}, and the application to cut sparsifiers is given in Theorem \ref{thm:sparsifierLB}.

\subsubsection{Upper Bound for Sketching Laplacian Matrices with Cut Queries, ``For Each'' Model}

To discuss the main ideas behind our $\tilde O(n \eps^{-1})$-bit sketch construction for Laplacian matrices with queries $x \in \{0,1\}^n$, let us first give some intuition on why the previous algorithms cannot yield a $\tilde O(n \eps^{-1})$ bound, and show how our algorithm circumvents these
roadblocks on a couple of illustrative examples. For concreteness, it
is convenient to think of $\eps=1/\sqrt{n}$.

All existing cut (and spectral) sparsifiers algorithms construct the sparsifier by taking a subgraph of the original graph $G$, with the
``right'' re-weightening of the edges \cite{BK96,SS11,BSS14,FHHP11,KP12}. 
In fact, except for \cite{BSS14}, they all proceed by sampling edges 
independently, each with its own probability (that depends on the graph).

Consider for illustration the complete graph. In this case, 
these sampling schemes employ a uniform probability $p\approx
\tfrac{1/\eps^2}{n}$ of sampling every edge. It is not hard to
see that one cannot sample edges with probability less than $p$, as
otherwise anti-concentration results suggest that even the degree of a
vertex (i.e., the cut of a ``singleton'') is not preserved within 
$1+\eps$ approximation. 
Perhaps a more interesting example is a random graph ${\mathcal G}_{n,1/2}$;
if edges are sampled independently with (roughly) uniform probability,
then again it cannot be less than $p$, because of singleton cuts.
However, if we aim for a sketch for the complete graph or ${\cal G}_{n,1/2}$,
we can just store the degree of each vertex using only $O(n)$
space, and this will allow us to report the value of every singleton
cut (which is the most interesting case, as the standard deviation for
these cut values have multiplicative order roughly $1\pm \eps$).
These observations suggest that {\em sketching} a graph may go beyond 
considering a subgraph (or a different graph) to represent the
original graph $G$.

Our general algorithm proceeds in several steps. The core of our
algorithm is a procedure for handling cuts of value $\approx 1/\eps^2$ 
in a graph with unweighted edges, which proceeds as follows. 
First, repeatedly partition the graph along every {\em sparse} cut, 
namely, any cut whose sparsity is below $1/\eps$. 
This results with a partition of the vertices into some number of parts.
We store the cross-edges (edge connecting different parts) explicitly.
We show the number of such edges is only $\tilde O(n \eps^{-1})$, 
and hence they fit into the space allocated for the sketch.
Obviously, the contribution of these edges to any desired cut $w(S,\bar S)$ 
is easy to compute from this sketch.

The sketching algorithm still needs to estimate the contribution 
(to a cut $w(S,\bar S)$ for a yet unknown $S\subset V$)
from edges that are inside any single part $P$ of the partition. 
To accomplish this, we sample $\approx 1/\eps$ edges out of each vertex, 
and also store the exact degrees of all vertices. 
Then, to estimate the contribution of edges inside a part $P$ to $w(S,\bar S)$,
we take the sum of (exact) degrees of all vertices in $S\cap P$,
\emph{minus} an estimate for (twice) the number of edges inside $S\cap P$
(estimated from the edge sample).
This ``difference-based'' estimate has a smaller variance than 
a direct estimate for the number edges in $(S\cap P, \bar S\cap P)$ 
(which would be the ``standard estimate'', in some sense
employed by previous work). 
The smaller variance is achieved thanks to the facts that 
(1) the assumed cut is of size (at most) $1/\eps^2$; and 
(2) there are no sparse cuts in $P$.

Overall, we achieve a sketch size of $\tilde O(n \eps^{-1})$.  We can
construct the sketch in polynomial time by employing an $O(\sqrt{\log
  n})$-approximation algorithm for sparse cut
\cite{ARV09,sherman2009breaking} or faster algorithms with
$(\log^{O(1)} n)$-approximation \cite{madry2010fast}.

%The formal statement of our main upper bound appears in Theorem~\ref{thm:upper}.

\subsubsection{Upper Bound for Sketching Laplacian Matrices with Spectral Queries, ``For Each'' Model}
Now we consider spectral queries $x \in \mathbb{R}^n$,
starting first with a space bound of $n \eps^{-1.66} \textrm{polylog}(n)$ bits, 
and then discuss how to improve it further to 
$n \eps^{-1.6} \polylog(n)$. 

We start by making several simplifying assumptions. 
The first is that the total number of edges is $O(n \eps^{-2})$. 
Indeed, we can first compute a spectral sparsifier \cite{BSS14}. It is useful
to note that if all edges weights were between $1$ and $\textrm{poly}(n)$, then after spectral sparsification
the edge weights are between $1$ and $\textrm{poly}(n)$, for a possibly larger polynomial. Next, we 
can assume all edge weights are within a factor of
$2$. Indeed, by linearity of the Laplacian, if all edge weights are in $[1, \textrm{poly}(n)]$, then 
we can group the weights 
into powers of $2$ and sketch each subset of edges separately, incurring an $O(\log n)$ factor blowup in space. 
Third, and most importantly, we assume that Cheeger's constant $h_G$
of each resulting graph $G = (V,E)$ satisfies $h_G > \eps^{1/3}$, 
where recall that $h_G = \inf_{S \subset V} \Phi_G(S)$ where
$$
  \Phi_G(S) = \frac{w(S, \bar{S})}{\min\{\vol(S), \vol(\bar{S})\}} 
  \; \text{ and } \;
  \vol(S) = \sum_{u \in S} w(\{u\}, V\setminus \{u\} ) .
%\sum_{v \mid (u,v) \in E} w(u,v)
$$
We can assume $h_G > \eps^{1/3}$
because if it were not, then by definition of $h_G$ there is a sparse cut, that is, 
$\Phi_G(S) \leq \eps^{1/3}$. We can find a sparse cut (a polylogarithmic approximation suffices), store all
sparse cut edges in our data structure, and remove them from the graph $G$. We can then recurse on
the two sides of the cut. By a charging argument we can bound the total
number of edges stored across all sparse cuts. 

As for the actual data structure achieving our $n \eps^{-1.66} \textrm{polylog}(n)$ upper bound, 
we first store the weighted degree 
$\delta_u(G) = \sum_{v: (u,v) \in E} w(u,v)$ of each node (as that for the cut queries). A difference is that we now partition vertices into
``heavy'' and ``light'' classes $V_L$ and $V_H$, where $V_H$ contains those vertices whose weighted degree
exceeds a threshold, and light consists of the remaining vertices. We include all edges incident to light
vertices in the data structure. The remaining edges have both endpoints heavy and for each heavy vertex, we
randomly sample about $\eps^{-5/3}$ of its neighboring heavy edges; edge ${u,v}$ is sampled with probability
$\frac{w(u,v)}{\delta_u(G_H)}$ where $\delta_u(G_H)$ is the sum of weighted edges from the heavy vertex $u$
to neighboring heavy vertices $v$. 

For the estimation procedure, we write
$x^T L x = \sum_{(u,v) \in E} (x_u - x_v)^2 w(u,v)$ 
as
\begin{align*}
  x^T L x &= \sum_{u \in V} \delta_u(G) x_u^2 - \sum_{u \in V_L, v \in V} x_u x_v w(u,v)\\
  &- \sum_{u \in V_H, v \in V_L} x_u x_v w(u,v) - \sum_{u \in V_H} \sum_{v \in V_H} x_u x_v w(u,v)
\end{align*}
%$$x^T L x = \sum_{u \in V} \delta_u(G) x_u^2 - \sum_{u \in V_L, v \in V} x_u x_v w(u,v)
%- \sum_{u \in V_H, v \in V_L} x_u x_v w(u,v) - \sum_{u \in V_H} \sum_{v \in V_H} x_u x_v w(u,v),$$
and observe that our data structure has the first three summations on the right exactly; error can only come from
estimating $\sum_{u \in V_H} \sum_{v \in V_H} x_u x_v w(u,v)$, 
for which we use our sampled heavy edges. 
Since this summation has only heavy edges, 
we can control its variance and upper bound it by $\eps^{10/3} \|D^{1/2} x\|_2^4$,
where $D$ is a diagonal matrix with the degrees of $G$ on the diagonal.
We can then upper bound this norm by relating it to the first non-zero eigenvalue $\lambda_1(\tilde{L})$ 
of the normalized Laplacian $\tilde{L}$, which cannot be too small, since by Cheeger's inequality, 
$\lambda_1(\tilde{L}) \geq {h_G^2}/{2}$, and we have ensured that $h_G$ is large.

To improve the upper bound to $n \eps^{-1.6} \textrm{polylog}(n)$ bits, 
we partition the edges of $G$ into more refined groups, 
based on the degrees of their endpoints. More precisely, we classify edges
$e$ by the minimum degree of their two endpoints, call this number $m(e)$, 
and two edges $e, e'$ are in the same class if 
the nearest power of $2$ of $m(e)$ and of $m(e')$ is the same. 
We note that the total number of vertices
with degree in $\omega(\eps^{-2})$ is $o(n)$, since we are starting with a graph with only $O(n \eps^{-2})$ edges;
therefore, all edges $e$ with $m(e) = \omega(\eps^{-2})$ can be handled by applying our entire procedure recursively
on say, at most $n/2$ nodes. Thus, it suffices to consider $m(e) \leq \eps^{-2}$. 

The intuition now is that as $m(e)$ increases, the variance of our estimator decreases since the two endpoints
have even larger degree now and so they are even ``heavier'' than before. Hence, we need fewer edge samples
when processing a subgraph restricted to edges with large $m(e)$. On the other hand, a graph on edges $e$ 
for which every value of $m(e)$ is small simply cannot have too many edges; indeed, every edge is incident
to a low degree vertex. Therefore, when we partition the graph to ensure that Cheeger's constant $h_G$ is small,
since there are fewer total edges (before we just assumed this number was upper bounded by $n \eps^{-2}$), now
we pay less to store all edges across sparse cuts. Thus, we can balance these two extremes, and doing so we
arrive at our overall $n \eps^{-1.6} \textrm{polylog}(n)$ bit space bound. 

Several technical challenges arise when performing this more refined partitioning. One is that when doing the
sparse cut partitioning to ensure the Cheeger's constant is small, we destroy the minimum degree of endpoints of
edges in the
graph. Fortunately we can show that for our setting of parameters, the total number of edges removed along
sparse cuts is small, and so only a small number of vertices have their degree drop by more than a factor of $2$.
For these vertices, we can afford to store all edges incident to them directly, so they do not contribute
to the variance. Another issue that arises is that to have small variance, we would like to ``assign'' each
edge $\{u,v\}$ to one of the two endpoints $u$ or $v$. If we were to assign it to both, we would have higher
variance. This involves creating a companion
or ``buddy graph'' which is a directed graph associated with the original graph. This directed graph assists
us with the edge partitioning, and tells us which edges to potentially sample from which vertices.

\subsection{Application to Distributed Minimum Cut}
\label{sec:min-cut}

We now illustrate how a ``for each'' sketch can be useful algorithmically
despite its relaxed guarantees compared to a cut sparsifier.
In particular, we show how to $(1+\eps)$-approximate the global minimum cut
of a graph whose edges are distributed across multiple servers. 
Distributed large-scale graph
computation has received recent attention, where protocols for
distributed minimum spanning tree, breadth-first search, shortest paths,
and testing connectivity have been studied, among other problems,
see, e.g., \cite{knpr15,wz13}. 
In our case, each server locally computes the ``for each'' data structure 
of Sec.~\ref{sec:cut-each} on its subgraph (for accuracy $\eps$), 
and sends it to a central server.
Each server also computes a classical cut sparsifier, with
fixed accuracy $\eps' = 0.2$, and sends it to the central server. 
Using the fact that cut-sparsifiers can be merged, 
the central server obtains a $(1 \pm \eps')$-approximation to all cuts
in the union of the graphs. By a result of Henzinger and Williamson \cite{hw96}
(see also Karger \cite{karger2000minimum}), there are only $O(n^2)$
cuts strictly within factor $1.5$ of the minimum cut, and they can be found efficiently
from the sparsifier (see \cite{karger2000minimum} for an $\tO(n^2)$ time way of implicitly
representing all such cuts). The central server then
evaluates each ``for each'' data structure on each of these cuts, and sums up the
estimates to evaluate each such cut up to factor $1+\eps$,
and eventually reports the minimum found. 
Note that the ``for each'' data structures can be assumed,
by independent repetitions, to be correct with probability $1-1/n^4$ 
for any fixed cut (and at any server),
and therefore correct with high probability on all $O(n^2)$ candidate cuts.

%%% Local Variables: 
%%% mode: latex
%%% TeX-master: "paper"
%%% End: 

\section{Preliminaries}
%\noindent{\bf Notations and definitions.}
Let $G = (V, E, w)$ be an undirected graph with weight function $w : V \times V \to \mathbb{R}_{\ge 0}$. 
Denote $n = \card V$ and $m = \card E$, and assume by convention 
that $w(u,v)>0$ whenever $(u,v)\in E$, and $w(u,v)=0$ otherwise.
Let $w_{\max}$ and $w_{\min}$ be the maximum and minimum (positive) 
weights of edges in $E$ respectively.

For a vertex $u \in V$, let $\delta_u(G)$ be the weighted degree of $u$ in $G$, i.e. $\delta_u(G) = \sum_{v \in V} w(u,v)$, and let $d_u(G)$ be the unweighted degree of $u$ in $G$, i.e. $d_u(G) = \abs{\{v \in V\ |\ w(u,v) > 0\}}$.

For two disjoint vertex sets $A$ and $B$, let $\partial(A,B)=\{(u,v)\in E\ | u\in A,v\in B\}$ be the set of edges across two vertex sets $A$ and $B$.  Abusing the notation a little bit, let $\partial(A,A)$ denotes the set of edges whose endpoints are both inside $A$. Given a cut $(S,\bar{S})$ with $S\subset V$, we denote the cut weight in $G$ as $w_G(S,\bar{S})=\sum_{e\in \partial(S,\bar{S})}w(e)$ (we often omit the subscript $G$ when it is clear from the context).

Let $\vec{G} = (V, \vec{E}, w)$ be a directed positively weighted graph. For a vertex $u \in V$, define weighted in- and out-degree $\windeg{u} = \sum_{v: (v, u) \in \vec{E}} w(v, u)$, $\woutdeg{u} = \sum_{v: (u, v) \in \vec{E}} w(u, v)$, and unweighted in- and out-degree $\uwindeg{u} = |\{v\ |\ (v, u) \in \vec{E}, w(v, u) > 0\}|$, and $\uwoutdeg{u} = |\{v\ |\ (u, v) \in \vec{E}, w(u, v) > 0\}|$.

For a vertex set $S \subset V$, let $G(S)$ be the vertex-induced subgraph of $G$, and $E(S)$ be the edge set of $G(S)$. And for an edge set $F \subset E$, let $G(F)$ be the edge-induced subgraph of $G$, and $V(F)$ be the vertex set of $G(F)$; definitions will be the same if edges are directed.

Let $L(G)$ be the unnormalized Laplacian of $G$, and let $\tilde{L}(G)$ be the normalized Laplacian of $G$.  
%Let $\lambda_1(L)$ be the first non-zero eigenvalue of the matrix $L$ (recall that the smallest eigenvalue of $L$ is always $\lambda_0(L)=0$).

Let $\lambda_0(L),\lambda_1(L), \ldots, \lambda_{n-1}(L)$ be the eigenvaules of a Laplacian matrix $L$ such that $\lambda_0(L) \leq \lambda_1(L) \le \ldots \le \lambda_{n-2}(L)\leq \lambda_{n-1}(L)$. Recall that the smallest eigenvalue of $L$ is always $\lambda_0(L) = 0$. We can always assume $G$ is connected, since otherwise we can sketch each connected component separately.  We thus have $\lambda_1(L) >0$.

%For all the definitions above, we will sometimes omit ``$(G)$" when $G$ is clear from the context.

A random variable $X$ is called a $(1+\eps, \delta)$-approximation of $Y$ 
if $\Pr[X\in (1\pm\eps)Y] \geq 1-\delta$.
%$(1-\eps)Y \le X \le (1+\eps)Y$ with probability at least $1 - \delta$. 
We usually assume $\eps > 1/n$ (or similar) since otherwise 
the sketch can just store the whole graph.

We define cut and spectral sketch of a graph as follows.

\begin{definition}[$(1+\eps, \delta)$-cut-sketch]
A sketch of $G = (V,E,w)$, denoted $\sk{G}$, is called a {\em $(1+\eps, \delta)$-cut-sketch} of $G$ if there is a reconstruction function that given $\sk{G}$ and $S\subset V$,  outputs a $(1+\eps,\delta)$-approximation to the weight of the cut $(S,\bar{S})$, i.e., $w(S,\bar{S})$.
\end{definition}

\begin{definition}[$(1+\eps, \delta)$-spectral-sketch]
A sketch of $G = (V,E,w)$, denoted $\sk{G}$, is called a {\em $(1+\eps, \delta)$-spectral-sketch} of $G$ if there is a reconstruction function that given $\sk{G}$ and $x\in\mathbb{R}^n$, outputs a $(1+\eps,\delta)$-approximation to $x^TL(G)x$.
\end{definition}

%Note that spectral sketch is different from spectral sparsifier~\cite{BSS14} in that (1) spectral sketch is ``for each" (i.e., preserve the value of the Laplacian quadratic form $x^TL(G)x$ for each $x \in \mathbb{R}^{n}$ with high probability) while spectral sparsifier is ``for all" (i.e., works for all $x \in \mathbb{R}^{n}$ with high probability). Clearly, a spectral sparsifier is also a spectral sketch. And (2) spectral sparsifier is a subgraph of the original graph $G$, while spectral sketch is not necessarily a subgraph. 

%$V(F)$: $F$ is a subgraph of $G$ or a subset of $E$, $V(F)$ is the vertex set of $F$.

%$E(F)$: similar.

%Let $\alpha$ be a parameter which is a function of $\eps$ (will be specified later).

%Other notations will be introduced when needed. We make several further conventions to simplify our notations:

%any function $f$ that should be interpreted under a graph $G$ is written as $f_G(\cdot)$ or $f(G)$. When $G$ is clear from the context, we simply use $f$;  When no confusion from the context, we sometimes use the set of edges to represent its corresponding graph.

We will need Cheeger's constant and Cheeger's inequality.

\begin{definition}[Cheeger's constant]
Given a graph $G = (V,E,w_G)$, for any $S \subset V$, let $\vol_G(S) = \sum_{u\in S}\delta_u(G)$ be the weighted volume of $S$ in $G$. Let $\Phi_G(S) = \frac{w_G(S, \overline{S})}{\min\{\vol_G(S), \vol_G(\bar{S})\}}$ be the {\em conductance} of the cut $(S, \bar{S})$. We define {\em Cheeger's constant} to be $h_G = \min_{S\subset V}\Phi_G(S)$.
\end{definition}

\begin{lemma}[Cheeger's inequality]
\label{lem:cheeger}
Let $G = (V, E, w)$ be an undirected positively weighted graph. Let $\tilde L$ be the normalized Laplacian of $G$. Let $h_G$ be the Cheeger's constant of graph $G$. The following inequality holds,
\begin{equation}
  \label{eq:cheeger}
  \lambda_1(\tilde L) \geq \frac{h_G^2}{2}.
\end{equation}
\end{lemma}

We introduce another parameter of a graph.
\begin{definition}[Expansion constant]
For any $S \subset V$, we define {\em expansion constant} of graph $G = (V, E)$ to be $\Gamma_G=\min_{|S|\leq n/2} \frac{|\partial(S,\bar{S})|}{|S|}$.
\end{definition}

%%% Local Variables: 
%%% mode: latex
%%% TeX-master: "paper"
%%% End: 

%--------  PSD --------------

%--------  SDD --------------

\section{Symmetric Diagonally Dominant Matrices, ``For All'' Model}
\label{sec:sdd-all}

In this section we prove a lower bound $\Omega(n/\eps^2)$ on the size of ``for all'' sketches for symmetric diagonally dominant (SDD) matrices. 
Our lower bound in fact applies even in the the special case 
of a graph Laplacian $A$ with query vectors $x \in \{0,1\}^n$.
Concretely, we prove the following two theorems in Sections~\ref{sec:sketchLB} and~\ref{sec:sparsifierLB}, respectively. 
The first holds for arbitrary sketches, 
and the second holds for sketches that take the form of a {\em graph}.

\begin{theorem} \label{thm:sketchLB}
Fix an integer $n$ and $\veps\in (1/\sqrt n, 1)$, 
and let $\ske=\ske_{n,\veps}$ and $\est=\est_{n,\veps}$ be possibly randomized
sketching and estimation algorithms for unweighted graphs on vertex set $[n]$.
Suppose that for every such graph $G=([n],E)$, 
with probability at least $3/4$ we have%
\footnote{The probability is over the randomness of the two algorithms;
equivalently, the two algorithms have access to a common source of random bits.}
$$
  \forall S\subset [n],\quad \est\big(S,\ske(G)\big) \in\ (1\pm\veps)\cdot \card{\partial(S,\bar S)}.
$$
Then the worst-case size of $\ske(G)$ is $\Omega(n/\veps^2)$ bits.
\end{theorem}

\begin{theorem} \label{thm:sparsifierLB}
For every integer $n$ and $\veps\in (1/{\sqrt n},1)$, 
there is an $n$-vertex graph $G$ for which every $(1+\eps)$-cut sparsifier $H$
has $\Omega(n/\eps^2)$ edges, 
even if $H$ is not required to be a subgraph of $G$.
\end{theorem}

In addition, we show in Appendix~\ref{app:sdd-each} how to reduce
the quadratic form of an SDD matrix can to that of a Laplacian matrix
(with only a modest increase in the matrix size, from order $n$ to order $2n$).
Thus, the upper bound of sketching SDD matrices in both ``for each'' and ``for all'' cases will be the same as that for Laplacians. Since in the ``for all'' case, we can build the cut (or spectral) sparsifier for a graph using $\tilde{O}(n/\eps^2)$ bits (using e.g.~\cite{BSS14}), we can also construct a ``for all'' sketch for an SDD matrix using $\tilde{O}(n/\eps^2)$ bits of space.  This means that our $\Omega(n/\eps^2)$ lower bound is tight up to some logarithmic factors.

\subsection{Sketch-size Lower Bound}
\label{sec:sketchLB}

We prove Theorem~\ref{thm:sketchLB} using the following communication 
lower bound for a version of the Gap-Hamming-Distance problem, 
whose proof appears in section~\ref{sec:CC}.
Throughout, we fix $c = 10^{-3}$ (or a smaller positive constant),
and assume $\veps\leq c/10$.

\begin{theorem} \label{thm:GHD}
Consider the following distributional communication problem:
Alice has as input $n/2$ strings $s_1,\ldots,s_{n/2} \in \zo^{1/\veps^2}$
of Hamming weight $\tfrac{1}{2\veps^2}$, 
and Bob has an index $i\in[n/2]$ together with one string $t\in \zo^{1/\veps^2}$
of Hamming weight $\tfrac{1}{2\veps^2}$, 
%$n$ strings $t_1,\ldots,t_n \in \zo^{1/\veps^2}$. 
drawn as follows:%
\footnote{Alice's input and Bob's input are \emph{not} independent,
but the marginal distribution of each one is uniform over its domain,
namely, $\zo^{(n/2)\times (1/\veps^2)}$ and $[n]\times \zo^{1/\veps^2}$,
respectively.
}
\begin{itemize} \compactify
\item $i$ is chosen uniformly at random; 
\item $s_i$ and $t$ are chosen uniformly at random but conditioned on 
their Hamming distance $\Delta(s_i,t)$ being, with equal probability,  
either $\geq \tfrac{1}{2\veps^2} + \tfrac{c}{\veps}$ 
or     $\leq \tfrac{1}{2\veps^2} - \tfrac{c}{\veps}$;
\item the remaining strings $s_{i'}$ for $i'\neq i$ are chosen uniformly at random.
\end{itemize}
Consider a (possibly randomized) one-way protocol,
in which Alice sends to Bob an $m$-bit message,
and then Bob determines, with success probability at least $2/3$, 
whether $\Delta(s_i,t)$
is $\geq \tfrac{1}{2\veps^2} + \tfrac{c}{\veps}$ or $\leq \tfrac{1}{2\veps^2} - \tfrac{c}{\veps}$.
Then Alice's message size is $m\geq \Omega(n/\veps^2)$ bits. 
\end{theorem}

We can interpret the lower bound of Theorem \ref{thm:GHD} as follows:
Consider a (possibly randomized) algorithm that produces an $m$-bit
sketch of Alice's input $(s_1,\ldots,s_{n/2})\in\zo^{n/2\veps^2}$, and
suppose that the promise about $\Delta(s_i,t)$ can be decided
correctly (with probability at least $3/4$) given (only) the sketch
and Bob's input $(i,t)\in [n/2]\times \zo^{1/\veps^2}$.  Then $m\geq
\Omega(n/\veps^2)$.

We now prove Theorem \ref{thm:sketchLB} by a reduction to the above communication problem,
interpreting the one-way protocol as a sketching algorithm, as follows.
Given the instance $(s_1,\ldots,s_{n/2},i,t)$, define an $n$-vertex graph $G$ 
that is a disjoint union of the graphs $\aset{G_j: j\in[\veps^2n/2]}$,
where each $G_j$ is a bipartite graph, whose two sides,
denoted $L(G_j)$ and $R(G_j)$, are of size 
\[ \card{L(G_j)}=\card{R(G_j)}=1/\veps^2. \]
The edges of $G$ are determined by $s_1,\ldots,s_{n/2}$,
where each string $s_u$ is interpreted as a vector of indicators for 
the adjacency between vertex $u\in \cup_{j\in[\veps^2n/2]} L(G_j)$
and the respective $R(G_j)$.

Observe that Alice can compute $G$ without any communication,
as this graph is completely determined by her input.
She then builds a sketch of this graph, that with probability $\geq 99/100$, 
succeeds in simultaneously approximating all cut queries within factor 
$1\pm\gamma\veps$, where $\gamma>0$ is a small constant to be determined later. 
This sketch is obtained from the theorem's assumption about $m$-bit sketches 
by standard amplification of the success probability from $3/4$ to $0.99$
(namely, repeating $r=O(1)$ times independently 
and answering any query with the median value of the $r$ answers).
Alice then sends this $O(m)$-bit sketch to Bob. 

Bob then uses his input $i$ to compute $j=j(i)\in[\veps^2n/2]$ 
such that the graph $G_j$ contains vertex $i$ 
(i.e., the vertex whose neighbors are determined by $s_i$).
Bob also interprets his input string $t$ as a vector of indicators
determining a subset $T\subseteq R(G_j)$.
By construction of $G_j$ 
%(recall $j=j(i)$ for Bob's input $i$)
, the neighbor sets $N(v)$ of the vertices $v\in L(G_j)\setminus \aset{i}$
are uniformly distributed, independently of $T$ and of each other;
in particular, each $\card{N(v)\cap T}$ has a Binomial distribution 
$B(\tfrac{1}{\veps^2},\tfrac14)$.

\begin{lemma} \label{lem:BobRecovers}
Using the $O(m)$-bit sketch he received from Alice, 
Bob can compute a ``list'' $B\subset L(G_j)$ of size 
$\card{B} = \half\card{L(G_j)} = \tfrac{1}{2\veps^2}$,
and with probability at least $0.96$, this list contains 
at least $\tfrac{4}{5}$-fraction of the vertices in the set 
\begin{equation} \label{eq:Nlarge}
 \Lhigh = \aset{v\in L(G_j):\ 
      \card{N(v)\cap T} \geq \tfrac{1}{4\veps^2} + \tfrac{c}{\veps}
      }.
\end{equation}
Moreover, Bob uses no information about his input $i$ other than $j=j(i)$.
\end{lemma}

Before proving the lemma, let us show how to use it to decide 
about $\Delta(s_i,t)$ and derive the theorem.
%and basic but tedious additional calculations yield the next lemma.
We will need also the following simple claim, 
which we prove further below. %in Appendix \ref{app:proofs}. 

\begin{claim} \label{cl:binom1}
With probability at least $0.98$, 
the relative size of $\Lhigh$ is 
$
  \frac{ \card{\Lhigh} }{ \card{L(G_j)} } \in [\half-10c,\half]
$.
% the fraction of vertices $v\in L(G_j)$ satisfying \eqref{eq:Nlarge} 
% %for which |N(v) cap T| >= 1/(4eps^2) + c/(2eps) 
% is in the range $[\half-\eta,\half + o(1)]$,
\end{claim}

We assume henceforth that the events described in the above lemma and claim
indeed occur, which happens with probability at least $0.94$.
Notice that $\Delta(s_i,t) = \deg(i) + \card{T} - 2\card{N(i)\cap T}$.
\rnote{We assume here/recall that all the strings $s_i$ and also $t$ have 
Hamming weight $\tfrac{1}{2\veps^2}$, 
we get that $\deg(i)=\card{T}=\tfrac{1}{2\veps^2}$.}

Now suppose that $\Delta(s_i,t) \leq \tfrac{1}{2\veps^2} - \tfrac{c}{\veps}$.
Then $\card{N(i)\cap T} \geq \tfrac{1}{4\veps^2} + \tfrac{c}{2\veps}$, 
and because Bob's list $B$ is independent of the vertex $i\in L(G_j)$,
we have 
$\Pr[i\in B] 
 \geq \tfrac{4}{5}\card{\Lhigh} / \card{\Lhigh}
 = \tfrac{4}{5}
$.

Next, suppose that $\Delta(s_i,t) \geq \tfrac{1}{2\veps^2} + \tfrac{c}{\veps}$.
Then $\card{N(i)\cap T} \leq \tfrac{1}{4\veps^2} - \tfrac{c}{2\veps}$, 
and using Claim \ref{cl:binom1},
\[
  \Pr[i\in B] 
  \leq \frac{ \card{B} - \tfrac{4}{5}\card{\Lhigh} }{ \card{L(G_j)} } 
  \leq \frac{1}{4}.
\]
\aanote{Shouldn't it be divided by $L(G_j)-|L_{high}|$ ?}
Thus, Bob can decide between the two cases 
with error probability at most $1/4$. 
Overall, it follows that Bob can solve the Gap-Hamming-Distance problem for $(s_i,t)$,
with overall error probability at most $1/4 + 0.06 < 1/3$,
as required to prove the theorem.

\medskip
\begin{proof}[Proof of Claim \ref{cl:binom1}]
By basic properties of the Binomial distribution (or the Berry-Esseen Theorem),
there are absolute constants $\tfrac{1}{5}\leq K_1 \leq K_2\leq 5$ 
\rnote{I have not really verified these constants} 
such that for each vertex $v\in L(G_j)$, 
\[
  \Pr\Big[v\in \Lhigh\Big] 
  = \Pr\Big[\card{N(v)\cap T} \geq \tfrac{1}{4\veps^2} + \tfrac{c}{\veps}\Big] 
  \in [\half - K_2 c, \half - K_1 c].
\]
Denoting $Z = \card{\Lhigh}$, then $\E[Z] \in [\frac{1}{2\eps^2} - \frac{K_2c}{\eps^2}, \frac{1}{2\eps^2} - \frac{K_1c}{\eps^2}]$. We have by Hoeffding's inequality, 
\[
  \Pr\Big[ \abs{Z-\EX[Z]} > \tfrac{K_1 c}{\veps^2} \Big]
  \leq 2e^{-\half (K_1 c)^2 (1/\veps^2)}
  \leq 0.02.
\]
Thus, with probability at least $0.98$, we have both bounds
\begin{align*}
  Z &\leq \EX[Z] + \tfrac{K_1 c}{\veps^2} \le  \tfrac{1}{2\eps^2} - \tfrac{K_1c}{\eps^2} + \tfrac{K_1c}{\eps^2} = \tfrac{1}{2\veps^2}, 
  \text{ and } \\
  Z &\geq \EX[Z] - \tfrac{K_1 c}{\veps^2} \ge \tfrac{1}{2\eps^2} - \tfrac{K_1c}{\eps^2} - \tfrac{K_2c}{\eps^2}
     \geq (\half-2K_2c)\tfrac{1}{\veps^2} 
     \geq (\half-10c)\tfrac{1}{\veps^2}.
\end{align*}
The claim follows by noting $\abs{L(G_j)} = 1/\eps^2$.
\end{proof}

\begin{proof}[Proof of Lemma~\ref{lem:BobRecovers}]
We now show how Bob creates the ``list'' $B\subset L(G_j)$ of size 
$\card{B}=\tfrac{1}{2\veps^2}$. 
Bob estimates the cut value for $S\cup T$ 
for every subset $S\subseteq L(G_j)$ of size exactly $\frac{1}{2\veps^2}$. 
Observe that the cut value for a given $S$ is 
\[
  \delta(S\cup T)
  = \sum_{v\in S} \deg(v) + \sum_{u\in T} \deg(u) - 2 \sum_{v\in S} \card{N(v)\cap T}.
\]

The cut value is bounded by the number of edges in $G_j$, which is at most $1/\veps^4$, 
and since the data structure maintains all the cut values 
within factor $1+\gamma\veps$ for an arbitrarily small constant $\gamma>0$, 
the additive error on each cut value is at most $\gamma/\veps^3$
\aanote{$\gamma/\eps$ ?}.
Further, we can assume Bob knows the exact degrees of all vertices 
(by adding them to the sketch, using $O(n\log\tfrac{1}{\veps})$ bits), 
which he can subtract off, and since scaling by $-1/2$ can only shrink the
additive error, we can define the ``normalized'' cut value
\[ 
  n(S,T) = \sum_{v\in S} \card{N(v)\cap T}, 
\]
which Bob can estimate within additive error $\gamma/\veps^3$ \aanote{$\gamma/\eps$?}. 
Bob's algorithm is to compute these estimates for all the values $n(S,T)$,
and output a set $S$ that maximizes his estimate for $n(S,T)$ 
as the desired list $B\subset L(G_j)$.

% It remains to analyze the success probability of Bob's algorithm.
% Due to space constraints, the proof continues in Appendix~\ref{app:proofs}.

Let us now analyze the success probability of Bob's algorithm.
For each vertex $v\in L(G_j)$, let $f(v) = \card{N(v)\cap T}$. 
Observe that each $f(v)$ has a Binomial distribution 
$B(\tfrac{1}{\veps^2},\tfrac14)$, and they are independent of each other.
We will need the following bounds on the typical values of some order statistics
when taking multiple samples from such a Binomial distribution.
Recall that the \emph{$k$-th order statistic} of a sample (multiset) 
$x_1,\ldots,x_m\in\R$ is the $k$-th smallest element in that sample.
The following claim is proved further below. %in Appendix \ref{app:proofs}. 

\begin{claim} \label{cl:binom}
Let $\aset{X_j}_{j=1,\ldots,m}$ be independent random variables
with Binomial distribution $B(t,\tfrac14)$.
%Let $\alpha\in(\tfrac{7}{\sqrt m}+\tfrac{10}{t},\half)$ 
Let $\alpha\in(0,\half)$ 
such that $(\half+\alpha)m$ is integral,
and both $t,m \geq 10/\alpha^2$.
Then
\begin{align*}
  \Pr&\Big[\text{the $(\half-\alpha)m$ order statistic of $\aset{X_j}$ 
is $\leq \tfrac14 t - \tfrac{\alpha}{10}\sqrt{t}$}\Big] \geq 0.99,
  \text{ and } \\
  \Pr&\Big[\text{the $(\half+\alpha)m$ order statistic of $\aset{X_j}$ 
is $\geq \tfrac14 t + \tfrac{\alpha}{10}\sqrt{t}$}\Big] \geq 0.99.
\end{align*}
\end{claim}

Sort the vertices $v\in L(G_j)$ by their $f(v)$ value, 
and denote them by $v_1,\ldots,v_{1/\veps^2}$ such that $f(v_i) \leq f(v_{i+1})$.
Applying the claim (for $\alpha=0.05$ and $t,m=\tfrac{1}{\veps^2}$), 
we see that with probability at least $0.98$, the difference 
\begin{equation} \label{eq:fdiff}
  f(v_{0.55/\veps^2})-f(v_{0.45/\veps^2}) 
  \geq 0.01 / {\veps}.
\end{equation}
We assume henceforth this event indeed occurs.
Let $S^*$ include the $\frac{1}{2\veps^2}$ vertices $v\in L(G_j)$ 
with largest $f(v)$, i.e., $S^*= \aset{v_j}_{j > 0.5/\veps^2}$,
and let $S'\subset L(G_j)$ be any subset of the same size such that 
at least $\tfrac{1}{10}$-fraction of its vertices are not included in $S^*$
(i.e., their order statistic in $L(G_j)$ is at most $\tfrac{1}{2\veps^2}$).
Then we can write
\begin{align*}
  n(S^*,T)
  &= \sum_{j\in S^*} f(v) 
   = \sum_{j>0.5/\veps^2} f(v_j),
  \\
  n(S',T) 
  &= \sum_{j\in S'} f(v) 
   \leq \sum_{j>0.6/\veps^2} f(v_j) + \sum_{0.4/\veps^2 < j\leq 0.5/\veps^2} f(v_j).
\end{align*}
Now subtract them
\begin{align*}
  n(S',T) - n(S^*,T)
  &= \sum_{0.5/\veps^2 < j\leq 0.6/\veps^2} f(v_j)
    - \sum_{0.4/\veps^2 < j\leq 0.5/\veps^2} f(v_j),
\intertext{observe that elements in the normalized interval $(0.5,0.55]$ 
dominate those in $(0.45,0.5]$,
}
  &\geq \sum_{0.55/\veps^2 < j\leq 0.6/\veps^2} f(v_j)
    - \sum_{0.4/\veps^2 < j\leq 0.45/\veps^2} f(v_j)
\intertext{and bound the remaining elements using \eqref{eq:fdiff}, }
  &\geq (0.05/\veps^2)  \big[ f(v_{0.55/\veps^2})-f(v_{0.45/\veps^2}) \big]
  \geq 0.0005/\veps^3.
\end{align*}
Bob's estimate for each of the values $n(S^*,T)$ and $n(S',T)$ 
has additive error at most $\gamma/\veps^3$, 
and therefore for suitable $\gamma=10^{-4}$, 
the list $B$ Bob computes cannot be this set $S'$.
Thus, Bob's list $B$ must contain at least $9/10$-fraction of $S^*$,
i.e., the $\frac{1}{2\veps^2}$ vertices $v\in L(G_j)$ with highest $f(v)$.

Recall from Claim~\ref{cl:binom1} that with probability at least $0.98$, 
we have $\tfrac{1}{4\veps^2} \leq \card{\Lhigh} \leq \tfrac{1}{2\veps^2}$,
and assume henceforth this event occurs.
Since $S^*$ includes the $\tfrac{1}{2\veps^2}$ vertices with highest $f$-value,
it must contain all the vertices of $\Lhigh$,
i.e., $\Lhigh \subseteq S^*$.
We already argued that Bob's list $B$ contains all 
but at most $\tfrac{1}{10}\card{S^*} = \tfrac{1}{20\veps^2}$ vertices of $S^*$,
and thus 
\[
  \frac{\card{\Lhigh\setminus B}} {\card{\Lhigh}}
  \leq \frac{\card{S^*\setminus B}} {\card{\Lhigh}}
  \leq \frac{\ \tfrac{1}{20\veps^2}\ } {\tfrac{1}{4\veps^2} }
  = \frac{1}{5}.
\]
This bound holds with probability at least $0.96$
(because of two events that we ignored, each having probability at most $0.02$)
and this proves Lemma~\ref{lem:BobRecovers}.
\end{proof}

\medskip
\begin{proof}[Proof of Claim \ref{cl:binom}]
%Let $X_{(j)}$ denote the $j$-th order statistic of $\aset{X_j}$.
The $(\half-\alpha)m$ order statistic of $\aset{X_j}$ 
is smaller or equal to $T= \tfrac14 t - \tfrac{\alpha}{10}\sqrt{t}$
if and only if at least $(\half-\alpha)m$ elements are smaller or equal to $T$, 
which can be written as $\sum_{j=1}^m \indic_{\aset{X_j\le T}} \geq (\half-\alpha)m$.

Now fix $j\in\aset{1,\ldots,t}$. 
Then 
\begin{equation} \label{eq:binom1}
  \Pr[X_j\leq T] = \Pr[X_j \leq \tfrac14 t] \cdot \Pr[X_j \leq T \mid X_j \leq \tfrac14 t],
\end{equation}
and by the Binomial distribution's relationship between mean and median,
$\Pr[X_j \leq \tfrac14 t] \geq \half$.
Elementary but tedious calculations (or the Berry-Esseen Theorem) show 
%that for all $x$ that are close to the mean $\tfrac14t$, 
%the probability measure of the interval $(x-\tfrac{\alpha}{10}\sqrt{t},x]$ 
%is roughly the same. More precisely, 
there is an absolute constant $K\in(0,5)$ such that 
\rnote{I have not verified the constant is indeed smaller than 5;
if not, we should increase the other constant 10 accordingly.}
\[
  \Pr\Big[\tfrac14 t - \tfrac{\alpha}{10}\sqrt{t} < X_j \leq \tfrac14 t\Big]
  \leq K \tfrac{\alpha}{10}\cdot \Pr\Big[X_j \leq \tfrac14 t\Big],
\]
and plugging into \eqref{eq:binom1}, we obtain 
$\Pr[X_j\leq T] \geq \half(1-K\tfrac{\alpha}{10}) \geq \half - \half\alpha$.

Now bound the expectation by
$\EX[\sum_{j=1}^m \indic_{\aset{X_j\le T}}] \geq (\half -\half\alpha)m$,
and apply Hoeffding's inequality, 
\[
  \Pr \Big[\sum_j \indic_{\aset{X_j\le T}} < (\half-\alpha)m \Big]
  \leq e^{-\half (\half \alpha)^2 m}
  = e^{-\alpha^2 m/8}
  \leq 0.01,
\]
where the last inequality follows since $\alpha^2m$ is sufficiently large.
\end{proof}

\subsection{The Communication Lower Bound}
\label{sec:CC}

% Due to space constraints, this material is deferred to Appendix~\ref{app:CC},
% where we prove Theorem~\ref{thm:GHD} by relating our communication problem
% to the Gap-Hamming-Distance problem, where we can apply known 
% information complexity bounds \cite{BGPW13}.

We now prove Theorem \ref{thm:GHD} (see Theorem \ref{thm:GHD2} below),
by considering distributional communication problems
between two parties, Alice and Bob, as defined below. 
We restrict attention to the one-way model, in which Alice sends to Bob
a single message $M$ that is a randomized function of her input
(using her private randomness), and Bob outputs the answer.

\subsubsection{Distributional Versions of Gap-Hamming-Distance.}
Recall that our analysis is asymptotic for $\veps$ tending to $0$, and 
let $0 < c < 1$ be a parameter, considered to be a sufficiently small constant.
Alice's input is $S \in \{0,1\}^{\frac{1}{\eps^2}}$,
Bob's input is $T \in \{0,1\}^{\frac{1}{\eps^2}}$,
where the Hamming weights are $\wt(S) = \wt(T) = \frac{1}{2\eps^2}$, and Bob 
needs to evaluate the partial function
\[
  f_{c}(S,T) = 
  \begin{cases}
    1& \text{if $\Delta(S,T) \geq \frac{1}{2\eps^2} + \frac{c}{\eps}$; }\\
    0& \text{if $\Delta(S,T) \leq \frac{1}{2\eps^2} - \frac{c}{\eps}$. }
  \end{cases}
\]
The distribution $\mu$ we place on the inputs $(S,T)$ is the following: 
$S$ is chosen uniformly at random with $\wt(S) = \frac{1}{2\eps^2}$,
and then with probability $\frac{1}{2}$, we choose $T$ uniformly at random with $\wt(T) = \frac{1}{2\eps^2}$
subject to the constraint that $\Delta(S,T) \geq \frac{1}{2\eps^2} + \frac{c}{\eps}$, 
while with the remaining probability $\frac{1}{2}$, we choose $T$
uniformly at random with $\wt(T) = \frac{1}{2\eps^2}$ 
subject to the constraint that $\Delta(S,T) \leq \frac{1}{2\eps^2} - \frac{c}{\eps}$. 
We say Alice's message $M = M(S)$ is $\delta$-error for $(f_{c}, \mu)$ if 
Bob has a reconstruction function $R$ for which
$$\Pr_{(S,T) \sim \mu, \textrm{ private randomness}}[R(M, T) = f_{c}(S,T)] \geq 1-\delta.$$

Now consider a related but different distributional problem. 
Alice and Bob have $S, T \in \{0,1\}^{\frac{1}{\eps^2}}$, respectively,
each of Hamming weight exactly $\frac{1}{2\eps^2}$, 
and Bob needs to evaluate the function
\[
  g(S,T) = 
  \begin{cases}
    1& \text{if $\Delta(S,T) >    \frac{1}{2\eps^2}$; }\\
    0& \text{if $\Delta(S,T) \leq \frac{1}{2\eps^2}$. }
  \end{cases}
\]
We place the following distribution $\zeta$ on the inputs $(S,T)$: 
$S$ and $T$ are chosen independently and uniformly at random
among all vectors with Hamming weight exactly $\frac{1}{2\eps^2}$. 
We say a message $M$ is $\delta$-error for $(g, \zeta)$ if Bob has a reconstruction function $R$ for which
$$\Pr_{(S, T) \sim \zeta, \ \textrm{private randomness}}[R(M, T) = g(S,T)] \geq 1-\delta.$$

Let $I(S ; M) = H(S) - H(S | M)$ be the mutual information between $S$ and $M$, where $H$ is the entropy
function. Define
$\IC_{\mu, \delta}(f_c) = \min_{\text{$\delta$-error $M$ for $(f_c, \mu)$}} I(S ; M)$ 
and 
$\IC_{\zeta, \delta}(g) = \min_{\text{$\delta$-error $M$ for $(g, \zeta)$}} I(S ; M)$. 

\begin{lemma}\label{lem:icost}
For all $\delta>0$, 
 $\IC_{\mu, \delta}(f_c) 
\geq \IC_{\zeta, \delta + O(c)}(g)$. 
\end{lemma}
\begin{proof}
It suffices to show that if 
$M$ is $\delta$-error for $(f_c, \mu)$, then $M$ is $(\delta + O(c))$-error for $(g, \zeta)$.
Since $M$ is $\delta$-error for $(f_c, \mu)$, Bob has a reconstruction function $R$ for which
$$\Pr_{(S,T) \sim \mu, \textrm{ private randomness}}[R(M, T) = f_{c}(S,T)] \geq 1-\delta.$$ Now consider
$\Pr_{(S,T) \sim \zeta, \textrm{ private randomness}}[R(M, T) = g(S,T)]$. Observe that whenever $(S,T)$ lies in the support of $\mu$, if
$R(M,T) = f_c(S,T)$, then $R(M,T) = g(S,T)$. The probability that $(S,T)$ lies in the support of $\mu$
is $1-O(c)$, by standard anti-concentration arguments of the Binomial distribution (or the Berry-Esseen Theorem), 
and conditioned on this event we have that $(S,T)$ is distributed according to $\mu$. Hence,
$\Pr_{(S,T) \sim \zeta, \ \textrm{private randomness}}[R(M, T) = g(S,T)] 
  \geq [1 - O(c)][1- \delta]
  \geq 1 - O(c) - \delta
  $. 
\end{proof}
We now lower bound $\IC_{\zeta, \delta}(g)$. 
\begin{lemma}\label{lem:icost2}
For $\delta_0 > 0$ a sufficiently small constant, $IC_{\zeta, \delta_0}(g) = \Omega \left (\frac{1}{\eps^2} \right ).$
\end{lemma}
\begin{proof}
We use the following lower bound of 
Braverman, Garg, Pankratov and Weinstein \cite{BGPW13} 
for the following $h_c(S,T)$ problem. Like before, Alice
has $S \in \{0,1\}^{\frac{1}{\eps^2}}$, Bob has $T \in \{0,1\}^{\frac{1}{\eps^2}}$, 
and needs to evaluate the partial function
\[
  h_{c}(S,T) = 
  \begin{cases}
    1& \text{if $\Delta(S,T) \geq \frac{1}{2\eps^2} + \frac{c}{\eps}$; }\\
    0& \text{if $\Delta(S,T) \leq \frac{1}{2\eps^2} - \frac{c}{\eps}$. }
  \end{cases}
\]
However, now $\wt(S)$ and $\wt(T)$ may be arbitrary. 
Moreover, $S$ and $T$ are chosen independently and uniformly at random from $\{0,1\}^{\frac{1}{\eps^2}}$. Denote
this by $(S,T) \sim \eta$. 
Now it may be the case that $\left|\Delta(S,T) - \frac{1}{2\eps^2}\right| < \frac{c}{\eps}$,
in which case Bob's output is allowed to be arbitrary. A message $M$ is $\delta$-error for $(h_c, \eta)$ 
if Bob has a reconstruction
function $R$ for which
$$\Pr_{(S,T) \sim \eta, \textrm{ private randomness}} \left [\left (R(M,T) = h_c(S,T) \right )\wedge \left (\left| \Delta(S,T) - \frac{1}{2\eps^2} \right| \geq \frac{c}{\eps} \right )\right ] \geq 1-\delta.$$
It was proved in \cite{BGPW13} that for a sufficiently small constant $\delta > 0$, 
$$\IC_{\eta, \delta}(h_1) = \min_{\text{$\delta$-error $M$ for $(h_1,\eta)$}} I(S ; M) \geq \frac{C}{\eps^2},$$ for an absolute
constant $C > 0$. We show how to apply this result to prove the lemma.

An immediate corollary of this result is that 
$\IC_{\eta, \delta}(g) = \min_{\text{$\delta$-error $M$ for $(g,\eta)$}} I(S ; M) \geq \frac{C}{\eps^2}$. 
Indeed, if $M$ is $\delta$-error for $(g, \eta)$, then it is also $\delta$-error
for $(h_1, \eta)$. 

Now let $M$ be a $\delta$-error protocol for $(g, \zeta)$. Consider the following randomized protocol $M'$
for $g$ with inputs distributed according to $\eta$. 
Given $S$, Alice computes $s =$ wt($S$). If $s < \frac{1}{2\eps^2}$, Alice randomly chooses
$\frac{1}{2\eps^2} - s$ coordinates in $S$ that are equal to $0$ and replaces them with a $1$, otherwise she randomly
chooses $s - \frac{1}{2\eps^2}$ coordinates in $S$ that are equal to $1$ and replaces them with a $0$. Let $S'$ be
the resulting vector. Alice sends $M(S')$ to Bob, i.e., $M'(S)= M(S')$. 
Given the message $M(S')$ and his input $T$, Bob first computes $t = \wt(T)$. 
If $t < \frac{1}{2\eps^2}$,
Bob randomly chooses $\frac{1}{2\eps^2} - t$ coordinates in $T$ which are equal to $0$ and replaces them with a $1$, otherwise
he randomly chooses $t-\frac{1}{2\eps^2}$ coordinates in $T$ which are equal to $1$ and replaces them with a $0$. Let $T'$
be the resulting vector. Suppose $R$ is such that 
$\Pr_{(S', T') \sim \zeta, \textrm{ private randomness}} [R(M(S'), T') = g(S', T')] \geq 1-\delta.$ Bob outputs $R(M(S'),T')$.

We now lower bound $\Pr[g(S', T') = g(S, T)]$, where the probability is over $(S,T) \sim \eta$ and the random choices of Alice
and Bob for creating $S', T'$ from $S,T$, respectively. First, the number of coordinates changed by Alice or Bob
is $r = \Theta(1/\eps)$ with arbitrarily large constant probability. Since $S$ and $T$ are independent and uniformly random, 
after performing this change, the Hamming distance on these $r$ coordinates is $\frac{r}{2} \pm O(\sqrt{r})$ with
arbitrarily large constant probability. Finally, $\left|\Delta(S', T') - \frac{1}{2\eps^2} \right| = \omega(\sqrt{r})$
with arbitrarily large constant probability. Hence, with arbitrarily large constant probability, $g(S',T') = g(S,T)$.
It follows that $\Pr[g(S', T') = g(S, T)] \geq 1-\gamma$ for an arbitrarily small constant $\gamma > 0$, and
therefore if $R'$ describes the above reconstruction procedure of Bob, then
$\Pr_{(S, T) \sim \eta, \textrm{ private randomness}}[R'(M'(S),T) = g(S, T)] \geq 1-\gamma - \delta$. 

Hence, $M'$ is a $(\delta+\gamma)$-error protocol for $(g, \eta)$. We now bound $I(M' ; S)$ in terms of $I(M ; S')$. 
Let $J$ be an indicator random variable for the event 
$\wt(S)\in \left [\frac{1}{2\eps^2}-\frac{1}{\eps^{3/2}}, \frac{1}{2\eps^2} + \frac{1}{\eps^{3/2}} \right ]$. 
Then $\Pr[J = 1] = 1-o(1)$,
where $o(1) \rightarrow 0$ as $\eps \rightarrow 0$. Since conditioning on a random variable $Z$ can change the mutual
information by at most $H(Z)$, we have 
\begin{eqnarray}\label{eqn:derive1}
I(M' ; S) \leq I(M' ; S \mid J) + H(J) \leq I(M' ; S \mid J = 1) + 1.
\end{eqnarray}
$S$ is a probabilistic function of $S'$, which if $J = 1$, is obtained by changing at most $1/\eps^{3/2}$ randomly
chosen coordinates $A_1, \ldots, A_{1/\eps^{3/2}}$ of $S'$ from $0$ to $1$ or from $1$ to $0$. By the data processing inequality
and the chain rule for mutual information, 
\begin{eqnarray}\label{eqn:derive2}
I(M' ; S \mid J = 1) & \leq & I(M' ; S', A_1, \ldots, A_{1/\eps^{3/2}} \mid J = 1) \nonumber \\
& = & I(M' ; S' \mid J = 1) + \sum_{\ell=1}^{1/\eps^{3/2}} I(M' ; A_{\ell} \mid J = 1, A_1, \ldots, A_{\ell-1}) \nonumber \\
& \leq & I(M' ; S' \mid J = 1) + O \left (\frac{\log(1/\eps)}{\eps^{3/2}} \right ).
\end{eqnarray}
Observe that the joint distribution of $M'(S')$ and $S'$ is independent of $J$, and moreover is equal to the joint distribution
of $M(S')$ and $S' \sim \zeta$. We can take $M$ to be a $\delta$-error protocol for $(g, \zeta)$ for which
$I(M(S') ; S') = \IC_{\zeta, \delta}(g)$. Combining this with (\ref{eqn:derive1}) and (\ref{eqn:derive2}), 
$I(M' ; S ) \leq \IC_{\zeta, \delta}(g) + O \left (\frac{\log(1/\eps)}{\eps^{3/2}} \right ).$
Now since $M'$ is a $(\delta+\gamma)$-error protocol for $(g, \eta)$, we have $I(M' ; S) \geq \IC_{\eta, \delta+\gamma}(g) \geq  \frac{C}{\eps^2}$,
provided $\delta$ and $\gamma$ are sufficiently small constants. It follows that
$ \IC_{\zeta, \delta}(G) 
  \geq \frac{C}{\eps^2} - O \left (\frac{\log(1/\eps)}{\eps^{3/2}} \right )
  \geq \frac{C}{2\eps^2}
$, as desired. 
\end{proof}
\begin{corollary}\label{cor:prim}
For sufficiently small constants $\delta, c > 0$,  $\IC_{\mu, \delta}(f_c) = \Omega(1/\eps^2)$. 
\end{corollary}
\begin{proof}
This follows by combining Lemmas \ref{lem:icost} and \ref{lem:icost2}. 
\end{proof}

\subsubsection{$n$-fold Version of Gap-Hamming-Distance.}
We now consider the $n$-fold problem in which Alice is given $n$ strings
$S_1, \ldots, S_n \in \{0,1\}^{1/\eps^2}$, and Bob has an index
$I \in [n]$ together with one string $T \in \{0,1\}^{1/\eps^2}$. Here $(S_I, T) \sim \zeta$, while $S_j$ for $j \neq I$, are
chosen independently and uniformly at random from all Hamming weight $\frac{1}{2\eps^2}$ vectors. 
Thus the joint distribution of $S_1, \ldots, S_n$ 
is $n$ i.i.d. strings drawn uniformly from $\{0,1\}^{1/\eps^2}$ subject to each of their Hamming weights being $\frac{1}{2\eps^2}$. 
Here $I$ is drawn independently and uniformly from $[n]$. 
We let $\nu$ denote the resulting input distribution. 

We consider the one-way two-party model in which Alice sends a single, possibly randomized message $M$ of her inputs
$S_1, \ldots, S_n$, and Bob needs to evaluate $h(S_1, \ldots, S_n, T) = f_c(S_I, T)$.
We say $M$ is $\delta$-error for $(h, \nu)$ if Bob has a reconstruction
function $R$ for which 
$$\Pr_{\textrm{inputs}\sim \nu, \textrm{ private randomness}} \left [\left (R(M,T, I) = f_c(S_I, T) \right )\wedge \left (\left| \Delta(S_I,T) - \frac{1}{2\eps^2} \right| \geq \frac{c}{\eps} \right )\right ] \geq 1-\delta.$$
Let $\IC_{\nu, \delta}(h) = \min_{\text{$\delta$-error $M$ for $(h, \nu)$}} I(S_1, \ldots, S_n ; M)$.

\begin{theorem} \label{thm:GHD2}
For a sufficiently small constant $\delta > 0$, $\IC_{\nu, \delta}(h) = \Omega(n/\eps^2)$. In particular, the distributional 
one-way communication complexity of $h$ under input distribution $\nu$ is $\Omega(n/\eps^2)$. 
\end{theorem}
\begin{proof}
Say an index $i \in [n]$ is {\it good} if
$$\Pr_{\textrm{inputs}\sim \nu, \textrm{ private randomness}} \left [\left (R(M,T, I) = f_c(S_I, T) \right )\wedge \left (\left| \Delta(S_I,T) - \frac{c}{2\eps^2} \right| \geq \frac{1}{\eps} \right ) \mid I = i\right ] \geq 1-2\delta.$$
By a union bound, there are at least $n/2$ good $i \in [n]$. By the chain rule for mutual information and 
using that the $S_i$ are independent and conditioning does not increase entropy, 
\begin{eqnarray*}
I(M ; S_1, \ldots, S_n) \geq \sum_{i=1}^n I(M ; S_i) \geq \sum_{\textrm{good }i} I(M ; S_i).
\end{eqnarray*}
We claim that for each good $i$, $I(M ; S_i) \geq  \IC_{\mu, 2\delta}(f_c)$. Indeed, consider the following protocol $M_i$ for $f_c$ under distribution $\mu$. 
Alice, given her input $S$ for $f_c$, uses her private randomness to sample $S_j$ for all $j \neq i$ independently and uniformly
at random from $\{0,1\}^{1/\eps^2}$ subject to each of their Hamming weights being $\frac{1}{2\eps^2}$. Bob sets $I = i$ and
uses his input $T$ for $f_c$ as his input for $h$. Since $i$ is good, it follows that $M_i$ is $2\delta$-error
for $(f_c, \zeta)$. Hence $I(M ; S_i) \geq \IC_{\mu, 2\delta}(f_c)$, which by Corollary \ref{cor:prim}, is $\Omega(1/\eps^2)$
provided $\delta > 0$ is a sufficiently small constant. Hence, $\IC_{\nu, \delta} = \Omega(n/\eps^2)$. Since 
$\IC_{\nu, \delta}(h) \leq I(M ; S_1, \ldots, S_n)$ for each $\delta$-error $M$ for $(h, \nu)$, and 
$I(M ; S_1, \ldots, S_n) \leq H(M)$ which is less than the length of $M$, the communication complexity lower bound follows. 
\end{proof}

\subsection{Cut-sparsifiers Lower Bound}
\label{sec:sparsifierLB}

% Due to space constraints, this material is deferred to Appendix~\ref{app:sparsifierLB}, 
% where we proves Theorem~\ref{thm:sparsifierLB},
% by relying on the proof of Theorem~\ref{thm:sketchLB}.

We now prove Theorem~\ref{thm:sparsifierLB}.
Fix $n$ and $\eps$, and assume that every $n$-vertex graph has a $(1+\gamma\eps/2)$-cut sparsifier with at most $s$ edges, 
for a small constant $\gamma>0$ to be determined later. 
We wish to prove a lower bound on $s$.
Consider then the random graph $G$ constructed 
in the proof of Theorem~\ref{thm:sketchLB},
as a disjoint union of graphs $\aset{G_j: j\in[\veps^2n/2]}$, 
each being a bipartite graph with $1/\eps^2$ vertices in each side.
By our assumption above, such $G$ (always) has a subgraph $H$ 
which is a $(1+\gamma\eps/2)$-cut sparsifier having at most $s$ edges.
By Theorem~\ref{thm:sketchLB}, answering all possible cut queries correctly 
(in the sense of approximation factor $1\pm\eps$ with probability $1$) 
requires sketch size $\Omega(n/\eps^2)$ bits.
In fact, by inspecting the proof (specifically, of Lemma~\ref{lem:BobRecovers})
the above holds even if the correctness 
(1) holds only for cut queries contained in a single $G_j$,
i.e., queries $S\cup T$ for $S\subset L(G_j)$ and $T\subset R(G_j)$;
and (2) allows for each cut value 
an additive error of $\gamma/\veps^3$, where $\gamma=10^{-4}$.
The idea now is to encode $H$ using $m\approx s$ bits in a way 
that suffices to correctly answer all such cut queries
i.e., in the context of Lemma~\ref{lem:BobRecovers},
Alice will encode $H$ and send it as her sketch to Bob.
The sketch-size bound $m\geq \Omega(n/\eps^2)$ we proved for $G$ 
would then imply a similar bound on $s$.

Consider then the sparsifier $H$, which is an edge-weighted graph,
while the edges of $G$ all have unit weight.
Observe that $H$ must be a union of disjoint graphs $H_j$ on 
$L(G_j)\cup R(G_j)$ for $j\in[\veps^2n/2]$, 
because $H$ must preserve the corresponding cut, which has value zero.
Let $s_j$ denote the number of edges in $H_j$.
%then clearly $s_j\leq \binom{2/\eps^2}{2} \leq 2/\eps^4$.
It will be convenient to consider each such graph separately,
so fix for now some $j\in[\veps^2n/2]$.
%Every cut query $S\cup T$ contained in this $G_j$
%corresponds to at most $1/\eps^4$ edges in $G$, all of unit weight.

Consider first the case $s_j\leq \gamma^2/(6\eps^4)$,
and let us show how to encode $H_j$.
Construct from $H_j$ another graph $H'_j$ by rounding 
every non-integral edge weight
to one of its two nearby integers independently at random,
in an unbiased manner.
Specifically, each edge weight $w>0$ is rounded upwards to $w'=\ceil{w}$ 
with probability $w-\floor{w}$, and downwards to $w'=\floor{w}$ otherwise.
Now consider a fixed cut query $S\cup T$ in $G_j$,
denoting by $\delta_H(S\cup T)$ its cut value in $H$, and similarly for $H'$.
Then $\EX[\delta_{H'}(S\cup T)] = \delta_H(S\cup T)$,
and since the number of edges participating in this cut (in $H_j$) 
is at most $s_j\leq \gamma^2/(6\eps^4)$, 
by Hoeffding's inequality for $t=\gamma/(2\veps^3)$,
\[
  \Pr\Big[ \abs{\delta_{H'}(S\cup T) - \delta_H(S\cup T)} > t \Big]
%  = \Pr\Big[ \abs{\delta_{H'}(S\cup T) - \delta_H(S\cup T)} > \gamma/(4\veps^3) \Big]
  \leq e^{-2t^2/s_j} 
  \leq e^{-3/\eps^2}.
\]
Applying a union bound over at most $2^{2/\eps^2}$ possible cut queries $S\cup T$, 
we see that there exists $H'_j$ (it is in fact obtained with high probability)
such that for every cut query, the cut value in $H'_j$ is within additive $\gamma/(2\veps^3)$ from the cut value in $H_j$,
which in turn is within additive $\gamma\eps/2\cdot 1/\eps^4 = \gamma/(2\eps^3)$
from the cut value in $G_j$.
Hence, such $H'_j$ approximates all the cut values in $G_j$
sufficiently well for our intended application,
and Alice's sketch will thus encode $H'_j$ instead of $H_j$.
To simplify the description, let us include in $H'_j$ also edges of weight zero,
and then $H'_j$ has exactly $s_j$ edges (same as in $H_j$).
We further claim that the total edge-weight in $H'_j$ is at most $2/\eps^4$.
To see this, observe that 
(i) the total edge-weight in $H'_j$ (and similarly for $G_j$) 
is exactly twice the expected cut value of a random query in that graph;
and (ii) this expected cut value in $H'_j$ differs from the respective
expectation in $G_j$ by at most additive $\gamma/(2\veps^3)$.
It follows that the total edge-weight in $H'_j$ 
is at most $1/\eps^4$ larger than that in $G_j$,
which in turn is at most $1/\eps^4$.

The encoding of $H'_j$ has two parts, 
which describe separately the edges of $H'_j$ (without their weights), 
and their weights (assuming the edges are known).
Since $H'_j$ has $1/\eps^2$ vertices in each side, 
the number of possibilities for $s_j$ edges (without their weights)
among $\binom{2/\eps^2}{2} \leq 2/\eps^4$ vertex pairs 
is at most $\binom{2/\eps^4}{s_j}$.
Given the identity of $s_j$ edges in $H'_j$, 
the number of possibilities for their weights 
(recall the weights are integral and add up to at most $2/\eps^4$) 
is at most $\binom{s_j+2/\eps^4}{s_j} \leq \binom{4/\eps^4}{s_j}$.
We conclude that $H'_j$ can be encoded, on its two parts, 
using $O(\log \binom{4/\eps^4}{s_j})$ bits.

The second case $s_j > \gamma^2/(6\eps^4)$ is very simple,
and just encodes the original $G_j$ (instead of encoding $H_j$),
which trivially provides the value of every cut query inside $G_j$ exactly.
A straightforward encoding of $G_j$ takes $1/\eps^4$ bits.

Concatenating these encodings over all $j\in[\veps^2n/2]$,
yields a sketch that can approximate the value of all the needed cut queries 
(those that are inside a single $G_j$) within additive error $\gamma/\veps^3$.
It remains to bound the size of this sketch.
The number of $j$'s that fall into the second case $s_j > \gamma^2/(6\eps^4)$ 
is at most 
$\frac{\sum_j s_j}{\gamma^2/(6\eps^4)} = \frac{6\eps^4 s}{\gamma^2}$, 
and thus the total size of their encodings is at most $6s/\gamma^2$ bits.
For $j$'s that fall into the first case,
we use the fact $\binom{p}{k}\cdot \binom{p'}{k'} \leq \binom{p+p'}{k+k'}$,
and get that the total size of their encodings is at most
$ \sum_j O(\log \binom{4/\eps^4}{s_j}) 
  \leq  O(\log \binom{2n/\eps^2}{s}) 
  = O(s \log (\eps^{-2}n/s))
$ 
bits.
Altogether, there is a sketch of size $m=O(s(\gamma^{-2}+\log (\eps^{-2}n/s)))$ 
bits that encodes all the relevant cuts in $G$ within the desired accuracy.
Recalling that $\gamma$ is a constant and 
our sketch-size lower bound $m\ge \Omega(n/\eps^2)$
(from Lemma~\ref{lem:BobRecovers} and Theorem~\ref{thm:GHD}),
we conclude that the number of edges in $H$ is $s\geq \Omega(n/\eps^2)$,
which proves Theorem~\ref{thm:sparsifierLB}.

%%% Local Variables: 
%%% mode: latex
%%% TeX-master: "paper"
%%% End: 

\section{Symmetric Diagonally Dominant Matrices, ``For Each'' Model}
\label{sec:sdd-each}

In this section, we consider sketching SDD matrices with ``for each'' guarantee. That is, we want a sketch $\sk{A}$ of an $n \times n$ SDD matrix $A$, which can be used to produce a $(1 + \eps)$-approximation of $x^TAx$ for any fixed $x \in \mathbb{R}^{n}$ with constant probability (say $0.9$). Note that we can always use standard amplification argument (run the query on $O(\log n)$ independent copies of sketch, and return the median outcome) to boost the success probability to $1 - 1/n^{100}$.

Recall that we have a reduction from SDD matrices to Laplacian matrices (see Appendix \ref{app:sdd-each}).  
We thus only need to sketch Laplacian matrices, which we first do 
for \emph{cut queries} $x \in \{0, 1\}^n$ in Section~\ref{sec:cut-each},
and then for \emph{spectral queries} $x \in \R^n$ 
in Section~\ref{sec:spectral-each}
(for the main results, 
see Theorems~\ref{thm:upper} and~\ref{thm:improved-general}, respectively).

\subsection{Laplacian Matrices with Cut Queries}
\label{sec:cut-each}

Given a Laplacian matrix $L$, let $G = G(L) = (V, E, w)$  be the corresponding graph with $n=|V|$ vertices.
In this section we will construct a sketch for $G$ that can produce 
the weight of each \emph{cut} up to factor $1+\eps$ with constant probability.
Our work plan has three parts, each uses the previous one as a building block.
We start with an important special case in Section~\ref{sec:cut-special},
then extend it to all graphs with polynomial edge weights in Section~\ref{sec:cut-polyweight}, 
and finally extend that to all edge weights in Section~\ref{sec:cut-largeWeights}.

We then turn to proving a nearly matching lower bound in
Section~\ref{sec:nEpsLB}.
%, and an application of our sketch to  distributed minimum cut in Section~\ref{sec:min-cut}.

\subsubsection{Special Graphs and Special Queries} \label{sec:cut-special}
We start by sketching a class of special graphs with a special set of queries, which is sufficient to illustrate our basic ideas.

\begin{definition}[S1-graph]
We say an undirected weighted graph $G = (V, E, w)$ is an \emph{S1-graph} (reads ``simple type-$1$ graph'') if it satisfies the followings.
\begin{enumerate}
\item All edge weights are within factor $2$ of each other,  
i.e., there is $\gamma > \eps^2$ such that all $e \in E$ satisfy $w(e) \in [\gamma, 2 \gamma)$.

\item The expansion constant of $G$ is $\Gamma_G\geq\frac{1}{\eps}$.

%\item $w_{\max} / w_{\min} = \text{poly}(n)$.
\end{enumerate}
\end{definition}

\begin{algorithm}[t]
%\DontPrintSemicolon % Some LaTeX compilers require you to use \dontprintsemicolon instead
\KwIn{An S1-Graph $G = (V, E, w)$; a quality control parameter $\eps \in [1/n,1/30]$ }
\KwOut{A $(1+21\eps, 1/9)$-cut-sketch $\sk{G}$ of $G$ }
%$\alpha \gets \eps^{-\frac{5}{3}}$\;
Let $s=1/\eps$\;
$\sk{G} \gets \emptyset$\;
Add $\{\delta_u(G)\ |\ u \in V\}$ to $\sk{G}$\;

\For {$u \in V$}{
  $E_u \gets \{ (u, v)\ |\ v \in V \}$\;
  Sample (with replacement) $s$ edges from $E_u$, where each time the probability of sampling $e = (u, v)\in E_u$ is $p_e = {1}/{d_u(G)}$\;
%\tcc{Note that $\sum_{e\in E_u}p_e=1$}
Add the sampled $s$ edges to $\sk{G}$\;
}

\Return{$\sk{G}$}\;
\caption{{\bf Cut-S1}($G, \eps$)}
\label{alg:cutsketch-basic-simple}
\end{algorithm}

We consider a special set of cut queries where $w(S, \bar{S}) \le 5$.

The sketching algorithm for S1-graph and the special cut queries is described in Algorithm~\ref{alg:cutsketch-basic-simple}, which is fairly simple: we first add all weighted degrees of vertices to the sketch, and then for each vertex we sample a set of adjacent edges and store them in the sketch.

We first show the correctness of Algorithm~\ref{alg:cutsketch-basic-simple}. Let $Y_u^v$ be the random variable denoting the number of times edge $(u, v)$ is sampled at Line $6$ in Algorithm \ref{alg:cutsketch-basic-simple}. It is easy to see that
\begin{equation}
\label{eq:cut-b-1}
\E[Y_u^v] = \frac{s }{d_u(G)} \quad \text{and} \quad \var{Y_u^v} = s \left(1 - \frac{1}{d_u(G)}\right)\frac{1}{d_u(G)} \leq \frac{s }{d_u(G)}.
\end{equation}
Given a cut $(S,\bar{S})$ such that $|S|\leq |V|/2$ and $w(S,\bar{S})\leq 5$, we can approximate the cut weight $w(S,\bar{S})$ by the following estimator:
\begin{equation}
  \label{eq:cut-estimator}
  I_G = \sum_{u \in S} \Big[\delta_u(G)-\frac{d_u(G)}{s}\sum_{v \in S,\ (u,v)\in E} Y_u^v w(u,v)\Big].
\end{equation}

\begin{lemma}
  \label{lem:cutS1-Estimator}
 Let $G = (V, E, w)$ be an S1-Graph and $(S,\bar{S})$ be a cut of $G$ such that $|S|\leq |V|/2$ and $w(S,\bar{S})\leq 5$, then $I_G$ (defined in Equation (\ref{eq:cut-estimator})) is an unbiased estimator of $w(S,\bar{S})$. Furthermore, it gives a $(1 + 21\eps, 1/9)$-approximation to $w(S,\bar{S})$ for such cut $(S,\bar{S})$.
\end{lemma}

\begin{proof}
Since $\E[Y_u^v] = \frac{s}{d_u(G)}$ (by (\ref{eq:cut-b-1})), it follows that
$$\E[I_G] = \sum_{u \in S} \Big[\delta_u-\sum_{v \in S,\ (u,v)\in E} w(u,v)\Big]=w(S,\bar{S}).$$
We next analyze the variance of $I_G$. Recall that $\abs{\partial(A,B)}=\sum_{e=(u,v)\in E}\indic_{\{u\in A, v\in B\}}$ is the number of edges between sets $A$ and $B$.
\begin{eqnarray}
  && \var{I_G} \nonumber\\
   & = & \var{\sum_{u\in S}\frac{d_u(G)}{s} \sum_{v\in S,\ (u,v)\in E } Y_u^v w(u,v)} \nonumber \\
             & = &  \sum_{u\in S}\frac{(d_u(G))^2}{s^2} \sum_{v\in S,\ (u,v)\in E} \var{Y_u^v}(w(u,v))^2 \nonumber \\
             &\leq & \sum_{u\in S}\frac{(d_u(G))^2}{s^2} \sum_{v\in S,\ (u,v)\in E} \frac{s }{d_u(G)}(2\gamma)^2 \quad \quad \quad \quad \quad \quad \quad \quad \quad \quad (\text{by (\ref{eq:cut-b-1}) and $w(u,v)\in [\gamma,2\gamma)$}) \nonumber \\
             & = & 4\eps\gamma^2\sum_{u\in S}d_u(G) d_u(G(S))\quad \quad \quad   (\text{$s=1/\eps$ and $G(S)$ is the subgragh of $G$ induced by vertices in $S$}) \nonumber \\
             &\leq  & 4\eps\gamma^2|S|\sum_{u\in S}d_u(G) \quad \quad \quad \quad \quad \quad \quad \quad \quad \quad (\text{$d_u(G(S))\leq |S|$})\nonumber \\
             &\leq  & 4\eps\gamma^2|S|\cdot\Big[\abs{\partial(S,\bar{S})}+2|S|^2 \Big]\quad \quad \quad (\text{$\abs{\partial(S,\bar{S})}=\sum_{u\in S}d_u(G)-2\abs{\partial(S,S)}$ and $\abs{\partial(S,S)}\leq |S|^2$})\nonumber\\
             &\leq  & 4\eps^2\gamma^2(\abs{\partial(S,\bar{S})})^2\Big[1+2\eps^2 \abs{\partial(S,\bar{S})} \Big]\quad \quad \quad \quad \   (\text{$|S|\leq \eps \abs{\partial(S,\bar{S})}$})\nonumber\\
             &\leq  & 4\eps^2(w(S,\bar{S}))^2\Big[1+2\gamma \abs{\partial(S,\bar{S})} \Big]\quad \quad \quad \quad \quad \quad (\text{$\gamma \abs{\partial(S,\bar{S})}\leq w(S,\bar{S})$ and $\gamma\geq \eps^2$})\nonumber\\
             &\leq  & 44\eps^2(w(S,\bar{S}))^2.\quad \quad \quad \quad \quad \quad \quad \quad \quad \quad \quad \quad \quad (\text{$\gamma \abs{\partial(S,\bar{S})} \leq w(S,\bar{S})\leq 5$}) 
\end{eqnarray}
Now by a Chebyshev's inequality, with probability at least $8/9$, we have
$$
  \abs{I_G- w(S,\bar S))}
  \leq 3  \sqrt{44\eps^2 ( w(S,\bar S))^2}
  \leq 21 \eps\,  w(S,\bar S).
$$
\end{proof}

We summarize our result for S1-graph as below.

\begin{theorem}
%\label{thm:basic-simple}
There is a sketching algorithm which given an S1-graph, outputs a $(1+\eps, 1/9)
$-cut-sketch of size $\tilde{O}(n/\eps)$ for any cut $(S,\bar{S})$ such that $S\subseteq V$ and $w(S,\bar{S})\leq 5$.
\end{theorem}
\begin{proof}
The correctness immediately follows from Lemma~\ref{lem:cutS1-Estimator}. We only need to show the size of the sketch. Note that in Algorithm~\ref{alg:cutsketch-basic-simple} we only store all weighted degrees of vertices, and sample $1/\eps$ edges for each vertex. Thus the size of sketch can be bounded by $\tilde{O}(n/\eps)$.
\end{proof}

% Using the standard ``median-trick'' (i.e. run $\log(1/\delta)$ independent estimators of $I_G$ in (\ref{eq:cut-estimator}) and return the median of them), we can boost the success probability to $1 - \delta$ for any $\delta > 0$. We summarize our result for the S1-Graph $G$ in the following theorem.

% \begin{theorem}
% \label{thm:basic-simple}
% There is a sketching algorithm which given an S1-graph, outputs a $(1+\eps, \delta)$-cut-sketch of size $\tilde{O}(n/\eps)$ for any cut $(S,\bar{S})$ such that $S\subseteq V$ and $w(S,\bar{S})\leq 5$.
% \end{theorem}
% \begin{proof}
% The correctness immediately follows from Lemma~\ref{lem:cutS1-Estimator}. We only need to show the size of the sketch. Note that in Algorithm~\ref{alg:cutsketch-basic-simple} we only store all weighted degrees of vertices, and sample $1/\eps$ edges for each vertex. Thus the size of sketch can be bounded by $\tilde{O}(n/\eps)$.
% \end{proof}

\subsubsection{Graphs with Polynomial Weights}\label{sec:cut-polyweight}

We now show how to sketch general graphs with polynomial integer weights, say, in $[1, n^5]$, by extending Algorithm~\ref{alg:cutsketch-basic-simple}.  
In Section~\ref{sec:cut-largeWeights} we will extend the construction to graphs with general weights.

\bigskip

\noindent{\bf Sketching.}
Our sketching algorithm consists of two components. The first component is a standard $1.2$-cut sparsifier
(recall that a $(1+\eps)$-cut sparsifier is a sparse graph $H$ on the same vertex-set as $G$ that approximates every cut in $G$ within a factor of $1+\eps$).
We can use the construction of Bencz{\'u}r and Karger \cite{BK96},
or subsequent constructions \cite{SS11,BSS14,FHHP11,KP12}
(some of which produce a spectral sparsifier, which is only stronger);
any of these methods will produce a graph $H$ with $\tO(n)$ edges.

The second component is the main ingredient of the sketch, which preprocesses a given graph into a bunch of S1-graphs.  Let $\tilde C=\{1.4^i\mid 0\le i\le \log_{1.4}n^5\}$ be the set of size
$O(\log n)$ such that each cut value in $G$ is $1.4$-approximated by
some value $c\in \tilde C$.  For each value $c\in \tilde C$ we
construct a data structure $D_{c}$ as follows.
\begin{enumerate}
\item By scaling all the edge weights, let us assume $c=1$.
We discard all edges $e$ whose (scaled) weight $w(e)>5$. Note that those edges are too heavy to be included into any cut whose weight is at most $5$.

\item Apply the {\em importance sampling}: We sample each
(remaining) edge $e\in E$ independently with probability $p_e=
\minn{w(e)/\eps^2,1}$, and assign each sampled edge $e$ with a new weight
$\tilde w(e)= w(e)/p_e$.  Notice that $\tilde w(e)\in [\eps^2,5]$.
(It may be convenient to consider the non-sampled edges as having
weight $\tilde w(e)= 0$.)  Let $\tilde E$ be the set of
sampled edges. We further partition $\tilde E$ into $l=O(\log \frac{1}{\eps})$
classes according to the (new) edge weights, namely, $\tilde
E=L_1\cup\cdots\cup L_l$ where $L_i = \aset{e\in \tilde{E}:\ \tilde w(e)\in
  (5\cdot 2^{-i}, 5\cdot 2^{-i+1}] }$.

\item For each class $L_i$, we recursively partition the graph $(V, L_i)$ as
follows: For each connected component $P=(V_P,E_P)$ in $(V, L_i)$ which contains
a subset $S'\subset V_P$ of size $\card{S'}\leq \card{V_P}/2$ such that
$\abs{\partial(S',\bar{S'})}/\card{S'}< 1/\eps$ (where $\bar{S'}=V_P\setminus S'$), we replace $P$ with its two vertex-induced subgraphs $P(S')$ and $P(\bar{S'})$ in $(V, L_i)$, and store all edges in the cut $(S',\bar{S'})$ to $Q_i$. Once the recursive partitioning process is finished, denoting the
resulting components by $\calP_i$,  we store in the sketch all
the edges in $Q_{i}$ connecting different connected components of $\calP_i$.

\end{enumerate}
This preprocessing step is described in Algorithm~\ref{alg:cut-preprocessing}. Note that each component $P \in \calP_i$ is an S1-graph.

\begin{algorithm}[t]
%\DontPrintSemicolon % Some LaTeX compilers require you to use \dontprintsemicolon instead
\KwIn{A graph $G = (V, E, w)$ such that for any $e \in E$, $w(e) \in [1, n^5]$; a parameter $c > 0$; a quality control parameter $\eps \in [1/n,1/30]$}
\KwOut{
A set $\mathcal{P}$ of edge disjoint components of $G$ such that for each $P \in \mathcal{P}$, $\Gamma_P > 1/\eps$; and a graph $Q$ induced by the rest of the edges in $G$.
}
$\mathcal{P} \gets \emptyset$, $Q \gets \emptyset$\;
$\forall e \in E, w(e) \leftarrow w(e)/c$\;
Discard any edge $e \in E$ with $w(e)>5$\;
Sample each edge $e\in E$ independently with probability $p_e \gets \min\{w(e)/\eps^2,1\}$, and rescale its weight to $\tilde{w}(e) \gets w(e)/p_e$ if the edge is sampled. Let $\tilde{E}$ denote the set of sampled edges\;
Partition $\tilde{E}$ to $L_1\cup\cdots\cup L_l$ with $l=O(\log \frac{1}{\eps})$, where $L_i = \aset{e\in \tilde{E}:\ \tilde w(e)\in
  (5\cdot 2^{-i}, 5\cdot 2^{-i+1}] }$\;
%\* Important Sampleing
%*\
\ForEach{$L_i$}{
$\mathcal{P}_i \gets \{ (V, L_i) \}$, $Q_i \gets \emptyset$\;
\While {$\exists$ a connected component $P=(V_P,E_P) \in \mathcal{P}_i$ such that the expansion constant  $\Gamma_P < 1/\eps$} {
  Find an arbitrary cut $(S', \bar{S'})$ (where $|S'|\leq |V_P|/2$) in $P$ such that $\abs{\partial(S',\bar{S'})}/|S'| < 1/\eps$\;
  Replace $P$ with its two subgraphs $P(S')$ and $P(\bar{S'})$ in $\mathcal{P}_i$\;
  Add all edges in the cut $(S', \bar{S'})$ into $Q_i$\;
}
$\mathcal{P}=\mathcal{P}\cup \mathcal{P}_i$, $Q=Q\cup Q_i$\;
}
\Return{$(\mathcal{P}$, $Q$)}\;
\caption{{\bf Cut-Preprocessing}($G, c,\eps$)}
\label{alg:cut-preprocessing}
\end{algorithm}

\bigskip

\noindent{\bf Estimation.}
Given a query subset $S\subset V$,
we first use the graph sparsifier $H$ to compute $\tilde c$,
a $1.2$-approximation to the desired cut value $w(S,\bar S)$,
and then use the data structure $D_{c}$ to answer the query,
where $c\in \tilde C$ is an approximation to $\tilde c/(1.4)^2$,
such that $c \in [\tilde c/(1.4)^3,\tilde c/1.4]$. This implies that $w(S,\bar S)\in [c,4c] $. Thus, by rescaling $c$ to $1$, we only need to estimate the (resclaled) cut weight $w(S,\bar{S})$, which is between $1$ and $4$.

The contribution to the cut $(S,\bar{S})$ from edges in each class $L_i$ consists of two parts: (i) the total weight of cut edges
between $S$ and $\bar S$, i.e., $\sum_{e \in Q_i \cap \partial(S, \bar{S})} \tilde{w}(e)$; and (ii) the weight of cut edges inside each
component $P\in \calP_i$, which can be estimated using the sketches $\sk{P}$ for each $P\in \calP_i$. We construct the estimator as follows.
\begin{equation}  \label{eq:cut-I_p}
  \hat{w}(S,\bar{S}) = \sum_{i\in [l]}\Big(\sum_{P\in \calP_i} I_P
  +\sum_{e\in Q_i\cap\partial(S,\bar{S})}\tilde{w}(e)\Big),
\end{equation}
where $I_P$, defined in Equation~(\ref{eq:cut-estimator}), is the estimator of the cut weight in the component $P\in \calP_i$.

In the rest of this section we show $\hat{w}(S,\bar{S})$ is a $1+O(\eps)$-approximation of the (scaled) cut weight $w(S,\bar{S})$ with constant probability. The following lemma guarantees that the importance sampling step preserves the cut weight.

\begin{lemma}
  \label{lem:cut-preprocess1}
 Let $\tilde{G}=(V,\tilde{E},\tilde{w})$ be the graph after the important sampling step (Line~4 of Algorithm~\ref{alg:cut-preprocessing}). Then for any $S \subset V$ such that $w(S,\bar{S})\in [c,4c]$, with probability at least $8/9$, $|\tilde{w}(S,\bar{S})-w(S,\bar{S})|\leq 3\eps w(S,\bar{S})$, where $w(S,\bar{S})$ and $\tilde{w}(S,\bar{S})$ are the cut weight of $(S,\bar{S})$ in $G$ and $\tilde{G}$, respectively.
\end{lemma}
\begin{proof}
Since $w(S,\bar{S})\in [c,4c]$, after
rescaling (Line~2) $c$ to $1$, the cut weight $w(S,\bar S)$ becomes a value in $[1, 4]$. Also note that the discarded edges (Line~3) will not affect the final cut estimation since they can not be part of the cut $(S,\bar S)$, given that their weights are more than $5$.

We now analyze the effect of the importance sampling step (Line~4) on the weight
of the cut. It is easy to see that $\EX[\tilde w(S,\bar S)] = w(S,\bar S)$, since every edge $e$ that was not discarded satisfies $\EX[\tilde w(e)] = w(e)$ (and more generally, it is a Horvitz-Thompson estimator),
and its variance is
\[
  \var {\tilde{w}(S,\bar S)}
   = \sum_{e\in \partial(S,\bar S)} \var{\tilde {w}_e}
  = \sum_{e\in \partial(S,\bar S)} w(e)^2/p_e-w(e)^2
  \leq \sum_{e\in \partial(S,\bar S)} \eps^2 w(e)
  = \eps^2 w(S, \bar S),
\]
where the inequality is verified separately for $p_e=1$ and for $p_e=w(e)/\eps^2$.
Thus, by a Chebyshev's inequality, with probability at least $8/9$, we have
$\abs{\tilde w(S,\bar S) - w(S,\bar S)}
  \leq 3 \eps \sqrt{w(S, \bar S)}
  \leq 3 \eps \cdot w(S, \bar S)$.
\end{proof}

Then, it suffices to show that the constructed estimator $\hat{w}(S,\bar{S})$ can approximate $\tilde{w}(S,\bar{S})$ well.
\begin{lemma}
  \label{lem:cut-preprocess2}
For any $S \subset V$ such that $w(S,\bar{S})\in [1,4]$, with probability at least $8/9$, $|\hat{w}(S,\bar{S})-\tilde{w}(S,\bar{S})|\leq 21\eps \tilde{w}(S,\bar{S})$.
\end{lemma}
\begin{proof}
First, it is clear that our estimator $\hat{w}(S,\bar{S})$ captures exactly the contribution of cut edges in all $Q_i$'s $(i\in[l])$, since the sketch stores all edges in each $Q_i$. We thus only need to show the contribution of cut edges in each $P\in \calP_i (i\in[l])$, denoted as $\tilde{w}_i(S\cap V_P,\bar{S}\cap V_P)$, can be approximated well by the estimator $I_P$ (Equation~(\ref{eq:cut-estimator})).

%Notice that in Algorithm~\ref{alg:cut-preprocessing}, after the important sampling step (Line~4), all sampled edges have weight $\tilde w(e)\in [\eps^2,5]$.
Let $\gamma_i= 5\cdot 2^{-i}\geq \eps^2$ (since $\tilde w(e)\in [\eps^2,5]$), and then all edges $e\in L_i$ have weight $\tilde w(e)\in [\gamma_i, 2\gamma_i]$. Also note that by the stopping condition of the recursive partition (Line~8 of Algorithm~\ref{alg:cut-preprocessing}), each returned $P\in \calP_i$ satisfies $\Gamma_P\geq 1/\eps$. Therefore each $P\in \calP_i$ is an S1-graph.

Now observe that
$$\tilde{w}_i(S\cap V_P,\bar{S}\cap V_P)\leq \tilde w(S,\bar S)\leq (1+3\eps)\, w(S,\bar S)\le 5,$$ together with the fact that $P$ is an S1-graph, we know, according to Lemma~\ref{lem:cutS1-Estimator}, that $I_P$ defined by Equation~(\ref{eq:cut-estimator}) is an unbiased estimator of $\tilde{w}_i(S\cap V_P,\bar{S}\cap V_P)$ with the variance $\var{I_P}\leq 44\eps^2(\tilde{w}_i(S\cap V_P,\bar{S}\cap V_P))^2$.
It follows that $\hat{w}(S,\bar{S})$ is an unbiased estimator of $\tilde{w}(S,\bar{S})$, since
\begin{eqnarray}
\EX[\hat{w}(S,\bar{S})] &=& \sum_{i\in [l]}\Big(\sum_{P\in \calP_i} \EX[I_P]
  +\sum_{e\in Q_i\cap\partial(S,\bar{S})}\tilde{w}(e)\Big)=\sum_{i\in [l]}\Big(\sum_{P\in \calP_i} \tilde{w}_i(S\cap V_P,\bar{S}\cap V_P)
  +\sum_{e\in Q_i\cap\partial(S,\bar{S})}\tilde{w}(e)\Big)\nonumber\\
&=&\tilde{w}(S,\bar{S}).
\end{eqnarray}
And the variance
\begin{eqnarray}
  \var{\hat{w}(S,\bar{S})}  & = & \sum_{i\in [l]} \sum_{P\in \calP_i} \var{I_P}\leq 44\eps^2\sum_{i\in [l]}\sum_{P\in \calP_i} (\tilde{w}_i(S\cap V_P,\bar{S}\cap V_P))^2\nonumber\\
  &\leq& 44\eps^2\Big(\sum_{i\in [l]}\sum_{P\in \calP_i} \tilde{w}_i(S\cap V_P,\bar{S}\cap V_P)\Big)^2=44\eps^2 (\tilde{w}(S,\bar{S}))^2.
\end{eqnarray}
Thus by a Chebyshev's inequality, with probability at least $8/9$, we have that
$
  \abs{\hat{w}(S,\bar{S})- \tilde w(S,\bar S))}
  \leq 3  \sqrt{44\eps^2 (\tilde w(S,\bar S))^2}
  \leq 21 \eps\, \tilde w(S,\bar S)
$.
\end{proof}

After the preprocessing (Algorithm~\ref{alg:cut-preprocessing}), we only need to run Algorithm~\ref{alg:cutsketch-basic-simple} on each resulting S1-graph and on special queries (i.e., $w(S, \bar{S}) \le 5$). The overall algorithm for general graph with polynomial weights is described in Algorithm~\ref{alg:cut-sketch-basic-general}.

\begin{algorithm}[t]
\KwIn{$G = (V, E, w)$ with all weights in $[1, n^5]$; a quality control parameter $\eps \in [1/n,1/30]$}
\KwOut{A $(1+27\eps, 2/9)$-cut-sketch $\sk{G}$ of $G$}
$\sk{G} \gets \emptyset$\;
Build a $1.2$-cut sparsifier $H$ of $G$, and add $H$ into $\sk{G}$\;
Let $\tilde C=\{1.4^i\mid 0\le i\le \log_{1.4}n^5\}$\;
\ForEach{$c\in \tilde C$}{
$D_c\gets \emptyset$\;
$\{\mathcal{P},Q\}\gets${\bf Cut-Preprocessing}($G, c$), and add $Q$ into $D_c$\;
\ForEach{$P\in \mathcal{P}$}{
$\sk{P}\gets${\bf Cut-S1}($P, \eps$), and add $\sk{P}$ in $D_c$\;
}
Add $D_c$ into $\sk{G}$\;
}
\Return{$\sk{G}$}\;
\caption{{\bf Cut-Basic}($G, \eps$)}
\label{alg:cut-sketch-basic-general}
\end{algorithm}

The following theorem summarizes the results of this section.

\begin{theorem}
\label{thm:cut-upper}
Given a weighted graph $G=(V,E,w)$ on $n$ vertices, where the non-zero
weights are in the range $[1,W]$ with $W=n^5$, and $1/n\le \eps\le 1/30$, there
exists a cut sketch of size $\tilde O(n/\eps)$
bits. Specifically, for every query $S\subset V$, the sketch produces
a $1+O(\eps)$ approximation to $w(S,\bar S)$, with probability at
least $7/9$.
\end{theorem}

\begin{proof}
Given a query subset $S\subset V$,
we use the data structure $D_{c}$ with $c\in[\tilde c/(1.4)^3,\tilde c/1.4]$, where
$\tilde c$ is a $1.2$-approximation to the cut weight $w(S,\bar S)$. Thus we have $w(S,\bar S)\in [c,4c]$. Since we rescale $c$ to $1$, we here show the constructed estimator is $1+O(\eps)$ approximation of the (sclaled) cut weight $w(S,\bar{S})$.

Then, by Lemma~\ref{lem:cut-preprocess1} and~\ref{lem:cut-preprocess2},
we have that with probability at least $7/9$,
the estimator $\hat{w}(S,\bar{S})$ (using the data structure $D_c$)
is a $1+O(\eps)$ approximation of $\tilde w(S,\bar S)$,
which in turn is a $1+O(\eps)$ approximation of $w(S,\bar S)$. More precisely,
\[
  \abs{\hat{w}(S,\bar{S})-w(S,\bar S)}
  \leq 3\eps\, w(S,\bar S) + 21 \eps (1+3\eps)\, w(S,\bar S)
  \leq 27 \eps\, w(S,\bar S).
\]

We next bound the sketch size.  The sparsifier $H$ has
$\tO(n)$ edges.  By construction, we have $O(\log n)$ possible cut
values $\tilde C$ and for each one, we have $l=O(\log
\tfrac{1}{\eps})\leq O(\log n)$ edge weight classes.  For each weight
class $L_i$, the sketch stores:
\begin{itemize}
\item $O(\tfrac{1}{\eps}n\log n)$
edges in $Q_i$: each step in the recursive
partition contributes $\abs{\partial(S',\bar{S'})}/\card{S'}<
1/\eps$ edges per vertex in $S'$, and each vertex appears in the
smaller subset $S'$ at most $\log n$ times.
\item At most $n/\eps$
sampled edges: for each non-isolated vertex we sample
$s=1/\eps$ incident edges.
\end{itemize}
Summing up, we have $\tO(n/\eps)$ edges,
each requiring $O(\log n)$ bits. Therefore the size of our cut sketch is $\tilde O(n/\eps)$ bits.

\end{proof}

%------------------------------------------------------------------
\subsubsection{Graphs with General Edge Weights}
\label{sec:cut-largeWeights}

We now build on the results of Section~\ref{sec:cut-polyweight} for polynomial weights to show the upper bound for general edge weights.
That is, assume that there is a sketching algorithm,
which we shall call the ``basic sketch'',
for the case where all edge weights are in a polynomial range,
say for concreteness $[1,n^5]$
(which by scaling is equivalent to the range $[b,n^5b]$ for any $b>0$),
which uses space $\tO(n/\eps)$.
We may assume the success probability of this sketch is at least $1-1/n^8$,
e.g., by using the standard ``$O(\log n)$ repetitions and then taking the median'',
thereby increasing the sketch size by at most $O(\log n)$ factor.
As before, we may assume $\eps>1/n$, as
otherwise the theorem is trivial.

\bigskip

\noindent{\bf Sketching.}
The sketch has two components: (i)
the first component is essentially a maximum-weight spanning tree $T$
computed using Kruskal's algorithm; see Algorithm~\ref{alg:cut-MST}; and (ii) the second component is a set of cut sketches of the graphs reduced according to those tree edges in $T$, which are constructed via the sketching algorithm introduced in Section~\ref{sec:cut-polyweight}; see Algorithm~\ref{alg:cut-sketch-largeweight}, where $G_j, G'_j$ are defined.

\begin{algorithm}[t]
%\DontPrintSemicolon % Some LaTeX compilers require you to use \dontprintsemicolon instead
\KwIn{A graph $G = (V, E, w)$}
\KwOut{A maximum-weight spanning tree $T$ }
%$\alpha \gets \eps^{-\frac{5}{3}}$\;
$E_T \gets \emptyset$, $T \gets (V,E_T)$\;
$\pi \gets$ an order of edges $e\in E$ in decreasing weight (break ties arbitrarily if any)\; %, denoted as $\pi=\{e\}$\;
\ForEach{$e \in E$ in the ordering $\pi$}
{Add $e$ into $E_T$ if this will not introduce a cycle in $T$, where edges in $E_T$ are sorted in the order of insertion (which is also their ordering according to $\pi$)\;}
\Return{$T$}\;
\caption{{\bf MST}($G$)}
\label{alg:cut-MST}
\end{algorithm}

\begin{algorithm}[t]
%\DontPrintSemicolon % Some LaTeX compilers require you to use \dontprintsemicolon instead
\KwIn{A graph $G = (V, E, w)$ with general weights; a quality control parameter $\eps\in [1/n,1/30]$}
\KwOut{A $(1+\eps, \delta)$-cut-sketch $\sk{G}$ of $G$}
%$\alpha \gets \eps^{-\frac{5}{3}}$\;
$\sk{G} \gets \emptyset$\;
$T\gets${\bf MST}($G$), denoted as $T=\{e_1,...,e_{n-1}\}$, and add $T$ into $\sk{G}$\;
\For{$j=1$ to $n-1$}
{\If {there is no earlier iteration $k<j$
with $w(e_k)/w(e_j)<2$ that is sketched and stored}
{Remove all edges $e\in E$ of weight $w(e) < w(e_j)/n^3$\;
Change all edges $e\in E$ of weight $w(e)\geq n^2\cdot w(e_j)$ to have
infinite weight, and denote the resulting graph as $G_j$\;
Contract all edges of infinite weight in $G_j$ (keeping parallel edges and removing
self-loops), and denote the resulting graph as $G_j'$\;
}
\ForEach{connected component $P$ of $G'_j$ of size at least $2$}{
$\sk{P}\gets${\bf Cut-Basic}($P, \eps$), and add $\sk{P}$ into $\sk{G}$\;
\tcc{(Note that $\forall e$ in $P$, $w(e) \in [n^{-3}\cdot w(e_j), n^2\cdot w(e_j)]$.)}
}

}
\Return{$\sk{G}$}\;
\caption{{\bf Cut-Sketch}($G,\eps$)}
\label{alg:cut-sketch-largeweight}
\end{algorithm}

\bigskip

\noindent{\bf Estimation.}
Given a query subset $S\subset V$,
find the smallest $j\in[n-1]$ such that $e_j$ is in the cut $\partial(S,\bar S)$;
such $j$ exists because $\aset{e_1,\ldots,e_{n-1}}$ forms a spanning tree (we assume the graph is connected, since otherwise we can sketch each connected component separately).
We further show in Lemma \ref{lem:heaviest} below that $e_j$
is the heaviest edge in this cut, hence $w_{G}(S,\bar S) / w(e_j) \in [1,n^2]$.
Now find the largest $k\leq j$ for which we sketched and stored in $\sk{G}$;
by construction $w(e_k)/w(e_j)\in[1,2)$. Lemma \ref{lem:GkApprox} below proves that the cut values in $G$ and in $G_k$
are almost the same.

Next, compute the connected components of the graph $(V,\aset{e_1,\ldots,e_k})$,
and observe they must be exactly the same as the connected components of $G_k$.
Obviously, the value of the cut $(S,\bar S)$ in $G_k$
is just the sum, over all connected components $P$ in $G_k$,
of the contribution to the cut from edges inside that component,
namely $w_{G_k}(S\cap V_P, \bar S\cap V_P)$ with $V_P$ as the vertex set of $P$.
Recall that $G'_k$ has essentially the same cuts as $G_k$
and we can thus estimate each such term $w_{G_k}(S\cap V_P, \bar S\cap V_P)$
using the sketch we prepared for $G'_k$
(more precisely, using the sketch of the respective component $V'$ of $G'_k$,
unless $\card{V'}=1$ in which case that term is trivially $0$).
To this end, we need to find out which vertices of $G_k$
were merged together to form $G'_k$,
which can be done using $e_1,\ldots,e_{n-1}$ as follows.
Find the largest $k^*$ such that $w(e_{k^*}) \geq n^2\cdot w(e_k)$,
and compute the connected components of the graph $(V,\aset{e_1,\ldots,e_{k^*}})$.
Lemma \ref{lem:ConnComp} below proves that these connected components
(or more precisely the partition of $V$ they induce)
are exactly the subsets of vertices that are merged in $G_k$ to create $G'_k$.
Now that knowing the vertex correspondence between $G_k$ and $G'_k$,
we estimate the cut value $w_{G_k}(S\cap V', \bar S\cap V')$
by simply using the estimate for the corresponding cut value in $G'_k$,
where the latter is obtained using the basic sketch prepared for $G'_k$. Thus, the estimator is
\begin{equation}\label{eq:cut-estimator-largeweight}
\hat{w}(S,\bar{S})=\sum_{P\in G'_k} \hat{w}(S\cap V_P,\bar{S}\cap V_P),
\end{equation}
where $\hat{w}(S\cap V_P,\bar{S}\cap V_P)$ is defined in Equation~(\ref{eq:cut-I_p}) and can be computed via using the sketch $\sk{P}$.

To show the performance of the estimator $\hat{w}(S,\bar{S})$, we first present three lemmas.

\begin{lemma} \label{lem:heaviest}
Fix $S\subset V$ and let $e'\in E$ be the first edge,
according to the ordering $\pi$, that in the cut $\partial(S,\bar S)$.
Then this $e'$ is the first edge in the sequence $e_1,\ldots,e_{n-1}$
that in the cut $\partial(S,\bar S)$.
% Let $e_1,\ldots,e_{n-1}$ be the edges computed by the sketching algorithm.
% Let $S\subset V$ and find the smallest $j\in[n-1]$
% such that $e_j$ crosses the cut $(S,\bar S)$.
% Then $e_j$ is a heaviest edge in this cut.
\end{lemma}
\begin{proof}
Let $e'\in E$ be the first edge, according to the ordering $\pi$,
that in the cut $\partial(S,\bar S)$.
Clearly, $e'$ is the heaviest edge in this cut.
Now observe that in the construction of $T$ (i.e., $e_1,\ldots,e_{n-1}$),
when $e'$ is considered, $T$ has no edge between $S$ and $\bar S$,
hence the endpoints of $e'$ lie in different connected components,
and $e'$ must be added to $T$.
\end{proof}
\begin{lemma} \label{lem:GkApprox}
Consider a query $S\subset V$ and let $k\in[n-1]$ be the value computed
in the estimation algorithm.
Then the ratio between the value of $w(S,\bar S)$ in the graph $G_k$
and that in the graph $G$ is in the range $[1-\tfrac1n,1]\subset [1-\eps,1]$,
formally
\[
   1-\tfrac1n \leq \frac{w_{G_k}(S,\bar S)}{w_{G}(S,\bar S)} \leq 1.
\]
\end{lemma}
\begin{proof}
The edges in $G_k$ are obtained from the edges of $G$, by either
(1) removing edges $e$ whose weight is $w(e)<w(e_k)/n^3$; or
(2) changing edges $e$ with $w(e)\ge n^2\cdot w(e_k)$ to have infinite weight.
The first case can only decrease any cut value, %including $w(S,\bar S)$,
while the second case can only increase any cut value.

Recall that the estimation process finds $j$ such that
$w_{G}(S,\bar S)/w(e_j) \in [1,n^2]$, and then finds $k\leq j$,
which we said always satisfies $w(e_k)/w(e_j) \in [1,2)$.
Thus, $w_{G}(S,\bar S)/w(e_k) \in (\half,n^2]$.
So one direction of the desired inequality follows by observing
that edges in $G$ that fall into case (1) have the total weight at most
\[
  \binom{n}{2} w(e_k)/n^3
  \leq \frac{2}{n} w(e_k)
  \leq \frac{1}{n} w_{G}(S,\bar S).
\]
The other direction follows by observing that edges $e$ that fall into case (2)
have (in $G$) weight $w(e) > n^2\cdot w(e_k) \ge w_{G}(S,\bar S)$,
and therefore do not belong to the cut $(S,\bar S)$.
\end{proof}

\begin{lemma} \label{lem:ConnComp}
Fix $w^*>0$, let $E^*=\aset{e\in E: w(E) \geq w^*}$,
and find the largest $i^*\in[n-1]$ such that $w(e_{i^*}) > w^*$.
Then the graphs $(V,E^*)$ and $(V,\aset{e_1,\ldots,e_{i^*}})$
have exactly the same connected components
(in terms of the partition they induce of $V$).
\end{lemma}
\begin{proof}
It is easy to see that executing our construction of $T$ above on the set $E^*$
gives the exact same result as executing it for $E$ but stopping once we reach
edges of weight smaller than $w^*$.
The latter results with the edges $e_1,\ldots,e_{i^*}$,
while the former is clearly an execution of Kruskal's algorithm,
i.e. computes a maximum weight forest in $E^*$.
\end{proof}

The next two lemmas show the performance of Algorithm~\ref{alg:cut-sketch-largeweight}.
\begin{lemma} \label{lem:combine}(\textbf{Accuracy Guarantee})
With high probability, $\hat{w}(S,\bar{S})$ defined in Equation~(\ref{eq:cut-estimator-largeweight}) is $1+O(\eps)$ approximation of $w_G(S,\bar{S})$.
\end{lemma}
\begin{proof}
Lemma~\ref{lem:heaviest} and \ref{lem:GkApprox} show that $(1-\eps)w_G(S,\bar{S})\leq w_{G_k}(S,\bar{S})\leq w_G(S,\bar{S})$. Lemma~\ref{lem:ConnComp} show that $w_{G_k}(S,\bar{S})$ is the same as the cut weight of $(S,\bar{S})$ in $G_k'$. Since we build the sketch of each connected component of size at least $2$ in $G_k'$, we can obtain a $1+O(\eps)$ approximation, \whp (say $1-1/n^8$), of the contribution to cut weight of each component in $G_k'$, i.e., the estimator $\hat{w}(S\cap V_P,\bar{S}\cap V_P)$ defined in Equation~(\ref{eq:cut-I_p}). Since the number of component is $O(n)$ (because they correspond to disjoint subsets of $V$), applying the union bound over the events of an error in any of the
basic estimates used along the way, we know that the constructed estimator $\hat{w}(S,\bar{S})$ defined in Equation~(\ref{eq:cut-estimator-largeweight}) is $1+O(\eps)$ approximation of $w_G(S,\bar{S})$, \whp. (The union bound is applicable because these basic sketch are queried
in a non-adaptive manner, or alternatively, because we make at most
one query to every basic sketch that is constructed independently of the others.)
\end{proof}

\begin{lemma} \label{lem:size}(\textbf{Size Analysis})
The total size of the sketch is at most $\tO(n/\eps \cdot \log\log W)$,
where we assume all non-zero edge weights are in the range $[1,W]$.
\end{lemma}
\begin{proof}
The first component of the sketch is just a list of $n-1$ edges with
their edge weights, hence its size is $O(n\log (\log W/\eps))$ (we can
store a $1+\eps/2$ approximation to each weight using space $\log
(\log W/\eps)$).

The second component of the sketch has $n-1$ parts,
one for each $G'_j$ where $j\in[n-1]$.
For some of these $j$ values, we compute and store the basic sketch
for every connected components of $G'_j$ that is of size at least $2$.
Denoting by $n_j$ the number of vertices in $G'_j$,
and by $m_j$ the number of connected components in $G'_j$,
the storage requirement for each $G'_j$
is at most $\tO(\tfrac{n_j-m_j}{\eps})$,
because each connected component of size $s\ge 2$
requires storage $\tO(\tfrac{s}{\eps}) \leq \tO(\tfrac{s-1}{\eps})$,
and these sizes (the different $s$ values) add up to at most $n$.

Denote the values of $j$ for which we do store a basic sketch for $G'_j$
by $j_1<j_2<\cdots<j_p$,
where by construction $w(e_{j_i}) / w(e_{j_{i+1}}) \geq 2$.
Summing over these values of $j$, the second component's storage requirement
is at most
\begin{equation} \label{eq:storage1}
\tO(\tfrac{1}{\eps} \sum_{i\in[p]}(n_{j_i}-m_{j_i}) ).
\end{equation}
To ease notation, let $M= 5\log_2 n$,
and consider the graphs $G'_{j_i}$ and $G'_{j_{i+M}}$ for some $i\in[ p-M-1]$.
Observe that every edge in $G'_{j_i}$ has weight at least
$w(e_{j_i})/n^3 \geq 2^M \cdot w(e_{j_{i+M}})/n^3 = n^2 \cdot w(e_{j_{i+M}})$
(because edges of smaller weight are removed);
thus, in $G_{j_{i+M}}$, these same edges have infinite weight,
and then to create the reduced form $G_{j_{i+M}}$, these edges are contracted.
It follows from this observation that every connected component in $G'_{j_i}$
becomes in $G_{j_{i+M}}$ a single vertex, hence $n_{j_{i+M}} \leq m_{j_i}$
(we do not obtain equality since additional contractions may occur).
Using this last inequality, for every $i^*\in[M]$,
we can bound the following by a telescopic sum
\[
  \sum_{i=i^*,i^*+M,i^*+2M,\ldots} (n_{j_i}-m_{j_i})
  \leq \sum_{i=i^*,i^*+M,i^*+2M,\ldots} (n_{j_i}-n_{j_{i+M}})
  \leq n_{j_{i^*}}
  \leq n,
\]
and therefore
\[
  \sum_{i\in[p]}(n_{j_i}-m_{j_i}) )
  \leq \sum_{i^*\in[M]} \sum_{i=i^*,i^*+M,i^*+2M,\ldots} (n_{j_i}-m_{j_i})
  \leq M\cdot n.
\]
Plugging this last inequality into \eqref{eq:storage1},
we obtain that the second component's storage requirement is at most
$M\cdot \tO(n/\eps)$, which is still bounded by $\tO(n/\eps)$.
\end{proof}

Based on Lemma~\ref{lem:combine} and \ref{lem:size}, we conclude the following theorem.
\begin{theorem}
\label{thm:upper}
Fix an integer $n$ and $\eps\in(1/n,1/30)$.
Then every $n$-vertex graph $G=(V,E,w)$ with edge weights in the range $[1,W]$
admits a cut sketch of size $\tilde O(n\eps^{-1}\cdot \log\log W)$ bits
with the ``for each'' guarantee. 
Specifically, for every query $S\subset V$ (equivalently, $x \in \{0,1\}^n$),
the sketch can produce \whp a $1+O(\eps)$ approximation to $w(S,\bar S)$.
\end{theorem}

%------------------------------------------------------------------
\subsubsection{A Tight Lower Bound}
\label{sec:nEpsLB}

The following theorem shows that our sketch from Theorem \ref{thm:upper}
achieves optimal space up to a poly-logarithmic factor,
even for unweighted graphs.

\begin{theorem} \label{thm:nEpsLB}
Fix an integer $n$ and $\eps\in[2/n,1/2]$. Suppose $\sk{\cdot}$
is a sketching algorithm that outputs at most $s=s(n,\eps)$ bits,
and $\est$ is an estimation algorithm, such that together
for every $n$-vertex graph $G$,
%and subset $S\subset V$,
%with probability at least $9/10$
%the estimation procedure is correct up to factor $1+\eps$, i.e.,
\[
  \forall S\subset V, \qquad
  \Pr\Big[\est(S,\sk{G})\in (1\pm \eps)\cdot\card{\partial(S,\bar S)}\Big]\ge 9/10,
\]
where $\partial(S,\bar S) = \{(u,v)\in E\ | u\in S, v\in \bar S\}$. 
Then $s \ge \Omega(n/\eps)$.
\end{theorem}
\begin{proof}
We will show how to encode a bit-string of length $l= n/(8\eps)$ into a
graph, so that, given its sketch $\sk{G}$, one can reconstruct any bit
of the string with constant probability.
Standard information-theoretical argument would
then imply that $s\geq\Omega(l) = \Omega(n/\eps)$.

Given a string $x\in \{0,1\}^l$, we embed it into a bipartite graph
$G$ on with $n/2$ vertices on each side,
and vertex degrees bounded by $D= 1/(4\eps)$ as follows.
Partition the vertices on each side into disjoint blocks of $D$,
and let the $i$-th block on the left side and on the right side
form a (bipartite) graph which we call $G_i$, for $i=1,\ldots,n/(2D)$.
Then partition the string $x$ in $n/(2D)$ blocks,
each block is of length $D^2$ and describes the adjacency matrix
of some bipartite $G_i$.

We now show that evaluating a bit from the string $x$ corresponds to testing
the existence of some edge $(u,v)$ from some graph $G_i$, which we can do
using the $1+\eps$ approximating sketch only.
Formally, let $\delta(S)$ be the cut
value of the set $S$, i.e., $\card{\partial(S,\bar S)}$,
and observe that
\[
  \delta(\{u\})+\delta(\{v\})-\delta(\{u,v\}) =
  \begin{cases}
    2 & \text{if $(u,v)$ is an edge in $G$;} \\
    0 & \text{otherwise}.
  \end{cases}
  \]
Since the considered values of $\delta(\cdot)$ are bounded by $D$,
the sketch estimates each such value with additive error at most $\eps D=1/4$,
which is enough to distinguish between the two cases.
Furthermore, since we query the sketch only $3$ times, the
probability of correct reconstruction of the bit is at least $7/10$.
The lower bound follows.
\end{proof}

%%% Local Variables: 
%%% mode: latex
%%% TeX-master: "paper"
%%% End: 

\subsection{Laplacian Matrices with Spectral Queries}
\label{sec:spectral-each}

In this section, we construct sketches for a Laplacian matrix with spectral queries. 
We first design sketches of size $\tilde{O}(n / \eps^{5/3})$,
in Section~\ref{sec:basic},
and then improve it to size $\tilde O( {n}/{\eps^{8/5}} )$ 
in Section~\ref{sec:improve}.
In each of these, we will start with an algorithm for a class of special graphs, and then extend it to general graphs.

\subsubsection{A Basic Sketching Algorithm}
\label{sec:basic}
In this section, we described a sketch of size $\tilde{O}(n / \eps^{\frac{5}{3}})$ for a Laplacian matrix with spectral queries.  Let $\alpha = c_\alpha \eps^{-\frac{5}{3}}$ be a parameter we will use in this section, where $c_\alpha > 0$ is a large enough constant.
We will start with an algorithm for a class of special graphs, and then extend it to general graphs.

\paragraph{Special Graphs}
\label{sec:basic-simple}
In this section we consider a class of special graphs, defined as follows.

\begin{definition}[S2-graph]
We say an undirected weighted graph $G = (V, E, w)$ is an \emph{S2-graph} (reads ``simple type-$2$ graph'') if it satisfies the followings.
\begin{enumerate}
\item All weights $\{w(e)\ |\ e \in E\}$ are within a factor of $2$,  i.e. for any $e \in E$, $w(e) \in [\gamma, 2 \gamma)$ for some $\gamma > 0$.

\item The Cheeger's constant $h_G > \alpha\eps^2 = c_\alpha \eps^{\frac{1}{3}}$.

%\item $w_{\max} / w_{\min} = \text{poly}(n)$.
\end{enumerate}
\end{definition}

Let $\mathcal{S}(G) = \{v\in V\ |\ \delta_v \leq \gamma\alpha\}$, $\mathcal{L}(G) = \{v \in V\ |\ \delta_v > \gamma\alpha\}$. For $u\in \mathcal{L}(G)$, let $\delta_u^{\L}(G) = \sum_{v\in\mathcal{L}(G)}w(u,v)$. We will omit ``$(G)$" when there is no confusion. The algorithm for sketching S2-graph is described in Algorithm~\ref{alg:sketch-basic-simple}.  When we say ``add an edge to the sketch" we always mean ``add the edge together with its weight".

\begin{algorithm}[t]
%\DontPrintSemicolon % Some LaTeX compilers require you to use \dontprintsemicolon instead
\KwIn{An S2-Graph $G = (V, E, w)$; a quality control parameter $\eps$}
\KwOut{a $(1+\eps, 0.001)$-spectral-sketch $\sk{G}$ of $G$ }
%$\alpha \gets \eps^{-\frac{5}{3}}$\;
$\sk{G} \gets \emptyset$\;
Add $\{\delta_u\ |\ u \in V\}$ to $\sk{G}$\;
\For {$u \in \mathcal{S}$}{
  Add all of $u$'s adjacent edges to $\sk{G}$\;
}

\For {$u \in \mathcal{L}$}{
  Add $\delta_u^\L$ to $\sk{G}$\;
  $E_u \gets \{ (u, v)\ |\ v \in \L \}$\;
  Sample (with replacement) $\alpha$ edges from $E_u$, where each time the probability of sampling $e = (u, v)\in E_u$ is $p_e = {w(e)}/{\delta_u^\L}$\;
%\tcc{Note that $\sum_{e\in E_u}p_e=1$}
Add the sampled $\alpha$ edges to $\sk{G}$\;
}

\Return{$\sk{G}$}\;
\caption{{\bf Spectral-S2}($G, \eps$)}
\label{alg:sketch-basic-simple}
\end{algorithm}

Let $Y_u^v$ be the random variable denoting the number of times edge $(u, v)$ is sampled at Line $8$ in Algorithm \ref{alg:sketch-basic-simple}. It is easy to see that
\begin{equation}
\label{eq:b-1}
\E[Y_u^v] = \frac{\alpha w(u,v)}{\delta_u^\L} \quad \text{and} \quad \var{Y_u^v} = \alpha \left(1 - \frac{w(u,v)}{\delta_u^\L}\right)\frac{w(u,v)}{\delta_u^\L} \leq \alpha \frac{w(u,v)}{\delta_u^\L}.
\end{equation}
Given a vector $x \in \mathbb{R}^{n}$, we use the following expression as an estimator of $x^TLx$:
\begin{equation}
  \label{eq:estimator}
  I_G = \sum_{u \in V} \delta_u x_u^2 - \sum_{u \in \mathcal{S}}\sum_{v \in V} x_u x_v w(u,v)
        - \sum_{u \in \mathcal{L}}\sum_{v \in \mathcal{S}}x_u x_v w(u,v)
        - \sum_{u \in \mathcal{L}}\frac{\delta_u^\L}{\alpha} \sum_{v \in \mathcal{L}} x_u x_v Y_u^v.
\end{equation}

\begin{lemma}
  \label{lem:S2-Estimator}
 Let $G = (V, E, w)$ be an S2-Graph and $L = L(G)$ be the (unnormalized) Laplacian of $G$, then $I_G$ (defined in Equation (\ref{eq:estimator})) is an unbiased estimator of $x^TLx$. Furthermore, it gives a $(1 + \eps, 0.001)$-approximation to $x^TLx$.
\end{lemma}

\begin{proof}
Since $\E[Y_u^v] = \frac{\alpha w(u,v)}{\delta_u^\L}$ (by (\ref{eq:b-1})), it is straightforward to show that
$$\E[I_G] = \sum_{u \in V} \delta_u x_u^2 - \sum_{u \in \mathcal{S}}\sum_{v \in V} x_u x_v w(u,v)
            - \sum_{u \in \mathcal{L}}\sum_{v \in V}x_u x_v w(u,v)
            = \sum_{(u, v) \in E}  (x_u - x_v)^2 w(u,v) = x^T L x. $$
Now let us compute the variance of $I_G$. Note that if $\var{I_G} = O\left(\eps^2 (x^TLx)^2\right)$, then by taking constant $c_\alpha$ in $\alpha = c_\alpha \eps^{-\frac{5}{3}}$ large enough, a Chebyshev's inequality immediately yields the lemma. The variance of $I_G$
\begin{eqnarray}
  \var{I_G}  & = & \var{\sum_{u\in\L}\frac{\delta_u^\L}{\alpha} \sum_{v\in\L} x_ux_vY_u^v} \nonumber \\
             & = &  \sum_{u\in\L}\frac{(\delta_u^\L)^2}{\alpha^2} \sum_{v\in\L} x_u^2x_v^2\var{Y_u^v} \nonumber \\
             &\leq & \sum_{u\in\L}\frac{(\delta_u^\L)^2}{\alpha^2} x_u^2\sum_{v\in\L} x_v^2\frac{\alpha w(u,v)}{\delta_u^\L} \quad \quad \quad (\text{by (\ref{eq:b-1})}) \nonumber \\
             & = & \frac{1}{\alpha}\sum_{u\in\L}\delta_u^\L x_u^2\sum_{v\in\L} x_v^2 w(u,v) \nonumber \\
             &\leq  & \frac{1}{\alpha}\sum_{u\in\L}\delta_u^\L x_u^2\sum_{v\in\L} x_v^2 \frac{2\delta_v}{\alpha} \quad \quad  (w(u,v) \le 2 \gamma \le \frac{2\delta_v}{\alpha} \text{ by def. of S2-graph and def. of $\L$}) \nonumber \\
             &\leq & \frac{2}{\alpha^2} \sum_{u \in V} \delta_ux_u^2\sum_{v \in V} \delta_v x_v^2 \quad \quad \quad \quad (\delta_u^\L \le \delta_u \text{ by definitions}) \nonumber \\
             & = & \frac{2}{\alpha^2}\norm{D^{1/2}x}_2^4, \label{eq:c-1}
\end{eqnarray}
where $D = \text{diag}(\delta_1, \delta_2, \ldots, \delta_n)$.  The normalized Laplacian of $G$ can be written as $\tilde{L} = D^{-\frac{1}{2}} L D^{-\frac{1}{2}}$. Define $\hat{x} = D^{1/2}x$, we have
$$
  \|\hat{x}\|_2^2 = \hat{x}^T \hat{x} \leq \frac{1}{\lambda_1(\tilde L)}\hat{x}^T \tilde L \hat{x}
                  \overset{by~(\ref{eq:cheeger})}{\leq} \frac{2}{h_G^2}(x^TLx) \overset{\text{property of S2-graph}}{<} \frac{2}{\alpha^2\eps^4}(x^TLx),
$$
which together with (\ref{eq:c-1}) gives $\var{I_G} < \frac{8}{\alpha^6\eps^8}(x^TLx)^2 = O\left (\eps^2(x^TLx)^2 \right )$.
\end{proof}

% Using the standard ``median-trick'' (i.e. run $\log(1/\delta)$ independent estimators of $I_G$ in (\ref{eq:estimator}) and return the median of them), we can boost the success probability to $1 - \delta$ for any $\delta > 0$.
We summarize our result for the S2-Graph $G$ in the following theorem.

\begin{theorem}
\label{thm:basic-simple}
There is a sketching algorithm which given an S2-graph, outputs a $(1+\eps, 0.001)$-spectral-sketch of size $\tilde{O}(n/\eps^{\frac{5}{3}})$.
\end{theorem}

\paragraph{General Graphs}
Now let us extend our result to general positively weighted simple graphs $G = (V, E, w)$. We now require $w_{\max}/w_{\min} = \text{poly}(n)$.

%One property of the Laplacian matrix $L = L(G)$ is that given any fixed $x \in \mathbb{R}^{n}$, $x^TLx = \sum_{u \in V} \sum_{v \in V} (x_u -x_v)^2 w(u,v)$, that is, $x^T L x$ is a linear combination of $w(e)\ (e \in E)$. This immediately yields the fact that we can partition the edge set $E$ into joint groups, create a sketch for each group. At the time of query we perform an estimation for each group based on its sketch, and then merge (simply sum up) all estimations to give the final answer. Formally,
%

The following lemma will be used in the analysis of our sketching algorithms.
\begin{lemma}
\label{lem:spielman}
Given an undirected positively weighted graph $G = (V, E, w)$ with ${w_{\max}}/{w_{\min}} = \text{poly}(n)$, there is an algorithm that takes $G$ as the input, and output a graph $\tilde G = (V, \tilde E, \tilde w)$ such that
  \begin{enumerate}
    \item $\tilde G$ is a $(1+\eps, 0.001)$-spectral-sketch of $G$ of size $\tilde{O}({n}/{\eps^2})$ bits.
    \item ${\tilde w_{\max}}/{ \tilde w_{\min}}$ is bounded by $\text{poly} (n)$.
  \end{enumerate}
\end{lemma}

\begin{proof}
We first run the spectral sparsification algorithm by Batson, Spielman and Srivastava \cite{BSS14}, which produces $(1+\eps)$-spectral sparsifier $H$ of $G$ with size of $\tilde{O}({n}/{\eps^2})$ bits with probability $0.999$.

Assume (by rescaling) that the edge weights in $G$ are between $1$ and $n^C$ for a constant $C > 0$. Since $H$ is also a cut sparsifier, all weights in $H$ must be at most $2n^C$, as otherwise $H$ would not preserve a specific cut up to a factor of $2$. Let $\tilde{G}$ be formed from $H$ by removing all edge weights smaller than $1/n^D$ for a sufficiently large constant $D > 0$. Let $x$ be a unit vector orthogonal to $1^n$, and assume the underlying graph with Laplacian $L(\tilde{G})$ is connected (otherwise we could have first split it into connected components). Then $x^T L(\tilde{G}) x \geq x^T L(H) x - O(n^2/n^D).$ Since $x^T L(H) x \geq (1-\eps)x^TL(G) x \geq (1-\eps) \lambda_1(L(G))$, and $\lambda_1(L(G)) \geq 1/n^2$ by \cite{gy}, it follows that $x^T L(\tilde{G}) x \geq (1-\eps) x^T L(H) x$, assuming $D > 0$ is a sufficiently large constant. Since also $x^T L(\tilde{G}) x \leq x^T L(H) x$, it follows that $\tilde{G}$ is a $(1 + O(\eps))$-spectral sparsifier of $G$ with edge weights that are between $1/n^D$ and $2n^C$.
\end{proof}

%\begin{proof}
%Let us first sketch the algorithm in \cite{spielman2011graph}: We start by sampling $q = \Theta(n\log n / \eps^2)$ edges independently with replacement from $E$, where each $e\in E$ is sampled with probability $p_e = w_e r_e / (n -1)$ ($r_e$ is the {\em effective resistance} of $e$). Finally, reweight each sampled $e$ with
% \begin{equation}
% \label{eq:a-1}
% \tilde w_e = w_e / (q p_e) = \Theta\left(1 / (\log n / \eps^2 \cdot r_e) \right),
% \end{equation}
% and add it to $\tilde E$. Let $\tilde G = (V, \tilde E, \tilde w)$ be the resulting reweighted subgraph.
%
% The first part is a direct result of \cite{spielman2011graph}. Now we show the second part. For $e = (s, t) \in E$, let its {\em edge connectivity} $\lambda_e$ be defined as the minimum weight of cut in $G$ separating the vertices $s$ and $t$.
% Ahn et al. \cite{ahn2013spectral} gave a relation between effective resistance $r_e$ and edge connectivity $\lambda_e$ for an edge $e$,
% \begin{equation}
% \label{eq:a-2}
%    n^{-2/3}\lambda_e \leq r_e^{-1} \leq \lambda_e.
% \end{equation}
% Note that we also have
% \begin{equation}
% \label{eq:a-3}
% w_{\min} \leq \lambda_e \leq n^2 w_{\max}.
% \end{equation}
% Combining (\ref{eq:a-1}), (\ref{eq:a-2}), (\ref{eq:a-3}), ${\tilde w_{\max}}/{ \tilde w_{\min}}$ is bounded by a polynomial of $n$.
% \footnote{In the original algorithm described in \cite{spielman2011graph}, they also sum up the weights if one edge is sampled more than once, however it will not affect our conclusion once note that each edge would be sampled at most $q$ times.}
%\end{proof}

The following observation is due to the linearity of Laplacian.
\begin{observation}
  \label{ob:Laplacian}
  Given any simple graph $G = (V, E, w)$,  let $L$ be its Laplacian. Let $E_1, E_2, \ldots, E_k$ be a disjoint partition of $E$, and let $G_i = (V, E_i, w)$. Let $L_i$ be the Laplacian of $G_i$. We have $x^TLx = \sum_{i=1}^kx^TL_ix$ for any $x\in \mathbb{R}^{n}$.
\end{observation}

Our high-level idea is to reduce general graphs to S2-graphs. Based on Observation~\ref{ob:Laplacian}, we can first partition the edge set $E$ into $E_1, \ldots, E_k\ (k = \Theta(\log n))$ such that for any $e \in E_i$ we have $w(e) \in [2^{i-1}w_{\min}, 2^iw_{\min})$, and then sketch each subgraph $G_i$ separately. Finally, at the time of a query, we simply add all estimators $I_{G_i}$ together. Thus it suffices to focus on a  graph with all weight $w(e) \in [\gamma, 2\gamma) $ for some $\gamma > 0$.

We next partition each subgraph $G_i$ further so that each component $P$ satisfies $h_P \ge c_{\alpha}\eps^{\frac{1}{3}}$. Once this property is established, we can use Algorithm~\ref{alg:sketch-basic-simple} to sketch each $P$ separately. We describe this preprocessing step in Algorithm~\ref{alg:preprocessing}.

\begin{algorithm}[t]
%\DontPrintSemicolon % Some LaTeX compilers require you to use \dontprintsemicolon instead
\KwIn{A graph $G = (V, E, w)$ such that for any $e \in E$, $w(e) \in [\gamma, 2\gamma)$; a parameter $h > 0$}
\KwOut{
A set $\mathcal{P}$ of edge disjoint components of $G$ such that for each $P \in \mathcal{P}$, $h_P > h$; and a graph $Q$ induced by the rest of the edges in $G$.
}
$\mathcal{P} \gets \{ G \}$, $Q \gets \emptyset$\;
\While {$\exists$ $P \in \mathcal{P}$ such that Cheeger's constant $h_P \leq h$} {
  Find an arbitrary cut $(S, \overline S)$ in $P$, such that $\Phi(S) \leq h$\;
  Replace $P$ with its two subgraphs $P(S)$ and $P(\overline{S})$ in $\mathcal{P}$\;
  Add all edges in the cut $(S, \overline S)$ into $Q$\;
}

\Return{$(\mathcal{P}$, $Q$)}\;
\caption{{\bf Spectral-Preprocessing}($G, h$)}
\label{alg:preprocessing}
\end{algorithm}

The $Q$ returned by Algorithm~\ref{alg:preprocessing} is a set of cut edges we will literally keep. The following lemma bounds the size of $Q$. The proof is folklore, and we include it for completeness.
\begin{lemma}
  \label{lem:boundQ}
  For any positively weighted graph $G = (V, E, w)$ such that for any $e \in E$, $w(e) \in [\gamma, 2\gamma)$ for some $\gamma > 0$, the number of edges of $Q$ returned by Algorithm \ref{alg:preprocessing} \text{Spectral-Preprocessing}($G, h$) is bounded by $O(h m \log m)$.
\end{lemma}

\begin{proof}
In Algorithm \ref{alg:preprocessing}, we recursively split the graph $G = (V, E, w)$ into connected components until for every component $P = (V_P, E_P)$, its Cheeger's constant $h_P$ = $\inf_{S\subset V_P}\Phi_C(S) > h$. Consider a single splitting step: We find a cut $(S, \bar{S})$ with $\vol_P(S) \le \vol_P(\bar{S})$ in $P$ such that $\Phi_P(S) = \frac{w_P(S, \bar{S})}{\vol_P(S)} \leq h$. We can think in this step each edge in $\vol_P(S)$ contributes at most $h$ edges to $Q$ on average, while edges in $\vol_P(\bar{S})$ contribute nothing to $Q$. We call $S$ the {\em Smaller-Subset}.

By the definition of volume, we have $\vol_P(V_P) = \vol_P(S\cup\overline{S}) = \vol_P(S) + \vol_P(\overline{S}) \geq 2 \vol_P(S)$, and $\vol_G(V) = 2 m$. Thus in the whole recursion process, each edge will appear at most $O(\log m)$ times in Smaller-Subsets, hence will contribute at most $O(h \log m)$ edges to $Q$. Therefore the number of edges of $Q$ is bounded by $O(h m \log m)$ words.
\end{proof}

Now we show the main algorithm for general graphs and analyze its performance. The algorithm is described in Algorithm~\ref{alg:sketch-basic-general}.

\begin{algorithm}[t]
\KwIn{$G = (V, E, w)$ with all weights in $[w_{\min}, w_{\max}]$; a quality control parameter $\eps$}
\KwOut{A $(1+\eps, 0.001)$-spectral-sketch $\sk{G}$ of $G$}
Edge-disjointly partition $G$ into $\mathcal{H} = \{H_1, \ldots, H_{k}\}$ s.t. all edges in $H_i$ have weights in $[2^{i-1}w_{\min}, 2^iw_{\min})$\;
$\sk{G} \gets \emptyset$\;
\ForEach{$H \in \mathcal{H}$} {
  %rescale all edges in $P_i$ by a factor $1/2^{i-1}$  to get $\tilde P_i$\;
  $(\mathcal{P}, Q) \gets $ {Spectral-Preprocessing}($H, \alpha\eps^2$)\;
  Add $Q$ into $\sk{G}$\;
  \For {$P \in \mathcal{P}$} {
    Add \textbf{Spectral-S2}($P, \eps$) into $\sk{G}$\;
  }
}

\Return{$\sk{G}$}\;
\caption{{\bf Spectral-Basic}($G, \eps$)}
\label{alg:sketch-basic-general}
\end{algorithm}

The following lemma summarize the functionality of Algorithm~\ref{alg:sketch-basic-general}.

\begin{lemma}
\label{lem:basic-general}
Given a graph $G = (V, E, w)$, let $\sk{G} \gets$~\text{Spectral-Basic}($G, \eps$), then for any given $x \in \mathbb{R}^{n}$, $\sk{G}$ can be used to construct an unbiased estimator $I_G$ which gives a $(1+\eps, 0.001)$-approximation to $x^T L(G) x$. The sketch $\sk{G}$ uses $\tilde{O}(\eps^{\frac{1}{3}}m + n/\eps^{\frac{5}{3}})$ bits.
\end{lemma}

\begin{proof}
In Algorithm~\ref{alg:sketch-basic-general}, $G$ is partitioned into a set of edge disjoint components $\mathcal{P} = \{P_1, \ldots, P_t\}$, and we build a sketch $\sk{P_i}$ for each $P_i \in \mathcal{P}$ from which we can construct an unbiased estimator $I_{P_i}$ for $x^T L(P_i) x$ with variance bounded by $O(\eps^2 (x^T L(P_i) x)^2)$ according to Lemma~\ref{lem:S2-Estimator}. Moreover, we have stored all edges between these components; let $Q$ be the induced subgraph of these edges.  Our  estimator to $x^T L(G) x$ is
\begin{equation}
\label{eq:est-basic-general}
I_G = \sum_{i = 1}^t I_{P_i} + x^T L(Q) x,
\end{equation}
where $I_{P_i}$ defined in Equation~(\ref{eq:estimator}) is an unbiased estimator of $x^T L(P_i) x$.
By the linearity of Laplacian (Observation~\ref{ob:Laplacian}), $I_G$ is an unbiased estimator of $x^T L(G) x$. Now consider its variance:
\begin{eqnarray*}
\var{I_G} &=& \var{\sum_{1\leq i\leq t} I_{P_i} + x^T L(Q) x} \\
&\leq& O(\eps^2)\sum_{1\leq i \leq t}\left(x^T L(P_i) x\right)^2 \quad \quad \quad \text{(due to the independence of $P_i$'s)}\\
&\leq& O(\eps^2) \left(\sum_{1\leq i \leq t}x^T L(P_i) x+x^T L(Q) x \right)^2 \quad \text{($L(P_i)$ and $L(Q)$ are  positive semidefinite)}\\
&=& O(\eps^2) \left(x^T L(G) x\right)^2.
\end{eqnarray*}
The correctness follows from a Chebyshev's inequality.

Now we bound the size of the sketch $\sk{G}$. Consider a particular $H$ at Line $3$ of Algorithm~\ref{alg:sketch-basic-general}. For each $P = (V_P, E_P) \in \mathcal{P}$, the size of $\sk{P}$ by running {\em Spectral-S2$(P, \eps)$} at Line $7$ is bounded by $\tilde{O}(\abs{V_P}/\eps^{\frac{5}{3}})$ bits (Theorem~\ref{thm:basic-simple}); and for the remaining subgraph $Q = (V_Q, E_Q)$, $\sk{Q}$ is bounded by $\tilde{O}(\eps^{\frac{1}{3}} \abs{E_Q})$ bits (Lemma~\ref{lem:boundQ}). Thus
\begin{eqnarray*}
\text{size}(\sk{P_i}) & = & \text{size}(\sk{Q}) + \sum_{P \in \mathcal{P}} \text{size}(\sk{P}) \\
& \le & \tilde{O}\left(\eps^{\frac{1}{3}} \abs{E_Q} + \sum_{P \in \mathcal{P}}  \abs{V_P} /\eps^{\frac{5}{3}}\right) \\
& \le & \tilde{O}\left(\eps^{\frac{1}{3}} m + n/\eps^{\frac{5}{3}}\right) \text{bits}. \quad \quad (\text{$\{P \in \mathcal{P}\}$ are vertex-disjoint})
\end{eqnarray*}
Since there are $k = \Theta(\log n)$ of $H_i$'s in $\mathcal{H}$, the size of $\sk{G}$ is bounded by
 $\tilde{O}\left(n/\eps^{\frac{5}{3}} + \eps^{\frac{1}{3}} m\right) \cdot \log n = \tilde{O}\left(n/\eps^{\frac{5}{3}} + \eps^{\frac{1}{3}} m\right)$ bits.
\end{proof}

We conclude this section with the following theorem.
\begin{theorem}
\label{thm:basic-general}
There is a sketching algorithm which given an undirected positively weighted graph $G = (V, E, w)$ with $w_{\max}/w_{\min} = \text{poly}(n)$, outputs a $(1+\eps, 0.01)$-spectral-sketch of size $\tilde{O}(n/\eps^{\frac{5}{3}})$.
\end{theorem}

\begin{proof}
The algorithm is as follows: we first run the spectral sparsification algorithm in \cite{BSS14}, obtaining a graph $\tilde{G} = (V, \tilde{E}, \tilde{w})$. By Lemma~\ref{lem:spielman} we have $|\tilde{E}| = \tilde{O}(n/\eps^2)$ and $\tilde{w}_{\max}/\tilde{w}_{\min} = \text{poly}(n)$. We then run Algorithm~\ref{alg:sketch-basic-general}, getting a $(1+\eps, 0.01)$-spectral-sketch of size
$
\tilde{O}\left(n/\eps^{\frac{5}{3}} + \eps^{\frac{1}{3}} |\tilde{E}|\right) = \tilde{O}(n / \eps^{\frac{5}{3}})
$.
\end{proof}

\subsubsection{An Improved Sketching Algorithm}
\label{sec:improve}
In this section, we further reduce the space complexity of the sketch to $\tilde O( {n}/{\eps^{\frac{8}{5}}} )$. At a high level, such an improvement is achieved by partitioning the graph into more subgroups (compared with a hierarchical partition on weights, and $\mathcal{S}(G)$ and $\mathcal{L}(G)$ for each weight class in the basic approach), in each of which vertices have similar unweighted degrees {\em and} weighted degrees. An estimator based on a set of sampled edges from such groups will have smaller variance. This finer partition, however, will introduce a number of technical subtleties, as we will describe below.

We set the constant $\beta = c_\beta \eps^{-\frac{8}{5}}$ throughout this section where $c_\beta$ is a constant.

\paragraph{Special Graphs}
We first consider a class of simple graphs.
\begin{definition}[S3-graph]
\label{def:S3-graph}
We say an undirected weighted graph $G = (V, E, w)$ is an {\em S3-graph} (reads ``simple type-3 graph") if we can assign directions to its edges in a certain way, getting a directed graph $\vec{G} = (V, \vec{E}, w)$
satisfying the following.
\begin{enumerate}
\item All weights $\{w(e)\ |\ e \in \vec{E}\}$ are within a factor of $2$, i.e., for any $e \in \vec{E}$, $w(e) \in [\gamma, 2\gamma)$ for some constant $\gamma > 0$.

\item For each $u \in V$, $\uwoutdeg{u} \in [2^\kappa\beta, 2^{\kappa+1}\beta)$, where $2^\kappa\beta \le 1/\eps^2$.

%\item $w_{\max} / w_{\min} = \text{poly}(n)$.
\end{enumerate}
We call $\vec{G}$ the {\em buddy} of $G$. Note that $\vec{G}$ is not necessarily unique, and we just need to consider an arbitrary but fixed one. In this section we will assume that we can obtain such a buddy directed graph $\vec{G}$ ``for free", and will not specify the concrete algorithm.  Later when we deal with general graphs we will discuss how to find such a direction scheme.
\end{definition}

We still make use of Algorithm~\ref{alg:preprocessing} to partition the graph $G$ into components such that the Cheeger's constant of each component is larger than $h = \beta \eps^2$. One issue here is that, after storing and removing those cut edges (denoted by $Q$ in Algorithm~\ref{alg:preprocessing}), the second property of S3-graph may not hold, since the out-degree of some vertices in $G$'s buddy graph $\vec{G}$ will be reduced. Fortunately, the following lemma shows that the number of vertices whose degree will be reduced more than half is small, and we can thus afford to store all their out-going edges. The out-degree of the remaining vertices is within a factor of $4$, thus we can still effectively bound the variance of our estimator.

\begin{lemma}
\label{alg:improve-preprocessing}
If we run Algorithm~\ref{alg:preprocessing} on an S3-graph $G = (V, E, w)$ with $h = 2^{-\kappa}$, then
\begin{enumerate}
  \item At most $\tilde O(\beta n)$ cut edges (i.e., $Q$) will be removed from $G$.

  \item There are at most $\tilde{O}(2^{1-{\kappa}}n)$ vertices in $G$'s buddy graph $\vec{G}$ which will reduce their out-degrees by more than a half after the removal of $Q$.
\end{enumerate}
\end{lemma}

\begin{proof}
Since $m = O(2^{\kappa} \beta n)$, $h = 2^{-{\kappa}}$ and $2^{\kappa} \beta = \tilde O(\frac{1}{\eps^2})$, Lemma \ref{lem:boundQ} directly gives the first part. For the second part, note that for each vertex $u$ we have $\woutdeg{u} \ge 2^\kappa \beta$, thus we need to remove at least $2^{\kappa - 1} \beta$ edges to reduce $\woutdeg{u}$ to $2^{\kappa - 1} \beta$. Therefore the number of such vertices is at most $\tilde{O}(\beta n/2^{\kappa - 1} \beta) = \tilde{O}(2^{1-{\kappa}}n)$.
\end{proof}

For each component $P = (V_P, E_P)$ after running Algorithm~\ref{alg:preprocessing}, let $\vec{P} = (V_P, \vec{E}_P)$ be its buddy directed graph. Slightly abusing the notation, define $\S(\vec{P}) = \{(u, v) \in \vec{E}_P\ |\ \uwoutdegP{u} < 2^{\kappa - 1}\beta\}$, and $\L(\vec{P}) = \{(u, v) \in \vec{E}_P\ |\ \uwoutdegP{u} \ge 2^{\kappa - 1}\beta\}$. We will again omit ``$(\vec{P})$" or ``$(P)$" when there is no confusion.

The sketch for an S3-graph $G$ is constructed using Algorithm~\ref{alg:improved-simple}.

\begin{algorithm}[t]
\KwIn{An S3-graph $G$ and a parameter $\eps$}
\KwOut{A $(1+\eps, 0.01)$-spectral-sketch of $G$}
$\{\mathcal{P}, Q\} \gets$ {\bf Spectral-Preprocessing}($G, \beta \eps^2$)\;
\ForEach{$P = (V_P, E_P, w) \in \mathcal{P}$}{
Let $\vec{P} = (V_P, \vec{E}_P, w)$ be its buddy directed graph\;

\ForEach{$u \in V_P$}{Add $\windegP{u}$ and $\delta_u(P)$  to $\sk{G}$\;}
Add $\S(\vec{P})$ to $\sk{G}$\;
%\ForEach{$u \in \S(\vec{P})$}{Add all of its out-going edges in $\vec{E}$ to $\sk{G}$\;}

\ForEach{$u \in V_P$}{
Sample $\beta = c_\beta \eps^{-\frac{8}{5}}$ edges with replacement from $\{(v, u) \in \L\}$, where the probability that $(v, u)$ is sampled is $w(v,u)/\windegP{u}$. Add sampled edges to $\sk{G}$\;
%For each sampled edge, rescale its weight by a factor of $\windegP{u}/\beta$ and add to $\sk{G}$\;
}

Add $Q$ to $\sk{G}$\;
}
\Return $\sk{G}$
\caption{{\bf Spectral-S3}($G, \eps$)}
\label{alg:improved-simple}
\end{algorithm}

It is easy to see that the size of $\sk{G}$ is bounded $\tilde{O}(2^{1-\kappa}\abs{V_P}) \cdot 2^{\kappa-1}\beta + \tilde{O}(\beta n)  + \tilde{O}(\beta n) = \tilde{O}(\beta n)$, where the first term in LHS is due to the definition of $\S(\vec{P})$ and the second property of Lemma~\ref{alg:improve-preprocessing}.

Let $Y_u^v$ be the random variable denoting the number of times (directed) edge $(v, u)$
%(note, it is \emph{not} from $u$ to $v$)
is sampled when we process the vertex $u$. Clearly, $\E[Y_u^v] = \frac{\beta w(v,u)}{\windegP{u}} $ and $\var{Y_u^v} \leq \frac{\beta w(v,u)}{\windegP{u}}$. For a given $x \in \mathbb{R}^{\abs{V_P}}$, we construct the following estimator for each component $P$ using $\sk{G}$.

\begin{equation}
  \label{eq:estimator-P}
  I_P = \sum_{u \in V_P} x_u^2  \delta_{u}(P) - 2 \sum_{(u, v) \in \S}  x_u x_v w(u,v)
        - 2 \sum_{u \in V_P} \frac{\windegP{u}}{\beta}\sum_{(v, u) \in \L} x_u x_v Y_u^v
\end{equation}
Similar to the analysis in Section~\ref{sec:basic-simple}, it is easy to show that $I_P$ is an unbiased estimator of $x^T L(P) x$ by noticing

\begin{eqnarray*}
   x^T L(P) x
   &=& \sum_{(u, v) \in \vec{E}_P}(x_u - x_v)^2 w(u,v)\\
   & = & \sum_{u \in V_P} x_u^2 \delta_u(P) - 2 \sum_{(u, v) \in \S}  x_u x_v w(u,v)
        - 2 \sum_{(u, v) \in \L}  x_u x_v w(u,v).
\end{eqnarray*}
We next bound the variance
\begin{eqnarray*}
  \var{I_P}  &=& \var{2\sum_{u\in V_P}\frac{\windegP{u}}{\beta} \sum_{(v, u) \in \L} x_ux_vY_u^v}\\
             &=& 4\sum_{u\in V_P}\frac{\left(\windegP{u}\right)^2}{\beta^2} \sum_{(v, u) \in \L} x_u^2x_v^2\var{Y_u^v}
              \\
             &\leq& 4\sum_{u\in V_P}\frac{\left(\windegP{u}\right)^2}{\beta^2} \sum_{(v, u) \in \L} x_u^2x_v^2\frac{\beta w(v,u)}{\windegP{u}}
                \quad \quad \quad \left(\var{Y_u^v} \leq \frac{\beta w(v,u)}{\windegP{u}}\right)\\
             &=& 4\sum_{u\in V_P}x_u^2\frac{\windegP{u}}{\beta} \sum_{(v, u) \in \L} x_v^2w(v,u)\\
             &\le& 4\sum_{u\in V_P}x_u^2\frac{\windegP{u}}{\beta} \sum_{(v, u) \in \L} x_v^2\cdot 2\gamma
             \quad \quad \quad \left( w(v,u) \in [\gamma, 2\gamma)\right)\\\\
             &\leq& \frac{4}{\beta}\sum_{u\in V_P}x_u^2\windegP{u} \sum_{(v, u) \in \mathcal{L}} x_v^2\frac{2\woutdegP{v}}{2^{\kappa-1}\beta}
             \quad \left(\uwoutdegP{v} \geq 2^{\kappa - 1}\beta ~\text{and}~  \woutdegP{v}\geq \gamma \cdot \uwoutdegP{v} \right )\\
             &\leq& \frac{16}{2^{\kappa}\beta^2} \sum_{u\in V_P} \delta_u(P) x_u^2\sum_{v\in V_P} \delta_v(P) x_v^2
                \quad \quad \left(\windegP{u} \leq \delta_u(P) ~\text{and}~ \woutdegP{v} \leq \delta_v(P)\right) \\
             &=& \frac{16}{2^{\kappa}\beta^2} \|\hat x\|_2^4,
\end{eqnarray*}
where $\|\hat x\|_2^4 = \|D^{1/2} x\|_2^4 = \sum_{u\in V_P} \delta_u(P) x_u^2\sum_{v\in V_P} \delta_v(P) x_v^2$.

Similar to before we have $$
  \|\hat{x}\|_2^2 = \hat{x}^T \hat{x} \leq \frac{1}{\lambda_1(\tilde L)}\hat{x}^T \tilde L \hat{x}
                  \overset{by~(\ref{eq:cheeger})}{\leq} \frac{2}{h_G^2}(x^TLx) \overset{h_G > 2^{-\kappa} \text{ by Algorithm~\ref{alg:preprocessing}}}{<} {2} \cdot {2^{2\kappa}} \cdot (x^TLx),
$$
Recall that in an S3-graph, $2^{\kappa} \leq  1/(\beta \eps^2)$, hence
$$\frac{\var{I_P}}{\left(x^T L(P) x\right)^2} = O\left(\frac{2^{3\kappa}}{\beta^2}\right) = O\left(1/(\beta^{5}\eps^{6})\right) =  O(\eps^2).$$ Setting constant $c_\beta$ large enough in $\beta = c_\beta \eps^2$, by a Chebyshev's inequality, $I_P$ is a $(1+\eps, 0.01)$-approximation to $x^T L(P) x$.  

\begin{theorem}
\label{thm:improved-simple}
There is a sketching algorithm which given an S3-graph, outputs a $(1+\eps, 0.01)$-spectral-sketch of size $\tilde{O}(n/\eps^{\frac{8}{5}})$.
\end{theorem}

\paragraph{General Graphs}
To deal with general graphs $G = (V, E, w)$ for which the only requirement is $w_{\max}/w_{\min} \le \text{poly}(n)$, we try to ``partition" it to poly$\log{n}$ subgraphs, each of which is an S3-graph. We note that the partition we used here is {\em not} a simple vertex-partition or edge-partition, as will be evident shortly.

Our first step is to assign each edge in $E$ a direction so that in the later partition step we can partition $G$ to S3-graphs and simultaneously get their buddy directed graphs. The algorithm is described in Algorithm~\ref{alg:assign-direction}.

\begin{algorithm}[t]
%\DontPrintSemicolon % Some LaTeX compilers require you to use \dontprintsemicolon instead
\KwIn{A graph $G = (V, E)$ and a parameter $t$}
\KwOut{$\vec{G}=(V, \vec{E})$, a directed graph by assigning each edge in $G$ a direction}
Arbitrarily assign a direction to each edge in $E$, getting $\vec{E}$\;
\While{$\exists (u, v) \in \vec{E}$ s.t. $\uwoutdeg{u} \geq t$ and $\uwoutdeg{v} < t-1$}{
  Change the direction of $(u, v)$\;
}
\Return{$\vec{G}= (V, \vec{E})$}\;
\caption{{\bf Assign-Direction}($G, t$)}
\label{alg:assign-direction}
\end{algorithm}

\begin{lemma}
  \label{lem:assign-direction}
  Given $G = (V, E)$ and $s>1$ as input,  \text{Assign-Direction}$(G, s)$ (Algorithm~\ref{alg:assign-direction}) will finally stop and return $\vec{G} = (V, \vec{E})$ with the property that
  for each $(u, v) \in \vec{E}$, $\uwoutdeg{u} < s$ or $\uwoutdeg{v} \geq s - 1$.
\end{lemma}

\begin{proof}
  The second part is trivial according to Algorithm~\ref{alg:assign-direction}. Now we show that {Assign-Direction}$(G, s)$ will finally stop. Let $S = \{(u, v) \in \vec{E}\ |\ \uwoutdeg{u} \geq s ~\text{and}~ \uwoutdeg{v} < s-1 \}$, and $\Delta(S) = \sum_{(u, v)\in S} (\uwoutdeg{u} - \uwoutdeg{v})$. The algorithm stops if and only if $\Delta(S) = 0$. It is easy to see that $\Delta(S)$ is finite for arbitrary $\vec{G}$, thus the algorithm will stop if we can show that $\Delta(S)$ will decrease by \emph{at least} $2$ each time we execute Line $3$. To this end, we only need to show that each execution of Line $3$ will not add any new edge to $S$.

Consider executing Line $3$ on edge $(u, v)$. For $(u, w) \not \in S$, since $(u, v) \in S$, we have $\uwoutdeg{w} \geq s - 1$. Clearly, after executing Line $3$, $(u, w)$ will not be added to $S$ because $\uwoutdeg{w} \geq s - 1$ still holds.
For $(w, v) \not \in S$, since $(u, v) \in S$, we have $\woutdeg{w} < s$. After executing Line $3$, $\uwoutdeg{w} < s$ still holds, hence $(w, v)$ will not be added into $S$.
\end{proof}

For a set of directed edges $\vec{E}$, let $d_u^{\tt{out}}(\vec{E}) \gets |\{v\ |\ (u, v) \in \vec{E}\}|$.  Our partition step is described in Algorithm~\ref{alg:improved-partition}. We first run the spectral sparsification algorithm ~\cite{BSS14}, and then assign directions to each edge. Next, we partition the edges based on their weights, and then partition the directed graph based on the unweighted out-degree of each vertex. Finally, we recursively perform all the above steps on a subgraph induced by a set of edges which have large weights. Notice that the purpose of introducing directions on edges is to assist the edge partition.

\begin{algorithm}[t]
%\DontPrintSemicolon % Some LaTeX compilers require you to use \dontprintsemicolon instead
\KwIn{A graph $G = (V, E, w)$ and a parameter $\eps$}
\KwOut{A set of graph components $\mathcal{P}$}
\If {$n< 3$}{
\Return $\{ G \}$\;
}
Run \cite{BSS14} on $G$ with parameter $\eps$, get a spectral sparsifier $G' = (V, E', w')$ with $|E'| =  \eta\cdot \frac{n}{\eps^2}$ where $\eta = \tilde{O}(1)$\;
$\tilde{\eps} = \eps/\sqrt{\eta}$, $s \gets  1/{\tilde{\eps}}^2$\;
$\vec{G}' = (V, \vec{E}', w') \gets$ {\bf Assign-Direction}$(G', 2s)$ \;
Partition $\vec{E}'$ into $\vec{E}_j'$'s such that for each $\vec{E}_j'$, all $e \in \vec{E}_j'$ have $w(e) \in [2^j, 2^{j+1})$\;
%\tcc{$j$ can be negative}
$\mathcal{P} \gets \emptyset$\;
\ForEach{$\vec{E}_j'$} {
\tcc{Recall $\beta = c_\beta \eps^{-\frac{8}{5}}$ for a large enough constant $c_\beta$}
  Let $\vec{E}_{-\infty} \gets \left\{(u, v) \in \vec{E}_j'\ |\ d_u^{\tt{out}}(\vec{E}_j') < \beta \right\}$\;
  Let $\vec{E}_i \gets \left\{(u, v) \in \vec{E}_j'\ |\ d_u^{\tt{out}}(\vec{E}_j') \in [2^i \beta, 2^{i+1}\beta) \right\}$ for all $i \geq 0$ ~s.t.~ $2^i\beta \leq s = 1/{\tilde{\eps}}^2$\;
  Add $G(\vec{E}_{-\infty})$ and $G(\vec{E}_i)$ for all $0 \le i \le \log(1/(\tilde{\eps}^2 \beta))$ into $\mathcal{P}$\;
  Remove $\vec{E}_{-\infty}, \vec{E}_i$ from $\vec{E}'$\;
}
\tcc{Recursively apply on the remaining edges $E'$ (remove  directions on edges)}
\Return{$\mathcal{P} \cup \text{\bf Partition } (G(E'))$}\;
\caption{{\bf Partition}($G$)}
\label{alg:improved-partition}
\end{algorithm}

For the analysis, we first show that after each recursion in Algorithm~\ref{alg:improved-partition}, the number of vertices of the graph induced by the remaining edges will decrease by at least a constant fraction. In this way we can bound the number of recursion steps by $O(\log n)$.

\begin{lemma}
  \label{lem:recursion}
Given a graph $G = (V, E)$ with $m \leq sn\ (s > 1)$, let $\vec{G} \gets$ \text{Assign-Direction}$(G, 2s)$. If we remove all $(u, v)$ with $d_u^{\tt{out}}(\vec{E}) < 2s$ from $\vec{E}$ and get a subset $\vec{E}_r\subset \vec{E}$, then we have $|V_r(\vec{E}_r)| \leq n / (2 - 1/s)$
%where $V_r$ is the vertex set of $G(\vec{E}_r)$.
\end{lemma}

Before proving this lemma, note that in Algorithm~\ref{alg:improved-partition}, $m \leq sn$ is guaranteed by Line $3$, and ``remove all $(u, v)$ with $d_u^{\tt{out}}(\vec{E}) < 2s$" is done by Line $12$.

\begin{proof}
By Lemma \ref{lem:assign-direction},
  for each $(u, v) \in \vec{E}$, we have $d_u^{\tt{out}}(\vec{E}) < 2s$ or $d_v^{\tt{out}}(\vec{E}) \geq 2s - 1$.  If we remove all $(u, v)$ with $d_u^{\tt{out}}(\vec{E}) < 2s$, then for each $(u, v) \in \vec{E}_r$, we have $d_u^{\tt{out}}(\vec{E}_r) \geq 2s$ and $d_v^{\tt{out}}(\vec{E}_r) \ge 2s-1$. Consequently, $|V_r(\vec{E}_r)| (2s - 1) \leq m \leq sn$. Therefore we have
  $|V_r(\vec{E}_r)| \leq n / (2 - 1/s)$.
\end{proof}

The following lemma summarizes the properties of $\mathcal{P}$ returned by Algorithm~\ref{alg:improved-partition}.

\begin{lemma}
  \label{lem:improved-partition}
  Given $G = (V, E, w)$ with $w_{\max}/w_{\min} = \text{poly}(n)$, let $\mathcal{P} \gets$~\text{Partition}$(G)$ be a set of graphs after the partition, then (1) $|\mathcal{P}| = \text{poly}(\log n)$; and (2) for each $\vec{P} = (V_P, \vec{E}_P, w_P) \in \mathcal{P}$, if $|V_P| > 2$ then for any $e \in \vec{E}_P$,  $w(e) \in [\gamma, 2\gamma)$ for some $\gamma > 0$, and one of the following properties holds:
  \begin{enumerate}
  \item[] {\em Property $1$:} For each $u \in V_P$, $\uwoutdegP{u} < \beta$.
  \item[] {\em Property $2$:} There exists $i\ (0 \le i \le \log(\eta/(\beta \eps^2)))$,
    for each $u \in V_P$, $\uwoutdegP{u} \in [2^i\beta, 2^{i+1}\beta)$.
  \end{enumerate}
\end{lemma}

\begin{proof}
We only need to bound the size of $\mathcal{P}$. The rest directly follows from the algorithm.

First, Line $6$ and Line $10$ will partition $\vec{E}'$ to $O(\log^2 n)$ (assuming $n > 1/\eps$) sets.
Second, we bound the number of recursion steps. Note that if we directly remove $E_s = \{(u, v) \in \vec{E}'\ |\ d_u^{\tt{out}}(\vec{E}') < 2s\}$ from $\vec{E}'$, then by Lemma~\ref{lem:recursion} we know that there are at most $O(\log n)$ recursion steps. The subtlety is that we first partition $\vec{E}'$ into $\vec{E}'_j$'s and then remove all $(u, v) \in \vec{E}'_j$ with $d_u^{\tt{out}}(\vec{E}'_j) < 2s$. However, since $d_u^{\tt{out}}(\vec{E}'_j) \le d_u^{\tt{out}}(\vec{E}')$, every edge in $E_s$ will still be removed by at Line $12$. Therefore $\abs{\mathcal{P}} = O(\log^3 n)$.
\end{proof}

%With Lemma~\ref{lem:improved-partition},
We summarize our main result.

\begin{theorem} \label{thm:improved-general}
Given an undirected positively weighted $G = (V, E, w)$ with $w_{\max}/w_{\min} \le \text{poly}(n)$, there is a sketching algorithm that outputs a $(1+\eps, 0.01)$-spectral-sketch of size $\tilde{O}(n/\eps^{8/5})$.
\end{theorem}

\begin{proof}
Our final algorithm is described in Algorithm~\ref{alg:sketch-improved-general}.
For each component $H \in \mathcal{P}$, we store the whole $H$ if Property $1$ in Lemma~\ref{lem:improved-partition} holds (let $\mathcal{P}_1$ be this set), and we apply Theorem~\ref{thm:improved-simple} on $H$ (set $\delta = 1/\text{poly}\log n$) if Property $2$  in Lemma~\ref{lem:improved-partition} holds (let $\mathcal{P}_2$ be this set), thus the space usage is bounded by $\tilde{O}(n/\tilde{\eps}^{\frac{8}{5}}) = \tilde{O}(n/\eps^{\frac{8}{5}})$.  Since by $\abs{\mathcal{P}_1 \cup \mathcal{P}_2} = \abs{\mathcal{P}} = O(\log^3n)$ by Lemma~\ref{lem:improved-partition}, we can bound the total space by $\tilde{O}(n/\eps^{\frac{8}{5}})$.

Our final estimator, given $x \in \mathbb{R}^{n}$, is
$$I_G = \sum_{H \in \mathcal{P}_2} I_H + \sum_{H \in \mathcal{P}_1} x^T L(H) x,$$ where $I_H$ is defined in Equation~(\ref{eq:estimator-P}). Similar to the proof of Lemma~\ref{lem:basic-general}, we can bound $\var{I_G}$ by $O(\eps^2 \cdot (x^T L(G) x)^2)$. By a Chebyshev's inequality, $I_G$ is a $(1+\eps, 0.01)$-approximation of $x^T L(G) x$.
\end{proof}

\begin{algorithm}[t]
\KwIn{$G = (V, E, w)$; a quality control parameter $\eps$}
\KwOut{A $(1+\eps, 0.01)$-spectral-sketch $\sk{G}$ of $G$}
Let $\mathcal{H} \gets$ {\bf Partition}(G)\;
$\sk{G} \gets \emptyset$\;
\ForEach{$H \in \mathcal{H}$} {
\If{$H$ satisfies Property $1$ in Lemma~\ref{lem:improved-partition}}{
     Add the whole $H$ to $\sk{G}$\;
}
\ElseIf{$H$ satisfies Property $2$ in Lemma~\ref{lem:improved-partition}}{Add \textbf{Spectral-S3}($H, \eps$) into $\sk{G}$ \;}
}

\Return{$\sk{G}$}\;
\caption{{\bf Spectral-Improved}($G, \eps$)}
\label{alg:sketch-improved-general}
\end{algorithm}

%%% Local Variables: 
%%% mode: latex
%%% TeX-master: "paper"
%%% End: 

\section{Positive-Semidefinite Matrices}
\label{sec:psd}

In this section we consider sketching positive-semidefinite matrices, in the ``for all'' and ``for each'' models, respectively. 
%When we say a sketch satisfies ``for all'' / ``for each'' guarantee we always mean that the success probability is a constant, say, 0.99.

\subsection{``For All'' Model}
\label{sec:psd-all}

We first show that in the ``for all'' case (for general PSD matrices $A$), 
there is no better sketching algorithm (up to a logarithmic factor) than simply storing the whole matrix. 

%\subsection{Lower Bound for General Matrices with the For All Guarantee}
\begin{theorem}
\label{thm:all-exact}
For a general PSD matrix $A$ and 
relative approximation $\eps>0$ that is a sufficiently small constant,
every sketch $\sk{A}$ that satisfies the ``for all'' guarantee 
(with constant probability of success), 
even if all of entries of $A$ are promised to be in the range $\{-1, -1+1/n^C, -1+2/n^C, \ldots, 1-1/n^C, 1\}$ for a sufficiently large constant $C > 0$, 
must use $\Omega(n^2)$ bits of space.
\end{theorem}
\begin{proof}
%We follow the outline given in the introduction. 
Consider a net of $n \times n$ projection (thus PSD) matrices 
$P$ onto $n/2$-dimensional subspaces $U$ of $\mathbb{R}^n$. It is known 
(see Corollary 5.1 of \cite{kt13}, which uses a result of \cite{aek04}) that there
exists a family $\mathcal{F}$ of $r = 2^{\Omega(n^2)}$ distinct matrices 
$P_{1}, \ldots, P_{r}$ so that for all $i \neq j$, 
$\|P_i - P_j\|_2 \geq \frac{1}{2}$. By rounding each
of the entries of each matrix $P_i$ to the nearest additive 
multiple of $1/n^C$,
obtaining a symmetric matrix $Q_i$ with entries 
in $\{-1, -1+1/n^C, \ldots, 1\}$, 
we obtain a family $\mathcal{F}'$ of $r=2^{\Omega(n^2)}$ distinct matrices
$Q_i$ such that $\|Q_i - Q_j\|_2 \geq \frac{1}{4}$ for all $i \neq j$. This
implies there is a unit vector $x^*$ for which
$\|Q_ix^*-Q_jx^*\|_2 \geq \frac{1}{4}$, or equivalently, 
\begin{eqnarray}\label{eqn:optimize}
\|Q_ix^*\|_2^2 + \|Q_jx^*\|_2^2 - 2 \langle (x^*)^TQ_i^T, Q_j x^* \rangle \geq \frac{1}{16}.
\end{eqnarray}
In the rest of the proof, for simplicity, we often abuse the notation by directly using a matrix $M$ to denote the column space of $M$.
Let $J$ be the subspace of $\mathbb{R}^n$ which is the intersection of the spaces
spanned by the columns of $Q_i$ and $Q_j$, let $K_i$ be the subspace of $Q_i$ orthogonal to $J$,
and let $K_j$ be the subspace of $Q_j$ orthogonal to $J$. We identify $J$, $K_i$, and $K_j$,
with their corresponding projection matrices. To maximize the lefthand side of 
(\ref{eqn:optimize}), we can assume the unit vector $x^*$ is in the span of the union of columns of $J, K_i,
$ and $K_j$. We can therefore write $x^* = Jx^* + K_ix^* + K_j x^*$, and note that the three summand
vectors are orthogonal to each other. Expanding (\ref{eqn:optimize}),
the lefthand side is equal to 
$$2\|Jx^*\|_2^2 + \|K_ix^*\|_2^2 + \|K_j x^*\|_2^2
- 2 \|Jx^*\|_2^2 =  \|K_ix^*\|_2^2 + \|K_j x^*\|_2^2.$$
Hence, by (\ref{eqn:optimize}), it must be that either $\|K_ix^*\|_2^2 \geq \frac{1}{32}$ or
$\|K_j x^*\|_2^2 \geq \frac{1}{32}$. This implies the vector $z = K_ix^*$ satisfies
$\|Q_iz\|_2^2 \geq \frac{1}{32}$, but $\|Q_jz\|_2^2 = 0$. 

Therefore,
if there were a sketch
$\sk{A}$ which had the ``for all'' guarantee for any matrix $A \in \mathcal{F}$, 
one could query $\sk{A}$ on the vector $z$ given above for each pair $Q_i, Q_j \in \mathcal{F}$, 
thereby recovering the matrix $A \in \mathcal{F}$. Hence, $\sk{A}$ is an encoding of an arbitrary
element $A \in \mathcal{F}$ which implies that the size of $\sk{A}$ is 
$\Omega(\log |\mathcal{F}|) = \Omega(n^2)$ bits, completing the proof. 
\end{proof}

\subsection{``For Each'' Model}
\label{sec:psd-each}

In the ``for each'' case for PSD matrices $A$, we can use the Johnson-Lindenstrauss lemma to obtain a sketch $\sk{A}$ of size $O(n/\eps^2 \cdot \log(1/\delta) \log n)$. 
One instantiation of this lemma works by choosing 
a random $r \times n$ matrix $S$ with i.i.d.\ entries 
in $\{-1/\sqrt{r}, +1/\sqrt{r}\}$ for $r = \Theta(\eps^{-2} \log 1/\delta)$. 
Then for any $n \times n$ matrix $B$ and fixed $x$, we have
$\Pr[ \|SBx\|_2^2 \in (1\pm\eps)\|Bx\|_2^2] \ge 1-\delta$. 
The sketch simply computes a matrix $B$ where $B$ satisfies $B^T B = A$
(such a decomposition exists since $A$ is PSD), 
hence $x^T A x = x^T B^T B x = \norm{Bx}_2^2$.
Therefore, the sketch $\sk{A}$ can be $SB$, 
which is an $r \times n$ matrix of size $\tilde{O}(n/\eps^2)$ bits,
and the estimate it produces would be $\|SBx\|^2$.

For general PSD matrices $A$, this $\tilde O(n/\eps^2)$ upper bound 
turns out to be optimal up to logarithmic factors. 

%\subsection{Lower Bound for General Matrices with the For Each Guarantee}
\begin{theorem}
\label{thm:all-approx}
For a general PSD matrix $A$ and relative approximation $\eps \in (1/\sqrt{n}, 1)$,
every sketch $\sk{A}$ that satisfies the ``for each'' guarantee 
(with constant probability) must use $\Omega(n/\eps^2)$ bits of space.
\end{theorem}

The proof is similar (and inspired from) a result in \cite{GWWZ14} for approximating the number of non-zero entries of $Ax$. Before giving the proof, we need some tools. Let $\Delta(a,b)$ be the Hamming distance between two bitstrings $a$ and $b$.

\begin{lemma}[modified from \cite{JKS08}]
\label{lem:gap-ham}
Let $x$ be a random bitstring of length $\gamma = 1/\eps^2$, and let $i$ be a random index in $[\gamma]$. Choose $\gamma$ public random bitstrings $r^1, \ldots, r^{\gamma}$, each of length $\gamma$. Create $\gamma$-length bitstrings $a$, $b$ as follows:
\begin{itemize}
\item
For each $j \in [\gamma]$, $a_j = \text{majority}\{r^j_k\ |\ \text{indices } k \text{ for which } x_k = 1\}$.
\item
For each $j \in [\gamma]$, $b_j = r^j_i$.
\end{itemize}
There is a procedure which, with probability $1/2+\delta$ for a constant $\delta > 0$, can determine the value of $x_i$ from any $c \sqrt{\Delta(a,b)}$-additive approximation to $\Delta(a,b)$, provided $c > 0$ is a sufficiently small constant.
\end{lemma}

We introduce the \IND\ problem. In \IND, we have two randomized 
parties Alice and Bob. Alice has $x \in \{0, 1\}^n$, and Bob has an index $i \in [n]$. The communication is one-way from Alice to Bob, and the goal is for Bob to compute $x_i$.

\begin{lemma}[see, e.g., \cite{KN97}]
\label{lem:index}
To solve the \IND\ problem with success probability $1/2+\delta$ for any constant $\delta > 0$, Alice needs to send Bob $\Omega(n)$ bits  even with shared randomness.
\end{lemma}

\begin{proof} (for Theorem~\ref{thm:all-approx})
Let $\gamma = 1/\eps^2$. The proof is by a reduction from the \IND\ problem, where Alice has a random bitstring $z$ of length $(n - \gamma) \cdot \gamma$, and Bob has an index $\ell \in [(n - \gamma) \cdot \gamma]$. 

Partition $z$ into $n - \gamma$ contiguous substrings $z^1, z^2, \ldots, z^{n - \gamma}$. Alice constructs a $0/1$ matrix $B$ as follows: she uses shared randomness to sample $\gamma$ random bitstrings $r^1, \ldots, r^\gamma$, each of length $\gamma$. For the leftmost $\gamma \times \gamma$ submatrix of $B$, in the $i$-th column for each $i \in [\gamma]$, Alice uses $r^1, \ldots, r^\gamma$ and the value $i$ to create the $\gamma$-length bitstring $b$ according to Lemma~\ref{lem:gap-ham} and assigns it to this column. Next, in the remaining $n - \gamma$ columns of $B$, in the $j$-th column for each $j \in \{\gamma+1, \ldots, n\}$, Alice uses $z^{j - \gamma}$ and $r^1, \ldots, r^\gamma$ to create the  $\gamma$-length bitstring $a$ according to Lemma~\ref{lem:gap-ham} and assigns it to this column.

Our PSD matrix is set to be $A = B^T B$.

Alice then sends Bob the sketch $\sk{A}$, together with $\{\|B^i\|_2^2 \ (i \in [n])\}$ and $\{nnz(B^i) \ (i \in [n])\}$ where $nnz(x)$ is the number of non-zero coordinates of $x$. Note that both $\{\|B^i\|_2^2 \ (i \in [n])\}$ and $\{nnz(B^i) \ (i \in [n])\}$ can be conveyed using  $O(n \log (1/\eps)) = o(n/\eps^2)$ bits of communication, which is negligible.

Bob creates a vector $x$ by putting a $1$ in the $i$-th and $j$-th coordinates, where $i, j\ (i \in [\gamma], j \in \{\gamma + 1, \ldots,  n\})$ satisfies $\ell = i + (j - \gamma - 1) \cdot \gamma$. Then $Bx$ is simply the sum of the $i$-th and $j$-th columns of $B$, denoted by $B^i + B^j$. Note that $B^i, B^j$ correspond to a pair of $(a, b)$ created from $r^1, \ldots, r^\gamma$ and $z^{j - \gamma}$ (and $(z^{j - \gamma})_i = z_\ell$). Now, from a $(1 + c\eps)$-approximation to $x^T A x = \|B^i + B^j\|_2^2$ for a sufficiently small constant $c$, and exact values of $\|B^i\|_2^2$ and  $\|B^j\|_2^2$, Bob can approximate $2 \langle B^i, B^j \rangle= \|B^i + B^j\|_2^2 - \|B^i\|_2^2 - \|B^j\|_2^2$ up to an additive error $c \eps \cdot \|B^i + B^j\|_2^2 \le c' /\eps$ for a sufficiently small constant $c'$. Then, using a $(c' /\eps)$-additive approximation of $2 \langle B^i, B^j \rangle$, and exact values of $nnz(B^i)$ and $nnz(B^j)$, Bob can approximate $\Delta(B^i, B^j) = nnz(B^i) + nnz(B^j) - 2 \langle B^i, B^j \rangle$ up to a $c'/\eps = c'' \sqrt{\Delta(B^i, B^j)}$ additive approximation for a sufficiently small constant $c''$, and consequently compute $z_\ell$ correctly with probability $1/2+\delta$ for a constant $\delta > 0$ (by Lemma~\ref{lem:gap-ham}).  

Therefore, any algorithm that produces a $(1+c \eps)$-approximation of $x^T A x = \|B^i + B^j\|_2^2$ with probability $1 - \delta'$ for some sufficiently small constant $\delta' < \delta$ can be used to solve the Indexing problem of size $(n - \gamma) \gamma = \Omega(n/\eps^2)$ with probability $1 - \delta' - (1/2 - \delta) > 1/2 + \delta''$ for a constant $\delta'' > 0$. The theorem follows by the reduction and Lemma~\ref{lem:index}. 
\end{proof}

%%% Local Variables: 
%%% mode: latex
%%% TeX-master: "paper"
%%% End: 

%{\bf Acknowledgments:} David Woodruff would like to thank Alex Andoni, Jon Kelner, and Robi Krauthgamer for helpful discussions, as well as acknowledge support from the XDATA program of the Defense Advanced Research Project Agency (DARPA), administered through Air Force Research Laboratory contract FA8750-12-C-0323. Qin Zhang would like to thank the organizers of Bertinoro Workshop on Sublinear Algorithms 2014 where he heard this problem.

{\small
\bibliographystyle{alphaurlinit}
\bibliography{paper,robi}
}

\newpage

\appendix
\section{General Matrices, ``For Each'' Model}
\label{app:general-each}

\begin{theorem}
\label{thm:nonPSD-LB}
Any sketch $\sk{A}$ of a general $n \times n$ matrix $A$ that satisfies the ``for each''  guarantee with probability $0.9$, even when all entries of $A$ are promised to be in the set $\{0,1\}$, must use $\Omega(n^2)$ bits of space.
\end{theorem}

\begin{proof}
Let A be a symmetric matrix with zero on the diagonal, and a random bit in every other entry. Set the query vector $x = (e_i + e_j)$. Then using $\frac{1}{2} x^T A x$ we can recover the entry $A_{i,j}$, with probability $0.9$.  Think the sketching problem as a communication problem where Alice holds the matrix $A$; she sends a message (the sketch) $M$ to Bob such that Bob can recover each entries of $A$ with probability $0.9$ (except for the diagonal entries, which are fixed to be $0$, Bob can recover exactly). Then,
\begin{eqnarray*}
H(A\ |\ M) &=& \sum_{i,j \in [n]} H(A_{i,j}\ |\ M) \quad \text{($A_{i,j}$ are independent)} \\
&\le& (H_2(0.9) + 0.1) \cdot n^2 \quad  \text{(Fano's inequality)} \\
&<& 0.6 n^2.
\end{eqnarray*}
Thus $H(M) \ge H(A) - H(A\ |\ M) = \Omega(n^2)$.
\end{proof}

\section{Reduction from SDD Matrices to Laplacian Matrices}
\label{app:sdd-each}
In this section we show that the quadratic form of an SDD matrix, $x^TAx$, can be reduced to the quadratic form of a Laplacian, therefore our upper bounds for Laplacian matrices in Section~\ref{sec:cut-each} and Section~\ref{sec:spectral-each} can be extended to SDD matrices.

An SDD matrix $A$ has the property that $A_{i,i} \geq \sum_{j\neq i} |A_{i,j}|$ for all $i$. In the case when $A_{i,i} = \sum_{i\neq j}|A_{i,j}|$ for all $i$, we can write $A$ as $A_p + A_n +D$ where $D$ is the diagonal of $A$, $A_n$ is the matrix consisting of only the negative off-diagonal entries of $A$, and $A_p$ is the matrix consisting of only the positive off-diagonal entries of $A$. It is straightforward to verify that
\[
\left ( \begin{array}{cc} x^T & -x^T\end{array} \right ) \left ( \begin{array}{cc} D + A_n & -A_p \\ -A_p & D + A_n \end{array} \right ) \left ( \begin{array}{c} x \\ -x\end{array} \right ) = 2 x^T A x.
\]
The matrix $\left ( \begin{array}{cc} D + A_n & -A_p \\ -A_p & D + A_n \end{array} \right )$ is clearly a Laplacian matrix.

For the general case when $A_{i,i} \geq \sum_{i\neq j} |A_{i,j}|$. We can remove some ``weights'' from the diagonal entries of $A$, so that $A$ can be written as $A = D + B$ where $D$ is a diagonal matrix and $B$ satisfies the requirement $B_{i,i} = \sum_{i \neq j}|B_{i,j}|$ for all $i$.  We then have $x^TAx = x^TDx + x^TBx$. 
The matrix $D$ can be stored explicitly, and $x^TBx$ can be reduced to the quadratic form of a Laplacian matrix as discussed above.

\begin{theorem} \label{thm:sdd-sketch-for-each}
Given an $n\times n$ SDD matrix $A$, let $w_{\max} = \max_{i,j} |A_{i,j}|$ and $w_{\min} = \min_{i,j ~\text{with}~ A_{i,j}\neq 0} |A_{i,j}|$,
and assume ${w_{\max}}/{w_{\min}} = \text{poly}(n)$
We can then construct a sketch of $A$ that gives a $(1 + \eps,0.99)$-approximation to $x^TAx$ for any fixed $x \in \mathbb{R}^n$. 
The size of this sketch is $\tilde{O}(n/\eps^{{8}/{5}})$ bits.
\end{theorem}

\end{document}